\documentclass[11pt]{article}

\usepackage[letterpaper,margin=0.98in]{geometry}

\usepackage{etoolbox}
\newbool{doubleCol}

\usepackage{bbold}
\usepackage{amsmath}
\usepackage{amsthm}
\usepackage{amssymb}
\usepackage{pgffor}
\usepackage[inline]{enumitem}

\setlist[enumerate]{nosep,topsep=0.3em}
\setlist[enumerate,1]{label=(\roman*)}
\setlist[enumerate,2]{label=(\alph*)}

\setlist[itemize]{nosep,topsep=0.1em}

\usepackage{mathtools}

\usepackage{bm}
\usepackage{microtype}
\usepackage{xspace}
\usepackage{xfrac}

\usepackage[margin=10pt,font=small,labelfont=bf]{caption}

\usepackage{times}
\usepackage{array}
\usepackage[table]{xcolor}
\usepackage{multirow}
\usepackage{makecell}
\usepackage{booktabs}

\usepackage{changepage}

\usepackage{url}
\usepackage{hyperref}

\hypersetup{colorlinks, linkcolor=darkblue, citecolor=darkgreen, urlcolor=darkblue}

\usepackage{float}
\usepackage{graphicx}
\usepackage{subcaption}

\graphicspath{{graphics/}}

\usepackage{tikz}
\usepackage{tcolorbox}
\tcbuselibrary{skins}

\newtcolorbox{thmEnvBox}[2][]{
enhanced,
attach boxed title to top left={yshift=-3mm,xshift=3mm},
title=#2,
top=1em,
colback=white,
colbacktitle=white,
coltitle=black,
before skip=1.0em,
boxrule=1pt,
top=2mm,
#1
}

\usetikzlibrary{calc}
\usetikzlibrary{decorations.pathreplacing}
\usetikzlibrary{decorations.pathmorphing}
\usetikzlibrary{shapes.geometric}
\usetikzlibrary{fadings,decorations,fit}

\pgfdeclarelayer{bg}
\pgfdeclarelayer{fg}
\pgfsetlayers{bg,main,fg}

\tikzstyle{ns}+=[thick,draw=black,fill=white,circle,minimum size=1.5mm, inner sep=0pt]

\definecolor{darkblue}{rgb}{0,0,0.38}
\definecolor{darkred}{rgb}{0.6,0,0}
\definecolor{darkgreen}{rgb}{0.1,0.35,0}
\colorlet{cutcol}{white!50!black}

\usepackage{wrapfig}

\usepackage[vlined,ruled,linesnumbered,algo2e]{algorithm2e}
\usepackage{mdframed}

\usepackage[textsize=scriptsize, color=green!50!brown]{todonotes}

\usepackage{soul}

\makeatletter
\newcommand{\labeltarget}[1]{\Hy@raisedlink{\hypertarget{#1}{}}}
\makeatother

\newtheorem{theorem}{Theorem}
\newtheorem{lemma}[theorem]{Lemma}

\newtheorem{proposition}[theorem]{Proposition}
\newtheorem{definition}[theorem]{Definition}

\newtheorem{corollary}[theorem]{Corollary}

\newtheorem{example}[theorem]{Example}

\newcommand\OPT{\ensuremath{\mathsf{OPT}}\xspace}

\newcommand\expval{\mathbb{E}}
\newcommand\Prb{\mathbb{Pr}}

\DeclareMathOperator{\conv}{conv}

\DeclareMathOperator{\supp}{supp}

\makeatletter
\newcommand{\manuallabel}[1]{\def\@currentlabel{#1}\label{#1}}
\makeatother

 \makeatletter
 \def\@cite#1#2{\textup{[{#1\if@tempswa , #2\fi}]}}
 \makeatother

\newcommand{\bipmatch}[1][G]{\mathcal{M}(#1)}
\newcommand{\match}[1][G]{\mathcal{M}(#1)}
\newcommand{\bipmatchP}[1][G]{P_{\bipmatch[#1]}}
\newcommand{\matchP}[1][G]{P_{\match[#1]}}
\newcommand{\degreeP}[1][G]{P^{#1}_{\mathrm{deg}}}

\newcommand{\PoisD}[2][]{\mathsf{Pois}_{#1}( #2 )}
\newcommand{\ExpD}[2][]{\mathsf{Exp}_{#1}( #2 )}
\newcommand{\BernD}[2][]{\mathsf{Bern}_{#1}( #2 )}
\newcommand{\BinD}[3][]{\mathsf{Bin}_{#1}( #2 , #3 )}

\newcommand{\myexp}{\mathsf{e}}

\renewcommand{\epsilon}{\varepsilon}

\newcommand{\dcup}{\mathbin{\dot{\cup}}}

\newcommand{\mlinext}{F_{\mathrm{ML}}}

\newcommand{\colorofx}{blue}
\newcommand{\colorofy}{violet}
\newcommand{\colorofset}{green!50!black}

\title{An Optimal Monotone Contention Resolution Scheme for Bipartite Matchings via a Polyhedral Viewpoint}

\author{%
Simon Bruggmann\thanks{%
Department of Mathematics, ETH Zurich, Zurich, Switzerland.
Email: \href{mailto:simon.bruggmann@ifor.math.ethz.ch}%
{simon.bruggmann@ifor.math.ethz.ch}.
Supported by Swiss National Science Foundation grant 200021\_165866.
}
\and 
Rico Zenklusen\thanks{%
Department of Mathematics, ETH Zurich, Zurich, Switzerland.
Email: \href{mailto:ricoz@math.ethz.ch}%
{ricoz@math.ethz.ch}.
Supported by Swiss National Science Foundation grant 200021\_165866.
}
}

\date{}
 
\begin{document}

\maketitle

\begin{abstract}
Relaxation and rounding approaches became a standard and extremely versatile tool for constrained submodular function maximization. One of the most common rounding techniques in this context are contention resolution schemes. Such schemes round a fractional point by first rounding each coordinate independently, and then dropping some elements to reach a feasible set. Also the second step, where elements are dropped, is typically randomized. This leads to an additional source of randomization within the procedure, which can complicate the analysis. 
We suggest a different, polyhedral viewpoint to design contention resolution schemes, which avoids to deal explicitly with the randomization in the second step. This is achieved by focusing on the marginals of a dropping procedure. Apart from avoiding one source of randomization, our viewpoint allows for employing polyhedral techniques. Both can significantly simplify the construction and analysis of contention resolution schemes.

We show how, through our framework, one can obtain an optimal monotone contention resolution scheme for bipartite matchings. So far, only very few results are known about optimality of monotone contention resolution schemes. Our contention resolution scheme for the bipartite case also improves the lower bound on the correlation gap for bipartite matchings. 

Furthermore, we derive a monotone contention resolution scheme for matchings that significantly improves over the previously best one. 
At the same time, our scheme implies that the currently best lower bound on the correlation gap for matchings is not tight.

Our results lead to improved approximation factors for various constrained submodular function maximization problems over a combination of matching constraints with further constraints.
 \end{abstract}

\thispagestyle{empty}
\addtocounter{page}{-1}

\newpage

\section{Introduction}

Submodular function maximization problems enjoyed a surge of interest recently, both within the theory community and in application-focused areas. This is due to a wide set of applications in fields as diverse as combinatorial optimization, economics, algorithmic game theory, or machine learning~(see, e.g.,~\cite{LehmannLehmannNisan2006,BalcanBlumMansour2008,HartlineMirrokniSundararajan2008,CalinescuChekuriPalVondrak2011,MirzasoleimanKarbasiSarkarKrause2013,WeiIyerBilmes2014} and references therein).
The breadth of settings where submodular functions are found is not surprising in view of the fact that submodularity formalizes a very natural property, namely the property of diminishing marginal returns. More formally, given a finite ground set $E$, a function $f\colon 2^E\to \mathbb{R}_{\geq 0}$ defined on all subsets %
 of $E$ is \emph{submodular} if
\begin{equation*}
f(S\cup \{e\}) - f(S) \geq  f(T\cup \{e\}) - f(T)  \quad \forall S\subseteq T \subseteq E \text{ and } e\in E\setminus T\enspace.
\end{equation*}
In words, the marginal change in $f$ when adding an element $e$ to a set $S$ is the higher the smaller the set $S$ is. Equivalently, submodularity can also be defined by the following uncrossing property
\begin{equation*}
f(S) + f(T) \geq f(S\cup T) + f(S\cap T) \quad \forall S\subseteq T\subseteq E\enspace.
\end{equation*}
Some well-known examples of submodular functions include cut functions of graphs (undirected or directed), coverage functions, joint entropy, and rank functions of matroids.

Whereas there are some interesting unconstrained submodular maximization problems, as for example the $\mathsf{MaximumCut}$ problem, one often wants to maximize a submodular function only over some subfamily $\mathcal{F}\subseteq 2^E$ of feasible sets. This leads to the problem of constrained submodular function maximization~\ref{eq:CSFM} 
\begin{equation}\label{eq:CSFM}
\max_{S\in \mathcal{F}} f(S)\enspace.\tag{\ensuremath{\mathsf{CSFM}}}
\end{equation}
The family $\mathcal{F}$ is referred to as the family of \emph{feasible subsets} of $E$. For instance, if $E$ is the set of edges of a graph $G = (V,E)$, we could be interested in the family $\mathcal{F} \subseteq 2^E$ of all forests or all matchings in $G$.
Moreover, we assume the classical value oracle model, where the submodular function $f$ is given through a value oracle that, for any $S\subseteq E$, returns the value $f(S)$.

Submodular maximization is well-known to be hard to solve exactly, as even $\mathsf{MaximumCut}$---which is an unconstrained submodular maximization problem with an explicitly given submodular function---is $\mathsf{APX}$-hard. Moreover, the value oracle model leads to information-theoretic hardnesses, which do not rely on complexity-theoretic assumptions like $\mathsf{P} \neq \mathsf{NP}$. For example, Feige, Mirrokni, and Vondr\'ak~\cite{FeigeMirrokniVondrak2007,feige_2011_maximizing} showed that, in the value oracle model, without an exponential number of value queries it is impossible to obtain an approximation factor better than $\sfrac{1}{2}$ for unconstrained submodular function maximization. %
Clearly, in these hardness results, submodular functions $f$ are considered that are not monotone, because otherwise the ground set $E$ is trivially an optimal solution to the unconstrained %
problem.\footnote{%
 	A set function $f \colon 2^E \to \mathbb{R}_{\geq 0}$ is said to be \emph{monotone} if  $f(S) \leq f(T)$ whenever $S \subseteq T \subseteq E$.%
 } 
However, the problem remains hard to approximate beyond constant factors even for monotone submodular functions as soon as constraints are involved. 
Indeed, already for a single cardinality constraint, i.e., $\mathcal{F}\coloneqq \{S\subseteq E \mid |S| \leq k\}$, Nemhauser and Wolsey~\cite{NemhauserWolsey1978} proved that, for any $\epsilon >0 $, an exponential number of oracle calls is needed to obtain a $(1 - \sfrac{1}{\myexp} + \epsilon)$-approximation for~\ref{eq:CSFM}.\footnote{%
  	In this work, the symbol $e$ usually denotes an element of the ground set $E$. In order to avoid confusion, we therefore use the symbol $\myexp$ for Euler's number $\exp(1)=2.71828\ldots$.%
  }
Therefore, theoretical research in the area of~\ref{eq:CSFM} is mostly centered around the development of approximation algorithms, with (strong) constant-factor approximations being the gold standard.

The focus in~\ref{eq:CSFM} lies on down-closed feasibility families~$\mathcal{F}$.\footnote{%
	A family $\mathcal{F} \subseteq 2^E$ is called \emph{down-closed} if $T \in \mathcal{F}$ and $S \subseteq T$ imply $S \in \mathcal{F}$.%
}
If the function $f$ is monotone, this is without loss of generality. Indeed, we can simply optimize over the down-closure $\{S \subseteq E \mid S\subseteq T \text{ for some } T\in \mathcal{F}\}$ of $\mathcal{F}$, and then raise the found set to a feasible set by potentially adding additional elements, which can only increase the function value due to monotonicity.
 For non-monotone $f$, however, the problem quickly gets very hard to approximate. More precisely, Vondr\'ak~\cite{Vondrak2013} showed that exponentially many value oracle queries are needed to get a constant-factor approximation for non-monotone~\ref{eq:CSFM} already over the set of bases of a matroid.
Due to this, we assume throughout this work that $\mathcal{F}$ is down-closed.

\medskip

Most current approaches for~\ref{eq:CSFM} fall into one of three main classes:
\begin{enumerate}
\item greedy procedures~\cite{NemhauserWolseyFisher1978,FisherNemhauserWolsey1978,NemhauserWolsey1978, ConfortiCornuejols1984, Sviridenko2004, GuptaRothSchoenebeckTalwar2010, FeldmanHarshawKarbasi2017},
\item local search algorithms~\cite{LeeMirrokniNagarajanSviridenko2009,lee_2010_maximizing,LeeSviridenkoVondrak2010,feige_2011_maximizing,FeldmanNaorSchwartzWard2011,Ward2012, BruggmannZenklusen2017},
\item and relaxation and rounding approaches~\cite{AgeevSviridenko2004,Vondrak2008, chekuri_2009_dependent,ChekuriVondrakZenklusen2010, FeldmanNaorSchwartz2011a, FeldmanNaorSchwartz2011b, CalinescuChekuriPalVondrak2011, kulik_2013_approximations, ChekuriVondrakZenklusen2011, chekuri_2014_submodular, ene_2016_constrained, buchbinder_2016_constrained}.
\end{enumerate}

Here we are interested in relaxation and rounding approaches for~\ref{eq:CSFM}, two key advantages of which are the following: First, they are very versatile in terms of constraints they can handle, and, second, they allow for breaking the problem down into clean subproblems.  
Relaxation and rounding approaches work with a relaxation $P$ of the family of feasible sets $\mathcal{F}$.
We call a polytope $P$ a \emph{relaxation} of $\mathcal{F} \subseteq 2^E$ if $P \subseteq [0,1]^E$ and $P$ has the same integral points as the combinatorial polytope $P_{\mathcal{F}} \coloneqq \conv(\{\chi^S \mid S\in \mathcal{F}\})\subseteq [0,1]^E$ corresponding to $\mathcal{F}$,\footnote{For any set $S\subseteq E$, we denote by $\chi^S\in \{0,1\}^E$ the characteristic vector of $S$.} i.e., $P\cap \{0,1\}^E = P_{\mathcal{F}}\cap \{0,1\}^E$ (note that this implies $P_{\mathcal{F}}\subseteq P$).
Moreover, relaxation and rounding approaches require an extension $F \colon [0,1]^E \to\mathbb{R}_{\geq 0}$ of $f$. Having fixed a relaxation $P$ of $\mathcal{F}$ and an extension $F$ of $f$, the following two steps comprise a typical relaxation and rounding approach:

\begin{thmEnvBox}[width=15cm,center]{}
\textbf{Relaxation and rounding for~$\bm{\mathsf{CSFM}}$} (with relaxation $P$ of $\mathcal{F}$ and extension $F$ of $f$)
\begin{enumerate}[label=(\arabic*)]
\item\label{item:maxExt} Approximately maximize $F$ over $P$, i.e., $\max\{F(z) \mid z\in P\}$, to obtain $x\in P$.

\item\label{item:roundPoint} Round $x$ to a point $\chi^S \in P$ and return $S$.
\end{enumerate}
\end{thmEnvBox}

The by far most successful extension $F$ of $f$ used in~\ref{eq:CSFM} is the so-called \emph{multilinear extension} $\mlinext$. For a point $x \in [0,1]^E$, it is defined by the expected value $\mlinext(x) \coloneqq \expval [ f(R(x))]$, where $R(x)$ is a random subset of $E$ containing each element $e \in E$ independently with probability $x_e$.
The success of the multilinear extension is based on the fact that it has rich structural properties that can be exploited in both steps of a relaxation and rounding approach.

Very general results are known for step~\ref{item:maxExt} when $F=\mlinext$.
More precisely, Vondr\'ak~\cite{Vondrak2008} and C{\u{a}}linescu, Chekuri, P\'al, and Vondr\'ak~\cite{CalinescuChekuriPalVondrak2011} proved the following result that holds for the extension $\mlinext$ of any monotone submodular function $f \colon 2^E \to \mathbb{R}_{\geq 0}$ and any down-closed polytope $P \subseteq [0,1]^E$ that is \emph{solvable}, which means that any linear function can be optimized efficiently over $P$.
For any $\epsilon >0$ and $b\in [0,1]$, one can efficiently compute a point $x\in b P$ with $\mlinext(x) \geq (1-\myexp^{-b} - \epsilon)\cdot f(\OPT)$, where $\OPT$ denotes an optimal solution to our~\ref{eq:CSFM} problem.
As shown by Mirrokni, Schapira, and Vondr\'ak~\cite{MirrokniSchapiraVondrak2008}, this result is essentially optimal as it cannot be improved by any constant factor.

Analogous results hold also for non-monotone $f$, with a weaker constant approximation guarantee, as first shown by Chekuri, Vondr\'ak, and Zenklusen~\cite{chekuri_2014_submodular}. Better constants were later obtained by Feldman, Naor, and Schwartz~\cite{FeldmanNaorSchwartz2011b}, Ene and Nguy$\tilde{\hat{\mathrm{e}}}$n~\cite{ene_2016_constrained}, and Buchbinder and Feldman~\cite{buchbinder_2016_constrained}. 
More precisely, Feldman, Naor, and Schwartz~\cite{FeldmanNaorSchwartz2011b} provide an elegant algorithm that, for any $b\in [0,1]$ and $\epsilon >0$, returns a point $x \in b P$ with $\mlinext(x) \geq (b \cdot \myexp^{-b} - \epsilon) \cdot f(\OPT)$. For $b = 1$, the factor was increased to $0.372$ in~\cite{ene_2016_constrained}, and subsequently to $0.385$ in~\cite{buchbinder_2016_constrained}.\footnote{%
On the inapproximability side, Oveis Gharan and Vondr\'ak~\cite{OveisGharanVondrak2011} proved an information-theoretic hardness in the value oracle model of $0.478$ for $b = 1$.
}
The ideas in~\cite{ene_2016_constrained} and~\cite{buchbinder_2016_constrained} also seem to work for $b < 1$. It is not immediate to see, however, what the corresponding approximation guarantees would be in that case.
Obtaining a point $x$ in a down-scaled version $b P$ of $P$ can often be interesting because such points are typically easier to round in step~\ref{item:roundPoint}.

In summary, constant-factor approximations for step~\ref{item:maxExt} with $F = \mlinext$ (with respect to $\OPT$) can be obtained for any down-closed solvable polytope $P$.

For step~\ref{item:roundPoint}, i.e., the rounding of a fractional point, different techniques have been developed. For example, when $\mathcal{F}$ is the family of independent sets of a matroid and $P=P_{\mathcal{F}}$, then lossless rounding with respect to the multilinear extension is possible via pipage rounding~\cite{AgeevSviridenko2004} or swap rounding~\cite{ChekuriVondrakZenklusen2010}.
However, the currently most versatile rounding technique is based on so-called monotone contention resolution schemes (CR schemes), which were introduced in~\cite{chekuri_2014_submodular}. They also have applications beyond~\ref{eq:CSFM} (see, e.g.,~\cite{GuptaNagarajan2013, Adamczyk2015, FeldmanSvenssonZenklusen2016}).

The goal of this paper is to introduce a new polyhedral viewpoint on CR schemes that avoids some of its complications. Through this new point of view, we design an optimal monotone CR scheme for the bipartite matching polytope, and also beat the previously strongest monotone CR scheme for the general matching polytope. Before expanding on our contributions, we first present a brief introduction to CR schemes in the next subsection, which also allows us later to clearly highlight the benefits of our new viewpoint.

\subsection{Contention resolution schemes}\label{sec:CRschemes}

The key goal when rounding a point $x\in P$ to a feasible set $S\in \mathcal{F}$ is not to lose too much in terms of objective value, i.e., we would like that $f(S)$ is not much less than $\mlinext(x)$.
In view of the way the multilinear extension $\mlinext$ is defined, the arguably most natural strategy is to randomly round $x$ to a set $R(x)\subseteq E$ that includes each element $e \in E$ independently with probability $x_e$. By definition of the multilinear extension $\mlinext$, it immediately follows that $\expval[ f(R(x))] = \mlinext(x)$, meaning that, at least in expectation, the rounding is lossless. The resulting set $R(x)$, however, will typically not be feasible. 
Contention resolutions schemes address this issue by carefully dropping elements from $R(x)$ to achieve feasibility.

\begin{definition}\label{def:CRScheme}
Let $b,c \in [0,1]$. A \emph{$(b,c)$-balanced contention resolution scheme} $\pi$ for a relaxation $P$ of $\mathcal{F}$ is a procedure that, for every $x \in b P$ and $A \subseteq \supp(x)$, returns a (possibly random) set $\pi_x(A) \subseteq A$ with $\pi_x(A) \in \mathcal{F}$ and
\begin{equation}\label{eq:CRbalancedness}
\Prb [e \in \pi_x(R(x)) \mid e \in R(x) ]
= \frac{1}{x_e} \cdot \Prb[e\in \pi_x(R(x))]
\geq c \qquad \forall e\in \supp(x)\enspace .
\end{equation}
The scheme is \emph{monotone} if, for all $x \in bP$, it holds $\Prb[ e \in \pi_x(A) ] \geq \Prb[ e \in \pi_x(B) ]$ whenever $e \in A \subseteq B \subseteq \supp(x)$.
Moreover, a $(1,c)$-balanced CR scheme is also called a \emph{$c$-balanced CR scheme}.
\end{definition}

Notice that the equality in~\eqref{eq:CRbalancedness} always holds because $e\in \pi_x(R(x))$ can only happen if $e\in R(x)$, and, moreover, $\Pr[e\in R(x)]=x_e$.
Below, we exemplify the notion of contention resolution schemes by considering a matching constraint.

\begin{example}\label{ex:CRschemeMatch}
We describe a CR scheme for the relaxation $\degreeP[G] \coloneqq \big\{ x \in \mathbb{R}^E_{\geq 0}  \mid  x(\delta(v)) \leq 1 \;\;\forall v \in V     \big\}$ of the family of matchings in a simple graph $G=(V,E)$.
The box below describes the CR scheme (see Figure~\ref{fig:exampleCRS} for an illustration).
\begin{thmEnvBox}[width=15cm,center]{}
	Given $x \in \degreeP[G]$ and $A \subseteq \supp(x)$, we construct a matching $M^A \subseteq A$ in $G$ as follows:
	\begin{enumerate}[label=\normalfont(\roman*)]
		\item\label{enum:exampleCR-subsampling} Select each $e \in A$ independently with probability $\frac{1}{2}$. Denote the resulting set by $\overline{A}$.
		\item\label{enum:exampleCR-isolated} Return all isolated edges $M^A \subseteq \overline{A}$ in the graph $G'=(V, \overline{A})$.
	\end{enumerate}
\end{thmEnvBox}
We recall that an edge is isolated in a graph if it is not touched by any other edge in the graph, i.e., it does not have a common endpoint with another edge.
Because we only return isolated edges in some subgraph of $G$, the CR scheme indeed returns a matching.

Note that for an edge $e \in A$ to be included in the resulting matching $M^A$, all its adjacent edges in $A$ must be discarded in step~\ref{enum:exampleCR-subsampling} while $e$ needs to be selected.
The probability of the latter is $\frac{1}{2}$ whenever $e\in A$. Moreover, since the probability that all edges in $A$ neighboring $e$ get discarded in step~\ref{enum:exampleCR-subsampling} gets smaller the more such edges there are, the scheme is monotone.

To determine the balancedness, let $x \in \degreeP[G]$ and fix an edge $e = \{u,v\} \in \supp(x)$. 
Recall that $x$ is first rounded independently to obtain $A \coloneqq R(x)$, and then, in step~\ref{enum:exampleCR-subsampling}, each edge in $A$ is selected independently with probability $\frac{1}{2}$ to obtain $\overline{A}$. Hence, the distribution of $\overline{A}$ is the same as $R(\sfrac{x}{2})$. 
Using a union bound, the probability that $R(\sfrac{x}{2}) \setminus \{e\}$ contains an edge incident to $u$ is at most 
\begin{equation*}
\sum_{g \in \delta(u) \setminus \{e\}} \frac{1}{2} \cdot x_{g}
\leq \frac{1}{2} \cdot x(\delta(u)) \leq \frac{1}{2}
\enspace ,
\end{equation*}
where we used that $x \in \degreeP[G]$.	
Clearly, the same bound also holds for the probability that an edge in $R(\sfrac{x}{2}) \setminus \{e\}$ is incident to $v$. Moreover, these two events are independent because $e$ is the only edge with endpoints $u$ and $v$ (because $G$ is simple), implying that with probability at least $(1-\frac{1}{2})^2 = \frac{1}{4}$, no edge in $R(\sfrac{x}{2}) \setminus \{e\}$ has $u$ or $v$ as an endpoint.
When the latter happens and $e$ itself is contained in $R(\sfrac{x}{2})$, $e$ is included in the final matching $M^A = M^{R(x)}$. We therefore get that
\begin{equation*}
\Pr[ e \in M^{R(x)}] \geq \frac{1}{4} \cdot
\Pr[e \in R(\sfrac{x}{2})] 
= \frac{1}{8} \cdot x_e
\enspace ,
\end{equation*}
which shows that the scheme is $\frac{1}{8}$-balanced.

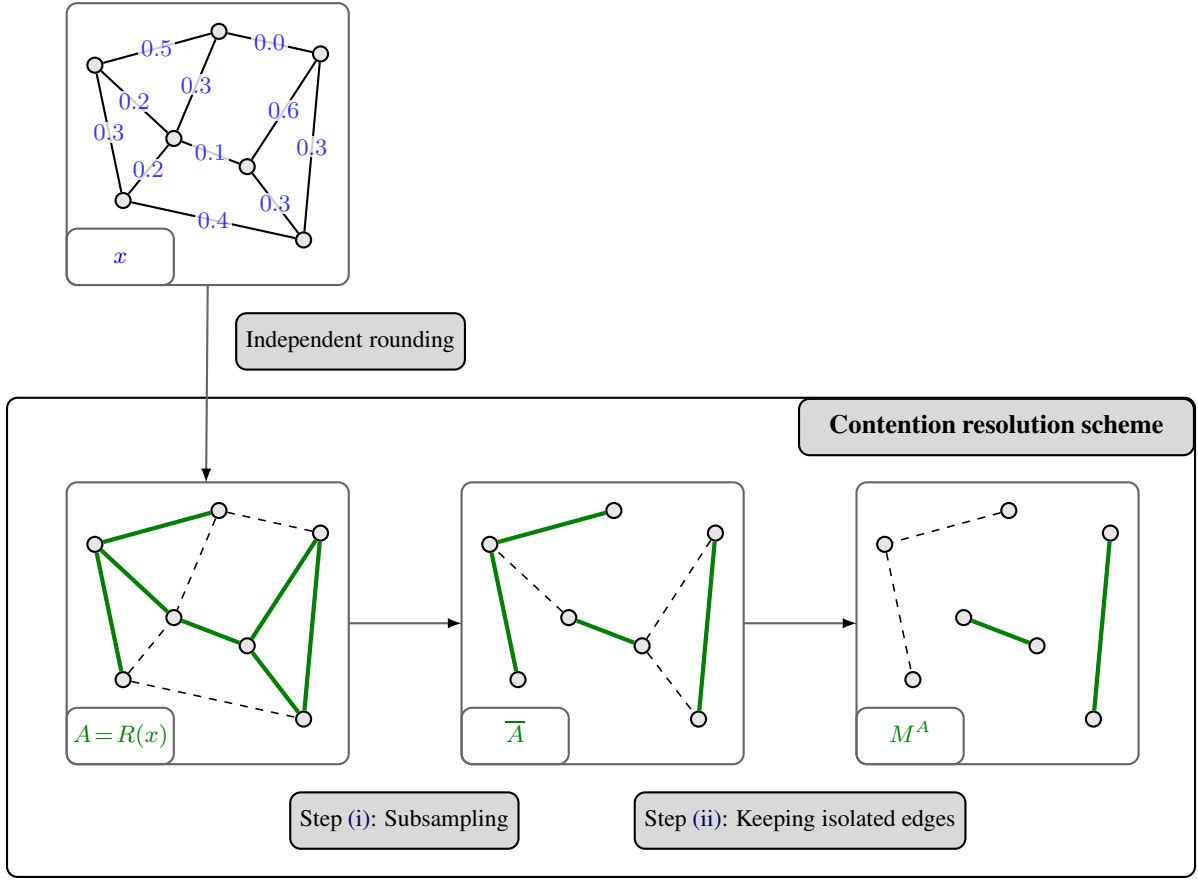
\begin{figure}[ht]
	\centering

\tikzstyle{vertex}=[circle, draw=black, thick, inner sep=1pt, minimum size=2mm, fill=black!10]
\tikzstyle{edgenode}=[rectangle, fill=white, fill opacity=0.8, text opacity=1, inner sep=1pt]
\tikzstyle{edge}=[thick, draw=black]
\tikzstyle{selectedEdge}=[ultra thick, draw=\colorofset]
\tikzstyle{discardedEdge}=[semithick, dashed, draw=black]
\tikzstyle{arc}=[thick, draw=black!60, -latex]
\tikzstyle{boxed}=[thick, draw=black!60, rounded corners]
\tikzstyle{boxedFat}=[thick, draw=black, rounded corners]
\tikzstyle{boxedLight}=[thick, draw=black, rounded corners, fill = black!15, minimum height=0.75*1cm]

\begin{tikzpicture}[scale=0.75]

\begin{scope}[shift={(0,1.5)}, local bounding box=random0]

\draw[boxed] (0,0) rectangle (5,5);
\draw[boxed] (0,0) rectangle (1.9,1);
\node[font=\footnotesize, color=\colorofx] at (0.95,0.5) {$\strut x$};

\node[vertex] (1) at (0.5,3.9) {};
\node[vertex] (2) at (2.7,4.5) {};
\node[vertex] (3) at (4.5,4.1) {};
\node[vertex] (4) at (3.2,2.1) {};
\node[vertex] (5) at (1.0,1.5) {};
\node[vertex] (6) at (1.9,2.6) {};
\node[vertex] (7) at (4.2,0.8) {};

\draw[edge] (6) -- node[font=\footnotesize, color=\colorofx, edgenode] {$0.2$} (1);
\draw[edge] (6) -- node[font=\footnotesize, color=\colorofx, edgenode] {$0.1$} (4);
\draw[edge] (1) -- node[font=\footnotesize, color=\colorofx, edgenode] {$0.3$} (5);
\draw[edge] (2) -- node[font=\footnotesize, color=\colorofx, edgenode] {$0.5$} (1);
\draw[edge] (5) -- node[font=\footnotesize, color=\colorofx, edgenode] {$0.2$} (6);
\draw[edge] (2) -- node[font=\footnotesize, color=\colorofx, edgenode] {$0.3$} (6);
\draw[edge] (3) -- node[font=\footnotesize, color=\colorofx, edgenode] {$0.0$} (2);
\draw[edge] (3) -- node[font=\footnotesize, color=\colorofx, edgenode] {$0.6$} (4);
\draw[edge] (5) -- node[font=\footnotesize, color=\colorofx, edgenode] {$0.4$} (7);
\draw[edge] (4) -- node[font=\footnotesize, color=\colorofx, edgenode] {$0.3$} (7);
\draw[edge] (3) -- node[font=\footnotesize, color=\colorofx, edgenode] {$0.3$} (7);

\end{scope}

\begin{scope}[shift={(0,-7)}, local bounding box=random1]

\draw[boxed] (0,0) rectangle (5,5);
\draw[boxed] (0,0) rectangle (1.9,1);
\node[font=\footnotesize, color=\colorofset] at (0.95,0.5) {$\strut A\!=\!R(x)$};

\node[vertex] (1) at (0.5,3.9) {};
\node[vertex] (2) at (2.7,4.5) {};
\node[vertex] (3) at (4.5,4.1) {};
\node[vertex] (4) at (3.2,2.1) {};
\node[vertex] (5) at (1.0,1.5) {};
\node[vertex] (6) at (1.9,2.6) {};
\node[vertex] (7) at (4.2,0.8) {};

\draw[selectedEdge] (6) --  (1);
\draw[selectedEdge] (6) -- (4);
\draw[selectedEdge] (1) -- (5);
\draw[selectedEdge] (2) -- (1);
\draw[discardedEdge] (5) -- (6);
\draw[discardedEdge] (2) -- (6);
\draw[discardedEdge] (3) -- (2);
\draw[selectedEdge] (3) -- (4);
\draw[discardedEdge] (5) -- (7);
\draw[selectedEdge] (4) -- (7);
\draw[selectedEdge] (3) -- (7);

\end{scope}

\begin{scope}[shift={(7,-7)}, local bounding box=random2]

\draw[boxed] (0,0) rectangle (5,5);
\draw[boxed] (0,0) rectangle (1.9,1);
\node[font=\footnotesize, color=\colorofset] at (0.95,0.5) {$\strut \overline{A}$};

\node[vertex] (1) at (0.5,3.9) {};
\node[vertex] (2) at (2.7,4.5) {};
\node[vertex] (3) at (4.5,4.1) {};
\node[vertex] (4) at (3.2,2.1) {};
\node[vertex] (5) at (1.0,1.5) {};
\node[vertex] (6) at (1.9,2.6) {};
\node[vertex] (7) at (4.2,0.8) {};

\draw[discardedEdge] (6) --  (1);
\draw[selectedEdge] (6) -- (4);
\draw[selectedEdge] (1) -- (5);
\draw[selectedEdge] (2) -- (1);
\path[] (5) -- (6);
\path[] (2) -- (6);
\path[] (3) -- (2);
\draw[discardedEdge] (3) -- (4);
\path[] (5) -- (7);
\draw[discardedEdge] (4) -- (7);
\draw[selectedEdge] (3) -- (7);

\end{scope}

\begin{scope}[shift={(14,-7)}, local bounding box=random3]

\draw[boxed] (0,0) rectangle (5,5);
\draw[boxed] (0,0) rectangle (1.9,1);
\node[font=\footnotesize, color=\colorofset] at (0.95,0.5) {$\strut M^A$};

\node[vertex] (1) at (0.5,3.9) {};
\node[vertex] (2) at (2.7,4.5) {};
\node[vertex] (3) at (4.5,4.1) {};
\node[vertex] (4) at (3.2,2.1) {};
\node[vertex] (5) at (1.0,1.5) {};
\node[vertex] (6) at (1.9,2.6) {};
\node[vertex] (7) at (4.2,0.8) {};

\path[] (6) --  (1);
\draw[selectedEdge] (6) -- (4);
\draw[discardedEdge] (1) -- (5);
\draw[discardedEdge] (2) -- (1);
\path[] (5) -- (6);
\path[] (2) -- (6);
\path[] (3) -- (2);
\path[] (3) -- (4);
\path[] (5) -- (7);
\path[] (4) -- (7);
\draw[selectedEdge] (3) -- (7);

\end{scope}

\draw[boxedFat] ($(random1.north west) + (-1,1.5)$) rectangle ($(random3.south east) + (1,-2)$);

\node[boxedLight, font=\small, minimum width=0.75*7cm, anchor=north east, inner sep = 0pt] at ($(random3.north east)+(1,1.5)$) {\textbf{Contention resolution scheme\strut}};

\draw[arc] (random0.south) -- (random1.north);
\draw[arc] (random1.east) -- (random2.west);
\draw[arc] (random2.east) -- (random3.west);

\node[boxedLight, font=\footnotesize, anchor=west] at ($(random0.south)!0.5!(random1.north) + (0.5,0.75)$) {Independent rounding\strut};
\node[boxedLight, font=\footnotesize] at ($(random1.south)!0.5!(random2.north) + (0,-3.5)$) {Step~\ref{enum:exampleCR-subsampling}: Subsampling\strut};
\node[boxedLight, font=\footnotesize] at ($(random2.south)!0.5!(random3.north) + (0,-3.5)$) {Step~\ref{enum:exampleCR-isolated}: Keeping isolated edges\strut};

\end{tikzpicture}  	\caption{A contention resolution scheme for the relaxation $\degreeP[G] = \big\{ x \in \mathbb{R}^E_{\geq 0}  \mid  x(\delta(v)) \leq 1 \;\; \forall v \in V \big\}$ of the family of matchings in a simple graph $G=(V,E)$.}\label{fig:exampleCRS}
\end{figure}

\end{example}

The key property of monotone CR schemes is that their balancedness translates into an approximation factor for the rounding procedure, as highlighted by the following theorem from~\cite{chekuri_2014_submodular}.
\begin{theorem}[\cite{chekuri_2014_submodular}]\label{thm:crSchemeGuarantee}
Let $\pi$ be a monotone $(b,c)$-balanced CR scheme for a relaxation $P \subseteq [0,1]^E$ of $\mathcal{F} \subseteq 2^E$ and let $x \in b P$. It then holds
\begin{equation*}
	\expval [ f(\pi_x(R(x)))  ] \geq c \cdot \mlinext(x)   \enspace ,
\end{equation*}
where $\mlinext$ is the multilinear extension of a montone submodular function $f \colon 2^E \to \mathbb{R}_{\geq 0}$.
Moreover, if $f$ is non-monotone, then there is an efficient procedure that takes $\pi_x(R(x))$ and returns a set $I_x \subseteq \pi_x(R(x))$ with 
\begin{equation*}
\expval [ f(I_x)  ] \geq c \cdot \mlinext(x)    \enspace .
\end{equation*}
\end{theorem}
Note that in the theorem above, we have $I_x\in \mathcal{F}$ because $I_x\subseteq \pi_x(R(x))\in \mathcal{F}$ and we assume that $\mathcal{F}$ is down-closed.

Hence, if we have a solvable relaxation $P$ of $\mathcal{F}$ and a monotone $(b,c)$-balanced CR scheme $\pi$ for $P$, then step~\ref{item:maxExt} and~\ref{item:roundPoint} of the relaxation and rounding approach can be realized as follows. %
We first compute a point $x\in bP$ with
$\mlinext(x) \geq \alpha(b) \cdot f(\OPT)$,
where $\alpha(b)$ is a constant factor depending on $b$ and on whether $f$ is monotone or not (see prior discussion).
We then use the monotone $(b,c)$-balanced contention resolution scheme $\pi$ with respect to $x$ to obtain a random set $I_x \in \mathcal{F}$ as highlighted in Theorem~\ref{thm:crSchemeGuarantee} with expected submodular value
\begin{equation*}
	\expval[ f(I_x) ] \geq c \cdot \mlinext(x) \geq c \cdot \alpha(b) \cdot f(\OPT)
	\enspace ,
\end{equation*}
thus leading to a constant-factor approximation for the considered~\ref{eq:CSFM} problem.

In summary, to obtain a constant-factor approximation for a~\ref{eq:CSFM} problem, it suffices to find a monotone constant-balanced CR scheme for a relaxation $P$ of $\mathcal{F}$ that is solvable. (We recall that $\mathcal{F}$ is assumed to be down-closed in this paper.) This reduces the problem to the construction of monotone CR schemes, where the latter is oblivious to the underlying submodular function. Another advantage of monotone CR schemes is their versatility in terms of constraints they can handle, as they allow for combining different constraint families. More precisely, let $\mathcal{F}_1\subseteq 2^E$ and $\mathcal{F}_2\subseteq 2^E$ be two down-closed families, and let $P_1$ and $P_2$ be relaxations of $\mathcal{F}_1$ and $\mathcal{F}_2$, respectively. Then, as shown in~\cite{chekuri_2014_submodular}, a monotone $(b, c_1)$-balanced CR scheme for $P_1$ and a monotone $(b, c_2)$-balanced CR scheme for $P_2$ can easily be combined to yield a monotone $(b, c_1 \cdot c_2)$-balanced CR scheme for $P_1\cap P_2$, which is a relaxation of $\mathcal{F}_1 \cap \mathcal{F}_2$. In fact, the scheme for $P_1\cap P_2$ simply returns the intersection of the sets obtained by applying both schemes.
Thus, one can get a monotone constant-balanced CR scheme for any combination of a constant number of arbitrary constraint types as long as one has a monotone constant-balanced CR scheme for every single constraint type. This construction goes along with a reduction in the balancedness, but especially when constraints of various types get intersected, it is often the only known way to get constant-factor approximations.

Examples of constraint types for which monotone constant-balanced CR schemes are known include matroid constraints, knapsack constraints, column-restricted packing constraints, interval packing constraints, and matching constraints (see~\cite{chekuri_2014_submodular, Feldman2013, BuchbinderFeldman2018}).
These schemes are with respect to the polytope $P=P_{\mathcal{F}}$ for matroid and bipartite matching constraints, or with respect to a natural linear relaxation for the other constraint types. %
Unfortunately, the construction and analysis of such schemes can be quite intricate.
Moreover, while we do have monotone constant-balanced CR schemes for many constraint types,
it seems very difficult to find optimal schemes (and prove their optimality). 
In fact, the only two schemes for such basic constraints that were proven to be (essentially) optimal are a monotone $(1-\epsilon, 1-\epsilon)$-balanced CR scheme for the natural linear relaxation of a knapsack constraint and a monotone $(b, \frac{1-\myexp^{-b}}{b})$-balanced CR scheme for the matroid polytope (see~\cite{chekuri_2014_submodular}).

The existence of the latter scheme was shown using linear programming duality and a concept called correlation gap, originally introduced by Agrawal, Ding, Saberi, and Ye~\cite{AgrawalDingSaberiYe2012} for set functions in the context of stochastic programming. In~\cite{chekuri_2014_submodular}, the notion was extended to families of subsets $\mathcal{\mathcal{F}}\subseteq 2^E$  as follows.
\begin{definition}\label{def:correlationGap}
	For $\mathcal{F} \subseteq 2^E$, the \emph{correlation gap} of $\mathcal{F}$ is defined as
	\begin{equation*}
		\kappa(\mathcal{F}) \coloneqq \inf_{x \in P_{\mathcal{F}}, y \geq 0} 
		\frac{\expval[   \max_{S \subseteq R(x), S \in \mathcal{F}} \sum_{e \in S} y_e     ]}{\sum_{e \in E} x_e y_e}
		\enspace ,
	\end{equation*}
	where $R(x) \subseteq E$ is a random set containing each element $e \in E$ independently with probability $x_e$.
\end{definition}
The notion of correlation gap is also used for a relaxation $P$ of $\mathcal{F}$, instead of $\mathcal{F}$. In this case, the correlation gap of $P$ is obtained by replacing $P_{\mathcal{F}}$ by $P$ in the above definition.

The correlation gap of a family $\mathcal{F}$ can be interpreted as a measure of how amenable a constraint family is to independent randomized rounding, as opposed to correlated rounding, and it has interesting applications beyond submodular maximization, for example in stochastic programming (see~\cite{AgrawalDingSaberiYe2012}).
Exploiting linear programming duality, one can show that the correlation gap of a family $\mathcal{F}$ is equal to the maximum number $c$ such that $P_{\mathcal{F}}$ admits a $(1,c)$-balanced CR scheme~\cite{chekuri_2014_submodular}.\footnote{%
	Both the notion of correlation gap as well as the result from Chekuri, Vondr\'ak, and Zenklusen~\cite{chekuri_2014_submodular}, which correspond to the case $b = 1$, extend to arbitrary $b \in [0,1]$. For simplicity, however, we only consider the case $b=1$.%
}
Hence, if one can compute (or find a lower bound for) the correlation gap of a family $\mathcal{F}$, one gets that there exists a CR scheme for $P_{\mathcal{F}}$ whose balancedness matches the found number. %
What is important to point out is that this CR scheme is \emph{not} guaranteed to be monotone.\footnote{%
	For the case of a matroid constraint, Chekuri, Vondr\'ak, and Zenklusen~\cite{chekuri_2014_submodular} showed that this approach actually allows for obtaining a monotone scheme. In general, however, this is not the case.
}
However, to obtain the rounding guarantees claimed by Theorem~\ref{thm:crSchemeGuarantee}, monotonicity is required.

\medskip

In the context of bipartite matchings, we present a new approach that addresses both difficulties of getting simple CR schemes on the one hand and (provably) optimal \emph{monotone} schemes on the other hand simultaneously. Our results have implications for~\ref{eq:CSFM} as well as for the correlation gap of matchings in bipartite graphs. While our technique works especially well for bipartite matching constraints, it is very general and its applicability is not restricted to this setting. 
Moreover, we also consider matchings in general (i.e., not necessarily bipartite) graphs.
In the following, we therefore provide a brief overview of the best known monotone CR schemes and bounds on the correlation gap for bipartite and general matchings, before we compare them with the results we obtain through our approach.

\subsection{Prior results on contention resolution schemes for matchings}

Bounds on the correlation gap for matchings in bipartite and general graphs follow directly or can be derived from results in~\cite{KarpSipser1981, chekuri_2014_submodular, CyganGrandoniMastrolilli2013, GuruganeshLee2017}.
While work from Karp and Sipser~\cite{KarpSipser1981} yields an upper bound of $0.544$ on the correlation gap for matchings in general graphs,
Guruganesh and Lee~\cite{GuruganeshLee2017} provide a lower bound of $\frac{1 - \myexp^{-2}}{2} \geq 0.4323$. For matchings in bipartite graphs, they were able to improve the lower bound to $1-\frac{5}{2 \myexp} + \frac{1}{\myexp} \geq 0.4481$.
 Using the highlighted link between correlation gap and CR schemes, this implies that there is a $(1, 0.4481)$-balanced CR scheme for the bipartite matching polytope and a $(1, 0.4323)$-balanced CR scheme for the general matching polytope. However, these schemes are not given explicitly, and, even more crucially, they are not guaranteed to be monotone. 
Theorem~\ref{thm:crSchemeGuarantee}, which is a crucial result for employing such schemes in the context of~\ref{eq:CSFM}, does therefore not apply to them.

 The currently best known \emph{monotone} CR scheme for bipartite matchings is obtained by observing that a bipartite matching constraint can be interpreted as the intersection of two partition matroid constraints. Thus, one can combine two monotone CR schemes for the (partition) matroid polytope. Using the monotone $(b,\frac{1-\myexp^{-b}}{b})$-balanced CR scheme for the matroid polytope from Chekuri, Vondr\'ak, and Zenklusen~\cite{chekuri_2014_submodular}, one obtains a monotone $(b, (\frac{1-\myexp^{-b}}{b})^2)$-balanced CR scheme for the bipartite matching polytope.
Moreover, there exists an elegant monotone CR scheme for the general matching polytope that is $(b, \myexp^{-2b})$-balanced (see~\cite{BuchbinderFeldman2018} for a description of the scheme reported in~\cite{FeldmanNaorSchwartz2011b}). 
This scheme is very similar to the one we described in Example~\ref{ex:CRschemeMatch}, which we showed to be monotone and $(1,\frac{1}{8})$-balanced. The difference is that the subsampling in step~\ref{enum:exampleCR-subsampling} is done in a more clever way, thereby improving the balancedness for $b=1$ to $\myexp^{-2} \geq 0.1353$.
For values of $b$ that are close to $1$, a better balancedness is achieved by doing the following. First, one randomly samples a vertex bipartition of the given graph and only keeps the edges crossing the bipartition, i.e., all edges not having both endpoints in the same part of the bipartition. Since the resulting subgraph is bipartite, one can then apply the scheme for the bipartite matching polytope described above. 
As one can check, this yields a monotone $(b, \frac{1}{2} (\frac{1-\myexp^{-b}}{b})^2)$-balanced CR scheme for the general matching polytope, where the additional factor of $\frac{1}{2}$ is due to an edge surviving the preliminary sampling step only with probability $\frac{1}{2}$.
Interestingly, all the above schemes for the general matching polytope also work for the degree-relaxation $\degreeP[G]$ (see Example~\ref{ex:CRschemeMatch}) of the set of all matchings in a not necessarily bipartite graph $G$, and do not require the stronger description of the non-bipartite matching polytope.

\begin{table}[ht]
	\centering
	\renewcommand{\arraystretch}{1.25}
	\setlength{\tabcolsep}{12.0pt}
	\begin{tabular}{ lll }
		\toprule[.08em]
		Constraint type   &   \parbox[c]{\widthof{monotone CR scheme}}{Balancedness of\\monotone CR scheme}  &  \parbox[c]{\widthof{Lower bound on}}{Lower bound on\\correlation gap}   \\
		\midrule[.08em]
		\multirow{2}{*}{Bipartite matching}   &  $(b, (\frac{1 - \myexp^{-b}}{b})^2)$  \enspace\cite{chekuri_2014_submodular}  &   \\ 
		&  $(1,0.3995)$\enspace\cite{chekuri_2014_submodular}     &  $0.4481$\enspace\cite{GuruganeshLee2017}   \\
		\midrule
		\multirow{3}{*}{General matching}   &   $(b, \myexp^{-2b})$ \enspace\cite{FeldmanNaorSchwartz2011b}  &   \\
		&    $(b, \frac{1}{2}(\frac{1 - \myexp^{-b}}{b})^2)$  \enspace\cite{chekuri_2014_submodular} & \\
		&  $(1,0.1997)$\enspace\cite{chekuri_2014_submodular}     &  $0.4323$\enspace\cite{GuruganeshLee2017} \\
		\bottomrule[.08em]
	\end{tabular}
	\renewcommand{\arraystretch}{1}
	\setlength{\tabcolsep}{6.0pt}
	\caption{Prior results on %
		monotone CR schemes and on %
		the correlation gap for bipartite and general matchings.}
	\label{tab:priorResults}
\end{table}

In Table~\ref{tab:priorResults}, an overview of the above results is given. Especially in the case of general matchings, there is a large gap between the balancedness of the best known monotone CR scheme and the balancedness that is known to be achievable by some (not necessarily monotone) CR scheme (via the correlation gap). Moreover, no result in Table~\ref{tab:priorResults} is known to be tight.

\subsection{Our results}
With the help of a new viewpoint that we sketch in Section~\ref{sec:ourTechniques}, we are able to significantly improve on the state of the art highlighted in Table~\ref{tab:priorResults}.
More precisely, we provide an improved monotone CR scheme for the bipartite matching polytope, which we later show to be optimal. To state its balancedness in the theorem below, we use the function
\begin{equation*}
\beta(b) \coloneqq  \expval \bigg[ \frac{1}{1 +  \max\{ \PoisD[1]{b}, \PoisD[2]{b} \}  }  \bigg] \text{ for }b \in [0,1]\enspace,
\end{equation*}
where $\PoisD[1]{b}$ and $\PoisD[2]{b}$ are two independent Poisson random variables with parameter $b$.
\begin{theorem}\label{thm:mCRs_existence_bipartite}
There is a monotone $(b, \beta(b))$-balanced contention resolution scheme for the bipartite matching polytope.
\end{theorem}
We notice that for $b=1$, one obtains $\beta(1) \geq 0.4762$.
By exploiting the highlighted link between CR schemes and correlation gap, this implies the following.
\begin{corollary}\label{cor:corGap_bipartite}
	The correlation gap for matchings in bipartite graphs is at least $\beta(1) \geq 0.4762$.\footnote{%
		This number matches the lower bound of Guruganesh and Lee~\cite{GuruganeshLee2017} on the unweighted version of the correlation gap for matchings in bipartite graphs (fixing %
		$y \equiv 1$ in Definition~\ref{def:correlationGap}). Since they were only able to prove this bound in the unweighted setting, however, their result had no implications regarding the existence of CR schemes with a certain balancedness.%
	}
\end{corollary}

Moreover, the same polyhedral viewpoint that we use for designing CR schemes also enables us to obtain an upper bound on the best possible monotone CR scheme for bipartite matchings, showing that our monotone $(b,\beta(b))$-balanced CR scheme is optimal.
\begin{theorem}\label{thm:mCRs_optimality_bipartite}
Let $\pi$ be a monotone $(b,c)$-balanced contention resolution scheme for the bipartite matching polytope, where $b,c \in [0,1]$.
	Then, it holds that $c \leq \beta(b)$.
\end{theorem}
Note that Theorem~\ref{thm:mCRs_optimality_bipartite} does not rule out the existence of a non-monotone CR scheme with higher balancedness.\footnote{%
	However, applying CR schemes in the context of~\ref{eq:CSFM} as described earlier, i.e., using Theorem~\ref{thm:crSchemeGuarantee}, requires monotonicity of the scheme.
} 
This is also why we cannot conclude that the correlation gap for matchings in bipartite graphs is exactly $\beta(1)$.

Given a bipartite graph $G=(V,E)$ and a non-monotone submodular function $f \colon 2^E \to \mathbb{R}_{\geq 0}$, we can combine the best known algorithm for approximately maximizing the multilinear extension with our monotone CR scheme to obtain a constant-factor approximation for maximizing $f$ over the bipartite matchings of $G$.
In fact, this yields an $\alpha(b) \cdot \beta(b)$-approximation, where $\alpha(b)$ is a constant such that we can approximately maximize $\mlinext$ over any down-closed solvable polytope $P \subseteq [0,1]^E$ to get a point $x\in b P$ satisfying $\mlinext(x) \geq \alpha(b) \cdot f(\OPT)$.
 Using the algorithm from Feldman, Naor, and Schwartz~\cite{FeldmanNaorSchwartz2011b} to get $\alpha(b)=b\cdot \myexp^{-b}-\epsilon$ (for arbitrary $\epsilon > 0$) and choosing $b=0.6$, the factor we can achieve this way is just slightly above $0.2$.
Since matchings form a $2$-exchange system, %
another option would be to use the algorithm in~\cite{FeldmanNaorSchwartzWard2011, Feldman2013} for $k$-exchange systems to get a $\frac{1}{4+\epsilon}$-approximation if the goal is to maximize $f$ only over a bipartite matching constraint.
While the latter yields a better constant, the approach based on relaxation and rounding is much more flexible in the sense that we can combine bipartite matching constraints with many other types of constraints and still get a constant-factor approximation (if there are monotone constant-balanced CR schemes for the other constraint types).

Our result for bipartite matchings also yields a monotone $(1, 0.3174)$-balanced CR scheme for the general matching polytope, i.e., the matching polytope of a not necessarily bipartite graph. %
Using further ideas, though, we manage to get a scheme with a stronger balancedness. In the theorem below, we use the function
\begin{equation*}
	\gamma(b) \coloneqq  \expval \bigg[ \frac{1}{1 +   \PoisD{2b}  }  \bigg] 
	= \frac{1- \myexp^{-2b}}{2b}
	\text{ for }b \in [0,1]\enspace,
\end{equation*}
where $\PoisD{2b}$ is a Poisson random variable with parameter $2b$.\footnote{%
The above equality holds because $\expval \big[ \frac{1}{1+\PoisD{2b}}  \big]
= \myexp^{-2b} \cdot \sum_{k=0}^\infty \frac{(2b)^k}{k! \cdot (k+1)} 
= \myexp^{-2b} \cdot \frac{ 1}{2b} \cdot\sum_{k=0}^\infty \frac{(2b)^{k+1}}{(k+1)!} 
=  \myexp^{-2b} \cdot \frac{(\myexp^{2b}-1)}{2b}
$.%
}
\begin{theorem}\label{thm:mCRs_existence_general}
There exists a monotone $(b, \gamma(b))$-balanced contention resolution scheme for the general matching polytope.
\end{theorem}
We note that for $b=1$, one obtains a balancedness of $\gamma(1) = \frac{ 1 - \myexp^{-2} }{2} \geq 0.4323$.
Compared to the previously best known monotone CR schemes for matchings with balancedness $0.1997$ (see Table~\ref{tab:priorResults}), our scheme represents a significant improvement.
Interestingly, the existence of a $(1,\gamma(1))$-balanced CR scheme for matchings already followed from the result %
of Guruganesh and Lee~\cite{GuruganeshLee2017} on the correlation gap for matchings in general graphs.
It was not known, however, if a \emph{monotone} scheme exists that attains the same balancedness.
Moreover, while Guruganesh and Lee~\cite{GuruganeshLee2017} arrive at the bound $\frac{ 1 - \myexp^{-2} }{2}$ via a differential equation,
our approach for general matchings also shows where this number %
 naturally comes from. %

One might wonder now
whether this number may be the optimal balancedness any monotone CR scheme for the general matching polytope can achieve.
By combining our schemes for bipartite and general matchings, however, we show that this is not the case by constructing a scheme whose balancedness is slightly higher than $\frac{ 1 - \myexp^{-2} }{2}$.
Since this improved scheme is unlikely to be optimal, and our goal in designing it was focused on showing that $\frac{1- \myexp^{-2}}{2}$ is not the optimal balancedness, we only consider the case $b=1$ and do not attempt to analyze the balancedness of this scheme exactly.
\begin{theorem}\label{thm:mCRs_existence_general_improved}
	There exists a monotone $(1, \gamma(1) + 0.0003)$-balanced contention resolution scheme for the general matching polytope.
\end{theorem}

Again exploiting the link between CR schemes and correlation gap, we can immediately infer the following result.
\begin{corollary}\label{cor:corGap_general}
	The correlation gap for matchings in general graphs is at least $\gamma(1) + 0.0003 \geq 0.4326$.
\end{corollary}

As for bipartite matchings, combining an algorithm for maximizing the multilinear extension with our monotone CR scheme will not give the best known approximation guarantee for maximizing a non-monotone submodular function over a matching constraint alone. The advantage of using CR schemes, however, is again their versatility when it comes to combinations of matching constraints with other constraint types. Such combinations quickly lead to problems where approaches based on monotone CR schemes give the only known constant-factor approximations, and our improvements in terms of monotone CR schemes also improve the approximation factors for~\ref{eq:CSFM} in any such setting.

\begin{table}[ht]
	\centering
	\renewcommand{\arraystretch}{1.25}
	\setlength{\tabcolsep}{12.0pt}
	\begin{tabular}{ lll }
		\toprule[.08em]
		Constraint type   &   \parbox[c]{\widthof{monotone CR scheme}}{Balancedness of\\monotone CR scheme}  &  \parbox[c]{\widthof{Lower bound on}}{Lower bound on\\correlation gap}   \\
		\midrule[.08em]
		\multirow{2}{*}{Bipartite matching}   &  $(b, \beta(b))$    &   \\ 
		&  $(1,0.4762)$     &  $0.4762$  \\
		\midrule
		\multirow{2}{*}{General matching}   &    $(b, \gamma(b))$  & \\
		&  $(1,0.4326)$     &  $0.4326$\\
		\bottomrule[.08em]
	\end{tabular}
	\renewcommand{\arraystretch}{1}
	\setlength{\tabcolsep}{6.0pt}
	\caption{Our results on %
		monotone CR schemes and on %
		the correlation gap for bipartite and general matchings.
		The balancedness given for general matchings in the case $b=1$ comes from Theorem~\ref{thm:mCRs_existence_general_improved}, i.e., it corresponds to $\gamma(1) + 0.0003$.}
	\label{tab:ourResults}
\end{table}

An overview of our results is provided in Table~\ref{tab:ourResults}. In particular, we are able to close the gap between the balancedness of monotone CR schemes and the best known lower bounds on the correlation gaps in the context of bipartite and general matchings. 
The correlation gaps might be %
larger than our bounds, though. 

\subsection{Organization of the paper}

In Section~\ref{sec:ourTechniques}, we provide a general discussion of our novel viewpoint on CR schemes, and show how it can be used to derive new CR schemes.
Leveraging this viewpoint, Section~\ref{sec:bipMatch} presents our optimal monotone CR scheme for bipartite matchings and proves its optimality.
Finally, our results for matchings in general graphs are discussed in Section~\ref{sec:genMatch}.

\section{Our techniques}\label{sec:ourTechniques}

Let $\mathcal{F}\subseteq 2^E$ be a down-closed family. For simplicity, we present our technique for the case where we are interested in a CR scheme for $P=P_{\mathcal{F}}$. This also covers all cases considered in the following.

We start by highlighting a viewpoint on CR schemes focusing on marginals. To do so, consider a CR scheme $\pi$ for $P_{\mathcal{F}}$. For any $x\in P_{\mathcal{F}}$ and $A\subseteq \supp(x)$, the CR scheme returns a (typically random) set $\pi_x(A)\in \mathcal{F}$ with $\pi_x(A)\subseteq A$. This defines marginals $y_x^A\in [0,1]^E$ of $\pi_x(A)$, i.e.,
\begin{equation*}
(y_x^A)_e \coloneqq \Pr[e\in \pi_x(A)] \qquad \forall x\in P_{\mathcal{F}}, A\subseteq \supp(x), \text{ and } e\in E\enspace.
\end{equation*}
Notice that the above definition of $y_x^A$ can also be written compactly as $y_x^A \coloneqq \expval\left[\chi^{\pi_x(A)}\right]$, where $\chi^{\pi_{x}(A)}\in \{0,1\}^E$ is the (random) characteristic vector of $\pi_x(A)$.
Our goal is to design CR schemes through their marginals $y_x^A$, and verify the balancedness and monotonicity properties directly through them. %
 As we will see, despite being merely a change in viewpoint, this can lead to several simplifications when designing strong monotone CR schemes.
 
First, observe that the marginals $y_x^A$ need to satisfy for all $x\in P_{\mathcal{F}}$ and $A\subseteq \supp(x)$ that
\begin{equation}\label{eq:marginals_support}
\supp(y_x^A) \subseteq A
\end{equation}
and
\begin{equation}\label{eq:marginals_feasibility}
y_x^A \in P_{\mathcal{F}} \enspace.
\end{equation}
The first property is a rephrasing of $\pi_x(A)\subseteq A$, and the second one follows because $\pi_x(A)$ is a (random) set in $\mathcal{F}$.
Conversely, it turns out that if, for any $x\in P_{\mathcal{F}}$ and $A\subseteq \supp(x)$, one has an efficient procedure to return a vector $y_x^A\in [0,1]^E$ fulfilling~\eqref{eq:marginals_support} and~\eqref{eq:marginals_feasibility}, then this can easily be transformed into a CR scheme whose marginals are given by $y_x^A$. Indeed, it suffices to first decompose $y_x^A$ into a convex combination of vertices of $P_{\mathcal{F}}$, i.e.,
\begin{equation}\label{eq:convexDecomp}
y_x^A = \sum_{i=1}^k \lambda_i \chi^{S_i}\enspace,
\end{equation}
where $k\in \mathbb{Z}_{\geq 1}$, $\lambda_i \geq 0$ and $S_i\in \mathcal{F}$ for $i\in [k]$, and $\sum_{i=1}^k \lambda_i =1$. Then, a random set among $S_1,\ldots, S_k$ is returned, where $S_i$ is chosen with probability $\lambda_i$.
 One can easily check that this leads to a CR scheme with marginals $y_x^A$.
We highlight that decomposing the vector $y_x^A\in P_{\mathcal{F}}$ into a convex decomposition as described in~\eqref{eq:convexDecomp} can be done efficiently through standard techniques whenever $P_{\mathcal{F}}$ is solvable (see Corollary~14.1f and Corollary~14.1g in~\cite{Schrijver1998}).\footnote{%
	In the context of~\ref{eq:CSFM}, we have to approximately maximize the multilinear extension over $P_{\mathcal{F}}$ before (monotone) CR schemes even come into play. Algorithms achieving that goal typically require that $P_{\mathcal{F}}$ be solvable. %
	Thus, determining convex decompositions of vectors in $P_{\mathcal{F}}$ will never be an issue in these cases.%
}

Notice that the marginals $y_x^A$ carry less information than a CR scheme that corresponds to them, because the CR scheme also determines, for a fixed $x\in P_{\mathcal{F}}$ and $A\subseteq \supp(x)$, the precise distribution of the sets $\pi_x(A) \in \mathcal{F}$ that correspond to the marginals $y_x^A$. However, this additional information turns out to be irrelevant for the way CR schemes are applied to~\ref{eq:CSFM} problems, because the two key properties we need from CR schemes in this context are monotonicity and balancedness, both of which can be determined solely from the marginals. Indeed, monotonicity of a CR scheme $\pi$ translates straightforwardly to the following property on the marginals $y_x^A$ of $\pi$:
\begin{equation}\label{eq:marginals_monotonicity}
	(y_x^{A})_e 
	\geq
	 (y_x^{B})_e
	 \qquad \forall e\in A\subseteq B \subseteq \supp(x)
	 \enspace .
\end{equation}

Similarly, we recall that a CR scheme $\pi$ for $P_{\mathcal{F}}$ is $(b,c)$-balanced if $\Pr[e\in \pi_x(R(x))] \geq c\cdot x_e$ for all $x\in b P_{\mathcal{F}}$ and $e\in E$. Hence, this property solely depends on the probabilities $\Pr[e\in \pi_x(R(x))]$, which can be written using the marginals $y_x^A$ of $\pi$ as %
\begin{equation*}%
\Pr[e\in \pi_x(R(x))] = \expval\Big[ \big( y_x^{R(x)}\big)_e\Big]\enspace.
\end{equation*}
Using the above relation, $(b,c)$-balancedness of a contention resolution scheme $\pi$ for $P_{\mathcal{F}}$ can be expressed in terms of its marginals as
\begin{equation}\label{eq:marginals_balancedness}
\expval\Big[ \big(y_x^{R(x)}\big)_e\Big] \geq c\cdot x_e
 \qquad \forall x\in bP_{\mathcal{F}} \textup{ and } e\in \supp(x)\enspace.
\end{equation}

The following proposition summarizes the above discussion.
\begin{proposition}\label{prop:marginalCRschemes}
Let $b,c\in [0,1]$ and $\mathcal{F}\subseteq 2^E$ be a down-closed family such that $P_{\mathcal{F}}$ is solvable.
Any efficient algorithm that for every $x\in b P_{\mathcal{F}}$ and $A\subseteq \supp(x)$ returns a vector $y^A_x\in [0,1]^E$ such that
	\begin{enumerate}[label=\normalfont(CR\arabic*), leftmargin=4em]
		\item\label{item:marginalSupp} $\supp(y_x^A) \subseteq A$, and
		\item\label{item:marginalConvComb} $y_x^A \in P_{\mathcal{F}}$, 
	\end{enumerate}
with the additional property that
	\begin{enumerate}[label=\normalfont(CR\arabic*), leftmargin=4em, resume]
		\item\label{item:marginalBalance} $\expval\Big[ \big( y_x^{R(x)}\big)_e\Big] \geq c\cdot x_e \qquad \forall e\in \supp (x)$,
	\end{enumerate}
can be efficiently transformed into a $(b,c)$-balanced CR scheme $\pi$ for $P_{\mathcal{F}}$. Moreover, $\pi$ is monotone if
\begin{enumerate}[label=\normalfont(CR\arabic*), leftmargin=4em, resume]
\item\label{item:marginalMonotone} $(y_x^{A})_e \geq (y_x^{B})_e \qquad \forall e\in A \subseteq B \subseteq \supp(x) $.
\end{enumerate}
Conversely, the marginals of any monotone $(b,c)$-balanced CR scheme for $P_{\mathcal{F}}$ satisfy~\ref{item:marginalSupp}--\ref{item:marginalMonotone}.
\end{proposition}

Defining CR schemes via marginal vectors $y_x^A$ has two main advantages. First, contrary to the random sets $\pi_x(A)$, the vectors $y_x^A$ are deterministic, and, more importantly, we can use polyhedral techniques to find CR schemes through these vectors.
When working with marginal vectors $y_x^A$, the only source of randomness left in the analysis is due to the random set $R(x)$. %
The randomization of the CR scheme derived from $y_x^A$, through the sampling from the convex decomposition~\eqref{eq:convexDecomp}, is not relevant for the analysis as discussed above.

We demonstrate these advantages of working directly with marginal vectors $y_x^A$ in %
the example below,
 where we give a simple monotone $0.4481$-balanced CR scheme for the bipartite matching polytope. This already improves on the previously best monotone CR scheme, which achieves a balancedness of $0.3995$, and matches the previously best known lower bound on the correlation gap for matchings in bipartite graphs (see Table~\ref{tab:priorResults}). Later, in Section~\ref{sec:bipMatch}, we show how to improve this scheme to obtain an optimal monotone CR scheme for the bipartite matching polytope.

\begin{example}\label{ex:CRscheme_from_marginals}
For a graph $G=(V,E)$, we denote the family of all matching of $G$ by $\match[G]\subseteq 2^E$. 
Here, we are interested in obtaining a monotone CR scheme for the matching polytope $\bipmatchP[G]$ of a bipartite graph $G=(V,E)$. We recall that for bipartite graphs $G$, the polytope $\bipmatchP[G]$ is equal to the degree-relaxation $\degreeP[G]$, i.e., it has the following inequality description (see, e.g.,~\cite{schrijver_2003_combinatorial,korte_2018_combinatorial}):
\begin{equation*}
\bipmatchP[G] = \big\{ x\in \mathbb{R}_{\geq 0}^E \;\big\vert\; x(\delta(v)) \leq 1 \;\;\forall v \in V\big\}\enspace.
\end{equation*}
For $c \in [0,1]$ to be determined later, we now derive a monotone $c$-balanced CR scheme for $\bipmatchP[G]$ via its marginals. Below, we describe how these marginals are defined. An illustration is provided in Figure~\ref{fig:exampleMarginalCRS}.
\begin{thmEnvBox}[width=15cm,center]{}
	Given $x\in \bipmatchP[G]$ and $A\subseteq \supp(x)$, we construct a marginal vector $y_x^A \in \bipmatchP[G]$ with $\supp(y_x^A) \subseteq A$ as follows:
	\begin{equation*}
	\big(y^A_x\big)_e \coloneqq
	\begin{cases}
	\frac{1}{\max\{|\delta(u)\cap A|,|\delta(v)\cap A|\}} &\forall e = \{u,v\} \in A\enspace,\\
	0 &\forall e\in E\setminus A \enspace.
	\end{cases}
	\end{equation*}
\end{thmEnvBox}
\begin{figure}[ht]
	\centering

\tikzstyle{vertex}=[circle, draw=black, thick, inner sep=1pt, minimum size=2mm, fill=black!10]
\tikzstyle{edgenode}=[rectangle, fill=white, fill opacity=0.8, text opacity=1, inner sep=1pt]
\tikzstyle{edge}=[thick, draw=black]
\tikzstyle{selectedEdge}=[ultra thick, draw=\colorofset]
\tikzstyle{discardedEdge}=[semithick, dashed, draw=black]
\tikzstyle{arc}=[thick, draw=black!60, -latex]
\tikzstyle{boxed}=[thick, draw=black!60, rounded corners]
\tikzstyle{boxedFat}=[thick, draw=black, rounded corners]
\tikzstyle{boxedLight}=[thick, draw=black, rounded corners, fill = black!15, minimum height=0.75*1cm]

\begin{tikzpicture}[scale=0.75]

\begin{scope}[shift={(3.5,0.5)}, local bounding box=marginal1]

\draw[boxed] (0,0) rectangle (5,5);
\draw[boxed] (0,0) rectangle (1.9,1);
\node[font=\footnotesize, color=\colorofset] at (0.95,0.5) {$\strut A$};

\node[vertex] (u1) at (0.75,4.5) {};
\node[vertex] (u2) at (0.75,3.5) {};
\node[vertex] (u3) at (0.75,2.5) {};
\node[vertex] (u4) at (0.75,1.5) {};
\node[vertex] (v1) at (4.25,4.5) {};
\node[vertex] (v2) at (4.25,3.5) {};
\node[vertex] (v3) at (4.25,2.5) {};
\node[vertex] (v4) at (4.25,1.5) {};
\node[vertex] (v5) at (4.25,0.5) {};

\draw[selectedEdge] (u1) -- (v1);
\draw[selectedEdge] (u1) -- (v2);
\draw[selectedEdge] (u1) -- (v4);
\draw[discardedEdge] (u2) -- (v1);
\draw[selectedEdge] (u2) -- (v2);
\draw[selectedEdge] (u3) -- (v2);
\draw[selectedEdge] (u3) -- (v3);
\draw[selectedEdge] (u3) -- (v4);
\draw[discardedEdge] (u4) -- (v4);
\draw[selectedEdge] (u4) -- (v5);

\end{scope}

\begin{scope}[shift={(10.5,0.5)}, local bounding box=marginal2]

\draw[boxed] (0,0) rectangle (5,5);
\draw[boxed] (0,0) rectangle (1.9,1);
\node[font=\footnotesize, color=\colorofx] at (0.95,0.5) {$\strut y_x^A$};

\node[vertex] (u1) at (0.75,4.5) {};
\node[vertex] (u2) at (0.75,3.5) {};
\node[vertex] (u3) at (0.75,2.5) {};
\node[vertex] (u4) at (0.75,1.5) {};
\node[vertex] (v1) at (4.25,4.5) {};
\node[vertex] (v2) at (4.25,3.5) {};
\node[vertex] (v3) at (4.25,2.5) {};
\node[vertex] (v4) at (4.25,1.5) {};
\node[vertex] (v5) at (4.25,0.5) {};

\draw[edge] (u1) -- node[font=\footnotesize, color=\colorofx, edgenode, pos=0.5] {$\sfrac{1}{3}$} (v1);
\draw[edge] (u1) -- node[font=\footnotesize, color=\colorofx, edgenode, pos=0.275] {$\sfrac{1}{3}$} (v2);
\draw[edge] (u1) -- node[font=\footnotesize, color=\colorofx, edgenode, pos=0.85] {$\sfrac{1}{3}$} (v4);
\draw[edge] (u2) -- node[font=\footnotesize, color=\colorofx, edgenode, pos=0.725] {$0$} (v1);
\draw[edge] (u2) -- node[font=\footnotesize, color=\colorofx, edgenode, pos=0.6] {$\sfrac{1}{3}$} (v2);
\draw[edge] (u3) -- node[font=\footnotesize, color=\colorofx, edgenode, pos=0.275] {$\sfrac{1}{3}$} (v2);
\draw[edge] (u3) -- node[font=\footnotesize, color=\colorofx, edgenode, pos=0.475] {$\sfrac{1}{3}$} (v3);
\draw[edge] (u3) -- node[font=\footnotesize, color=\colorofx, edgenode, pos=0.6] {$\sfrac{1}{3}$} (v4);
\draw[edge] (u4) -- node[font=\footnotesize, color=\colorofx, edgenode, pos=0.475] {$0$} (v4);
\draw[edge] (u4) -- node[font=\footnotesize, color=\colorofx, edgenode, pos=0.7] {$1$} (v5);

\end{scope}

\begin{scope}[shift={(1,-7)}, local bounding box=decomp1]

\draw[boxed] (0,0) rectangle (5,5);
\draw[boxed] (0,0) rectangle (1.9,1);
\node[font=\footnotesize, color=\colorofset] at (0.95,0.5) {$\strut \lambda_1=\sfrac{1}{3}$};

\node[vertex] (u1) at (0.75,4.5) {};
\node[vertex] (u2) at (0.75,3.5) {};
\node[vertex] (u3) at (0.75,2.5) {};
\node[vertex] (u4) at (0.75,1.5) {};
\node[vertex] (v1) at (4.25,4.5) {};
\node[vertex] (v2) at (4.25,3.5) {};
\node[vertex] (v3) at (4.25,2.5) {};
\node[vertex] (v4) at (4.25,1.5) {};
\node[vertex] (v5) at (4.25,0.5) {};

\draw[discardedEdge] (u1) -- (v1);
\draw[selectedEdge] (u1) -- (v2);
\draw[discardedEdge] (u1) -- (v4);
\draw[discardedEdge] (u2) -- (v1);
\draw[discardedEdge] (u2) -- (v2);
\draw[discardedEdge] (u3) -- (v2);
\draw[discardedEdge] (u3) -- (v3);
\draw[selectedEdge] (u3) -- (v4);
\draw[discardedEdge] (u4) -- (v4);
\draw[selectedEdge] (u4) -- (v5);

\end{scope}

\begin{scope}[shift={(7,-7)}, local bounding box=decomp2]

\draw[boxed] (0,0) rectangle (5,5);
\draw[boxed] (0,0) rectangle (1.9,1);
\node[font=\footnotesize, color=\colorofset] at (0.95,0.5) {$\strut \lambda_2=\sfrac{1}{3}$};

\node[vertex] (u1) at (0.75,4.5) {};
\node[vertex] (u2) at (0.75,3.5) {};
\node[vertex] (u3) at (0.75,2.5) {};
\node[vertex] (u4) at (0.75,1.5) {};
\node[vertex] (v1) at (4.25,4.5) {};
\node[vertex] (v2) at (4.25,3.5) {};
\node[vertex] (v3) at (4.25,2.5) {};
\node[vertex] (v4) at (4.25,1.5) {};
\node[vertex] (v5) at (4.25,0.5) {};

\draw[discardedEdge] (u1) -- (v1);
\draw[discardedEdge] (u1) -- (v2);
\draw[selectedEdge] (u1) -- (v4);
\draw[discardedEdge] (u2) -- (v1);
\draw[selectedEdge] (u2) -- (v2);
\draw[discardedEdge] (u3) -- (v2);
\draw[selectedEdge] (u3) -- (v3);
\draw[discardedEdge] (u3) -- (v4);
\draw[discardedEdge] (u4) -- (v4);
\draw[selectedEdge] (u4) -- (v5);

\end{scope}

\begin{scope}[shift={(13,-7)}, local bounding box=decomp3]

\draw[boxed] (0,0) rectangle (5,5);
\draw[boxed] (0,0) rectangle (1.9,1);
\node[font=\footnotesize, color=\colorofset] at (0.95,0.5) {$\strut \lambda_3=\sfrac{1}{3}$};

\node[vertex] (u1) at (0.75,4.5) {};
\node[vertex] (u2) at (0.75,3.5) {};
\node[vertex] (u3) at (0.75,2.5) {};
\node[vertex] (u4) at (0.75,1.5) {};
\node[vertex] (v1) at (4.25,4.5) {};
\node[vertex] (v2) at (4.25,3.5) {};
\node[vertex] (v3) at (4.25,2.5) {};
\node[vertex] (v4) at (4.25,1.5) {};
\node[vertex] (v5) at (4.25,0.5) {};

\draw[selectedEdge] (u1) -- (v1);
\draw[discardedEdge] (u1) -- (v2);
\draw[discardedEdge] (u1) -- (v4);
\draw[discardedEdge] (u2) -- (v1);
\draw[discardedEdge] (u2) -- (v2);
\draw[selectedEdge] (u3) -- (v2);
\draw[discardedEdge] (u3) -- (v3);
\draw[discardedEdge] (u3) -- (v4);
\draw[discardedEdge] (u4) -- (v4);
\draw[selectedEdge] (u4) -- (v5);

\end{scope}

\draw[boxedFat] ($(marginal1.north west) + (-1,1.5)$) rectangle ($(marginal2.south east) + (1,-0.5)$);

\node[boxedLight, font=\small, minimum width=0.75*7cm, anchor=north east, inner sep = 0pt] at ($(marginal2.north east)+(1,1.5)$) {\textbf{Defining the marginals\strut}};

\draw[arc] (marginal1.east) -- (marginal2.west);
\draw[arc] (marginal2.south) -- ($(decomp1.north)+(1.5,0)$);
\draw[arc] (marginal2.south) -- (decomp2.north);
\draw[arc] (marginal2.south) -- ($(decomp3.north)+(-1.5,0)$);

\node[boxedLight, font=\footnotesize, anchor=west] at (1,-1) {Convex decomposition\strut};

\end{tikzpicture}  	\caption{Defining marginal vectors in order to get a monotone contention resolution scheme for the bipartite matching polytope $\bipmatchP[G]$ via Proposition~\ref{prop:marginalCRschemes}.}\label{fig:exampleMarginalCRS}
\end{figure}
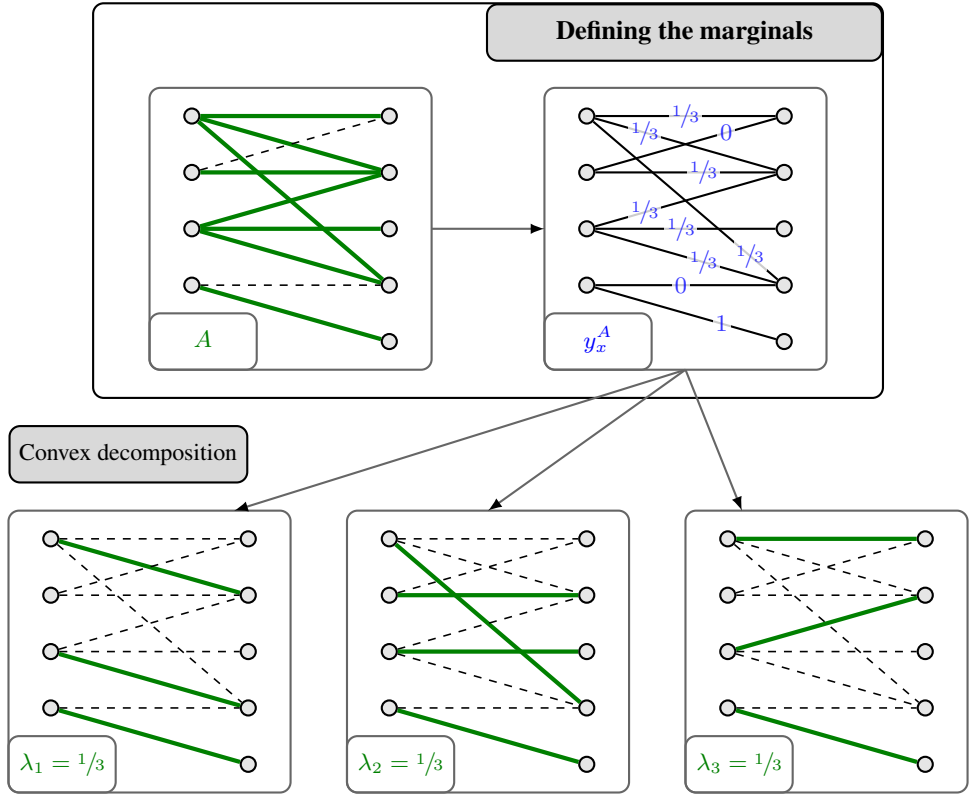
Observe that these vectors $y_x^A$ clearly satisfy~\ref{item:marginalSupp} and~\ref{item:marginalConvComb}. 
Since we want to apply Proposition~\ref{prop:marginalCRschemes} in order to get a monotone $c$-balanced CR scheme for $\bipmatchP[G]$, we also need to check conditions~\ref{item:marginalBalance} and~\ref{item:marginalMonotone}, which correspond to balancedness and monotonicity, respectively. 

For the latter, we notice that for any $x \in \bipmatchP[G]$ and $e=\{u,v\}\in A \subseteq B \subseteq \supp(x)$, we have
\begin{equation*}
\big(y_x^{A}\big)_e = \frac{1}{ \max\{  |\delta(u) \cap A|, |\delta(v) \cap A|  \} } \geq  \frac{1}{ \max\{  |\delta(u) \cap B|, |\delta(v) \cap B|  \} } = \big( y_x^{B}\big)_e
		\enspace ,
\end{equation*}
thus showing~\ref{item:marginalMonotone}, i.e., the property on marginal vectors that implies monotonicity. %

To determine the balancedness, we fix $x \in \bipmatchP[G]$ and $e = \{u,v\} \in \supp (x)$.
We expand the expression
$\expval\big[(y_x^{R(x)})_e \big]/x_e
=
\expval\big[(y_x^{R(x)})_e \;\big\vert\;  e\in R(x)\big]
$
as follows:
\begin{equation}\label{eq:bip_easyCR_balancedness}
\begin{aligned}
\expval\Big[\big(y_x^{R(x)}\big)_e \;\Big\vert\; e \in R(x) \Big] 
& = \expval\bigg[   \frac{1}{\max\{  |\delta(u) \cap R(x)|, |\delta(v) \cap R(x)|  \}} \;\bigg\vert\; e \in R(x) \bigg] \\
& = \expval\bigg[   \frac{1}{ 1 + \max\{  |\delta(u) \cap (R(x)\setminus \{e\})|, |\delta(v) \cap (R(x)\setminus \{e\})|  \}} \bigg]\enspace. 
\end{aligned}
 \end{equation}
One way to get a simple, though loose, lower bound on this expression is through Jensen's inequality:
\begin{equation}\label{eq:bip_easyCR_balancedness_lb}
\begin{aligned}
\expval\Big[\big(y_x^{R(x)}\big)_e \;\Big\vert\; e \in R(x) \Big] 
& \geq   \expval\bigg[    \frac{1}{1 + 
|\delta(u)\cap (R(x)\setminus \{e\})| +
|\delta(v)\cap (R(x)\setminus \{e\})|}  \bigg] \\
& \geq    \frac{1}{1 + \expval[ 
|\delta(u)\cap (R(x)\setminus \{e\})|]
+ \expval[
|\delta(v)\cap (R(x)\setminus \{e\})|]
} \\
& =    \frac{1}{1 +
  x(\delta(u)\setminus \{e\}) +
  x(\delta(v)\setminus \{e\})
}\\
& \geq    \frac{1}{1 + x(\delta(u)) + x(\delta(v))}\\
& \geq \frac{1}{3}
\enspace .
\end{aligned}
\end{equation}
In the last step, we used that $x \in \bipmatchP[G]$.
The above implies that the monotone CR scheme we get from our marginals through Proposition~\ref{prop:marginalCRschemes}  is $\frac{1}{3}$-balanced.
A more detailed analysis of expression~\eqref{eq:bip_easyCR_balancedness} shows that an even better balancedness than $\frac{1}{3}$ is actually obtained.
Indeed, Guruganesh and Lee~\cite[Section~4]{GuruganeshLee2017} also analyze expression~\eqref{eq:bip_easyCR_balancedness}, which, interestingly, they obtain through a different technique. They show the following sharp lower bound:
\begin{equation}\label{eq:GuruganeshLee}
\expval\bigg[   \frac{1}{ 1 + \max\{  |\delta(u) \cap (R(x)\setminus \{e\})|, |\delta(v) \cap (R(x)\setminus \{e\})|  \}} \bigg]
	\geq 1 - \frac{5}{2 \myexp} + \frac{1}{\myexp}
	\geq 0.4481
	\enspace .
\end{equation}
Thus, our monotone CR scheme for the bipartite matching polytope $\bipmatchP[G]$ is $0.4481$-balanced.
	Compared to the approach in~\cite{GuruganeshLee2017}, which uses local distribution schemes, we obtain the expression in~\eqref{eq:bip_easyCR_balancedness} in a much simpler and quicker way.
	More importantly, while the existence of a $0.4481$-balanced CR scheme for $\bipmatchP[G]$ already followed from the result by Guruganesh and Lee~\cite{GuruganeshLee2017} on the correlation gap of matchings in bipartite graphs, our approach shows that this balancedness can be achieved with a conceptually very simple scheme that satisfies the crucial monotonicity property (see Table~\ref{tab:priorResults} for prior results), which is needed to apply Theorem~\ref{thm:crSchemeGuarantee} to do the rounding step in the context of submodular maximization.
\end{example}

\section{An optimal monotone contention resolution scheme for bipartite matchings}\label{sec:bipMatch}

As before, let $G=(V,E)$ be a bipartite graph and $\bipmatchP[G]$ the corresponding (bipartite) matching polytope.
Since one can show that inequality~\eqref{eq:GuruganeshLee} %
can be tight, the analysis of the balancedness for the scheme presented in Example~\ref{ex:CRscheme_from_marginals} is tight as well.
New ideas are therefore needed to obtain a CR scheme for $\bipmatchP[G]$ with a stronger balancedness.

In the following, we first examine why the scheme from Example~\ref{ex:CRscheme_from_marginals} does not have a higher balancedness. This helps to build up intuition on how to adjust the scheme. Our polyhedral viewpoint then leads to a natural formalization of this intuition and allows for analyzing the resulting scheme. Moreover, our approach also allows for giving a simple proof that the scheme we get is optimal among all monotone CR schemes for bipartite matchings.

\subsection{Developing the scheme}

As shown by Guruganesh and Lee~\cite{GuruganeshLee2017}, tightness of inequality~\eqref{eq:GuruganeshLee} for an edge $e=\{u,v\} \in E$ is approached for the following graph structure and $x$-values:
The edge $e$ has $x$-value $x_e = \epsilon$ for tiny $\epsilon > 0$, $u$ has only one other edge $e_0$ incident to it with $x_{e_0} = 1 - \epsilon$, and $v$ has $k$ other edges $e_{1}, \ldots, e_{k}$ incident to it with $x_{e_i} = \frac{1-\epsilon}{k}$ for $i \in [k]$, where $k$  is large.  An illustration of this is provided in Figure~\ref{fig:CRSbipartiteBadPoisBern}.
\begin{figure}[ht]
	\centering

\tikzstyle{vertex}=[circle, draw=black, thick, inner sep=1pt, minimum size=2mm, fill=black!10]
\tikzstyle{edgenode}=[rectangle, fill=white, fill opacity=0.8, text opacity=1, inner sep=1pt]
\tikzstyle{edge}=[thick, draw=black]
\tikzstyle{selectedEdge}=[ultra thick, draw=\colorofset]
\tikzstyle{discardedEdge}=[semithick, dashed, draw=black]
\tikzstyle{arc}=[thick, draw=black!60, -latex]
\tikzstyle{boxed}=[thick, draw=black!60, rounded corners]
\tikzstyle{boxedFat}=[thick, draw=black, rounded corners]
\tikzstyle{boxedLight}=[thick, draw=black, rounded corners, fill = black!15, minimum height=0.75*1cm]

\begin{tikzpicture}[scale=0.75]

\node[vertex, minimum size=10/0.75, font=\small] (u) at (-1.75,-1) {$u$};
\node[vertex, minimum size=10/0.75, font=\small] (v) at (1.75,1) {$v$};

\coordinate (u0) at ($(u) + (-3,1.0)$) {};
\coordinate (u1) at ($(u) + (-3,1.5)$) {};
\coordinate (u2) at ($(u) + (-3,2.5)$) {};
\coordinate (u4) at ($(u) + (-3,0.5)$) {};
\coordinate (u5) at ($(u) + (-3,-0.5)$) {};

\coordinate (v0) at ($(v) + (3,-1.0)$) {};
\coordinate (v1) at ($(v) + (3,-1.5)$) {};
\node at ($(v1)+(-0.175,0.225)$) {$\vdots$};
\coordinate (v2) at ($(v) + (3,-2.5)$) {};
\coordinate (v4) at ($(v) + (3,-0.5)$) {};
\coordinate (v5) at ($(v) + (3,0.5)$) {};

\draw[edge] (u) edge node[font=\small, color=\colorofx, edgenode, pos=0.5] {$\epsilon$} (v);

\draw[edge] (u) edge node[font=\small, color=\colorofx, edgenode, pos=0.7] {$1-\epsilon$} ($(u)!0.85!(u0)$);
\draw[edge, dotted] ($(u)!0.85!(u0)$) edge (u0);

\draw[edge] (v) edge node[font=\small, color=\colorofx, edgenode, pos=0.7] {$\frac{1-\epsilon}{k}$} ($(v)!0.85!(v2)$);
\draw[edge, dotted] ($(v)!0.85!(v2)$) edge (v2);
\draw[edge] (v) edge node[font=\small, color=\colorofx, edgenode, pos=0.7] {$\frac{1-\epsilon}{k}$} ($(v)!0.85!(v4)$);
\draw[edge, dotted] ($(v)!0.85!(v4)$) edge (v4);
\draw[edge] (v) edge node[font=\small, color=\colorofx, edgenode, pos=0.7] {$\frac{1-\epsilon}{k}$} ($(v)!0.85!(v5)$);
\draw[edge, dotted] ($(v)!0.85!(v5)$) edge (v5);
\draw [decorate,decoration={brace,amplitude=5, mirror}]
($(v2)+(0.25,-0.1)$) -- node[font=\small,rotate=270, yshift=12.5] {$k$ edges} ($(v5)+(0.25,0.1)$);

\end{tikzpicture} 	\caption{Bad instance for the CR scheme described in Example~\ref{ex:CRscheme_from_marginals}.}\label{fig:CRSbipartiteBadPoisBern}
\end{figure}
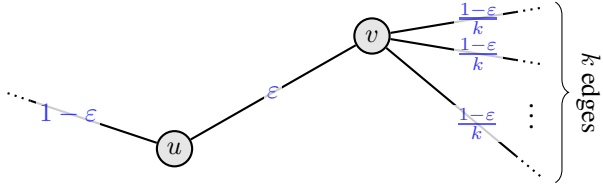

In particular, for~\eqref{eq:GuruganeshLee} to be (close to) tight %
 for some edge $e \in E$, the edge $e$ needs to be adjacent to an edge $g \in E$ that has $x$-value close to $1$.
For any such edge $g \in E$ with large $x$-value, however, it is not difficult to see that the expression in~\eqref{eq:GuruganeshLee} %
 is considerably larger than $0.4481$. 
Since the balancedness of a CR scheme is a worst-case measure over all edges, one would suspect that truly hard instances exhibit more symmetry and do not have edges with large $x$-value.
It therefore seems plausible that one should be able to do better on the above described instances by balancing things out, meaning that it is only a feature of the specific scheme given in Example~\ref{ex:CRscheme_from_marginals} that causes it to have problems rather than an inherent difficulty of these instances.
Indeed, the following lemma shows that the hardest instances %
 (independent of a specific scheme) are actually those where \emph{all} edges have small $x$-value.
\begin{lemma}\label{lem:CRscheme_smallx_bipmatch}
	Let $c \in [0,1]$ and $\epsilon > 0$. If there is a monotone contention resolution scheme for the matching polytope $\bipmatchP[G]$ of any bipartite graph $G=(V,E)$ that is $c$-balanced for all $x \in \bipmatchP[G]$ with $x \leq \epsilon \cdot \chi^E$, then there exists a monotone $c$-balanced contention resolution scheme for the matching polytope of any bipartite graph (without any restriction).
\end{lemma}
A generalized version of this statement that goes beyond (bipartite) matchings is given in Lemma~\ref{lem:CRscheme_smallx_general} in Appendix~\ref{app:CRscheme_smallx}. There, we also provide a full proof.
Nevertheless, also here we briefly want to give some intuition why the statement holds since this motivates the construction of the procedure we describe below.
If we are given a point $x \in \bipmatchP[G]$ such that there are edges with large $x$-value, the idea is to do the following.
\begin{enumerate}[label=\normalfont\arabic*., leftmargin=4em]
	\item Let $k \in \mathbb{Z}_{> 0}$ such that $\frac{1}{k} \leq \epsilon$.
	\item For every edge $e \in E$, split $e$ into $k$ copies $e_1, \ldots, e_k$ and distribute the value $x_e$ uniformly among these $k$ edges, i.e., each copy gets an $x$-value of $\frac{x_e}{k} \leq \epsilon$.
	\item Independently round the new point $x$ and apply the scheme for points with small components to the resulting set $R(x)$ to get a bipartite matching $I_x$ in the graph with the split edges.
	\item For each edge $e \in E$, replace any occurrence of a copy of $e$ in $I_x$ by $e$ (observe that at most one of $e_1, \ldots, e_k$ can be in $I_x$).
\end{enumerate}
Notice that the above yields a bipartite matching in the original graph $G$.
The issue with this approach, however, is that it will typically not be computationally efficient.
We can circumvent this problem by simulating the splitting without actually having to do it.
This is achieved by the scheme for $\bipmatchP[G]$ we describe next.

Let $x \in \bipmatchP[G]$ and $A \subseteq \supp(x)$. By Proposition~\ref{prop:marginalCRschemes}, all we have to do is to find an (efficient) procedure that returns a vector $y_x^A \in [0,1]^E$ which satisfies~\ref{item:marginalSupp}--\ref{item:marginalMonotone}.
We start by constructing a random subset $\overline{A} \subseteq A$, where each edge $e \in A$ is included in $\overline{A}$ independently with probability $(1-\myexp^{-x_e})/ x_e$.
Notice that $(1-\myexp^{-x_e})/ x_e$ is almost $1$ for small $x_e \in [0,1]$ but approaches $1-\myexp^{-1} = 0.63212\ldots$ for large $x_e \in [0,1]$. %
While edges with small $x$-value
are therefore almost unaffected by this step, 
edges with large $x$-value get penalized.
Next, we define a vector $q \in \mathbb{Z}_{\geq 0}^E$ as follows. For each edge $e \in \overline{A}$, we let $q_e \in \mathbb{Z}_{\geq 1}$ be an independent realization of a $\PoisD{x_e}$-random variable conditioned on it being at least $1$, and for each edge $e \in E \setminus \overline{A}$, we let $q_e \coloneqq 0$.
Then, for each edge $e = \{u,v\} \in E$, we set 
\begin{align*}
	(y_x^A)_e \coloneqq \frac{ q_e }{ \max\{  \sum_{g \in \delta(u)} q_g , \sum_{g \in \delta(v)} q_g  \} }  
	\enspace .
\end{align*}
Throughout this work, we use the convention that the latter is equal to $0$ if $q_e = 0$, even if the denominator also equals $0$. %
A formal description of this scheme is provided below (Algorithm~\ref{algo:BipMatchCRS}).
Clearly, we have that $y_x^A \in \bipmatchP[G]$ and $\supp(y_x^A) \subseteq A$. %
Moreover, we will show later that via Proposition~\ref{prop:marginalCRschemes}, this procedure yields the monotone $(b, \beta(b))$-balanced CR scheme for $\bipmatchP[G]$ whose existence is claimed in Theorem~\ref{thm:mCRs_existence_bipartite}.

\begin{algorithm2e}[ht]
	\caption{Contention resolution for $\bipmatchP[G]$}
	\label{algo:BipMatchCRS}
	\KwIn{Point $x \in \bipmatchP[G]$, set $A \subseteq \supp(x)$.}
		\KwOut{Point $y_x^A \in \bipmatchP$ with $\supp(y_x^A) \subseteq A$.} %
	Sample $\overline{A} \subseteq A$ by including edge $e \in A$ independently with probability  \smash{$ (1- \myexp^{-x_e})/ x_e$}\;
	For each edge $e \in \overline{A}$, let $q_e \in \mathbb{Z}_{\geq 1}$ be an independent realization of a $\PoisD{x_e}$-random variable conditioned on it being at least $1$\;
	For each edge $e \in E \setminus \overline{A}$, let $q_e \coloneqq 0$\;
	For each edge $e = \{u,v\} \in E$, let 
	\smash{$(y_x^A)_e \coloneqq \frac{ q_e }{ \max\{  \sum_{g \in \delta(u)} q_g , \sum_{g \in \delta(v)} q_g  \} }  $}\;
	Return $y_x^A$\;
\end{algorithm2e}

While a short motivation for the subsampling step was already given above (penalization of edges with large $x$-value), we also want to explain why the use of (conditioned) %
Poisson random variables is quite natural. %
As was mentioned before, the instances which cause most difficulties for the scheme given in Example~\ref{ex:CRscheme_from_marginals} are those where edges with large $x$-value are present (see also Figure~\ref{fig:CRSbipartiteBadPoisBern}). However, Lemma~\ref{lem:CRscheme_smallx_bipmatch} tells us that by splitting edges we can avoid such instances.
In this context, given $x \in \bipmatchP[G]$, the splitting would work as follows.
\begin{enumerate}[label=\normalfont\arabic*., leftmargin=4em]
	\item Let $k \in \mathbb{Z}_{> 0}$ be large.
	\item For every edge $e \in E$, split $e$ into $k$ copies $e_1, \ldots, e_k$ and distribute the value $x_e$ uniformly among these $k$ edges, i.e., each copy gets an $x$-value of $\frac{x_e}{k}$.
	\item Independently round the new point $x$ and apply the scheme from Example~\ref{ex:CRscheme_from_marginals} to the resulting set $R(x)$ to get a point $y_x^{R(x)}$ in the bipartite matching polytope of the graph with the split edges.
	\item For each edge $e \in E$, reverse the splitting by reuniting the edges $e_1, \ldots, e_k$ and summing up their values in $y_x^{R(x)}$.	
\end{enumerate}
Notice that the above yields a point $y_x^{R(x)} \in \bipmatchP[G]$ for the original graph $G$ since for every edge $e \in E$, all the copies $e_1, \ldots, e_k$ have the same endpoints as $e$. Moreover, for the determination of this final point $y_x^{R(x)}$ it is irrelevant \emph{which} copies $e_1, \ldots, e_k$ of a given edge $e \in E$ appeared in the independent rounding step; all that matters is how many of them appeared.
By construction, this number follows a $\BinD{k}{\frac{x_e}{k}}$-distribution, which, for $k \to \infty$, converges to a $\PoisD{x_e}$-distribution.
Hence, this is where the %
Poisson random variables naturally emerge. By using them directly, we can avoid the splitting %
 while still getting the same (improved) results we would get by actually doing it. In addition, it makes our analysis much shorter and easier to follow.

Moreover, the subsampling step and the subsequent use of conditioned $\PoisD{x_e}$-random variables in Algorithm~\ref{algo:BipMatchCRS} fit together very well. To see why, let $x \in \bipmatchP[G]$ be a point we want to round with the help of the scheme from Algorithm~\ref{algo:BipMatchCRS}.
First, we independently round $x$ to get a set $R(x) \subseteq \supp(x)$ and then construct a subset $\overline{R} (x) \subseteq R(x)$ by including each edge $e \in R(x)$ independently with probability $(1- \myexp^{-x_e})/x_e$. Thus, we have that a given edge $e \in E$ appears in the random set $\overline{R} (x)$ independently with probability
\begin{equation*}
\Pr [ e \in \overline{R}(x)] = \Pr[ e \in \overline{R}(x) \mid e \in R(x)] \cdot \Pr[e \in R(x)] 
= \frac{1-\myexp^{-x_e}}{x_e} \cdot x_e = 1 - \myexp^{-x_e}
\enspace .
\end{equation*}
The above exactly matches the probability that a $\PoisD{x_e}$-random variable is at least $1$ (or, equivalently, does not equal $0$).
Therefore, independently rounding a point $x \in \bipmatchP[G]$ and then applying Algorithm~\ref{algo:BipMatchCRS} to the resulting (random) set $R(x)$ actually amounts to do the following.\footnote{%
	Notice that Algorithm~\ref{algo:BipMatchCRSshort} is not a procedure that can directly be transformed into a contention resolution scheme for $\bipmatchP[G]$ via Proposition~\ref{prop:marginalCRschemes}.
}
\begin{algorithm2e}[ht]	
	\caption{Combining independent rounding with Algorithm~\ref{algo:BipMatchCRS}}
	\label{algo:BipMatchCRSshort}
	\KwIn{Point $x \in \bipmatchP[G]$.}
	\KwOut{Point $y_x \in \bipmatchP$ with $\supp(y_x) \subseteq \supp(x)$.}
	For each edge $e \in E$, let $q_e \in \mathbb{Z}_{\geq 0}$ be an independent realization of a $\PoisD{x_e}$-random variable\;
	For each edge $e = \{u,v\} \in E$, let 
	\smash{$(y_x)_e \coloneqq \frac{ q_e } { \max\{  \sum_{g \in \delta(u)} q_g , \sum_{g \in \delta(v)} q_g  \}  } $}\;
	Return $y_x$\;
\end{algorithm2e}

It is important to note here that the vectors $y_x^A \in \bipmatchP[G]$, which are constructed by the scheme described in Algorithm~\ref{algo:BipMatchCRS}, are random.
Thus, they only correspond to ``conditional'' marginals (conditioned on the set $\overline{A} \subseteq A$ and the values of $q_e \in \mathbb{Z}_{\geq 1}$ for $e \in \overline{A}$).
The true marginals of the scheme are given by $\expval[ y_x^A ] \in \bipmatchP[G]$.
Because we do not know of an elegant closed-form expression for these expectations, however, we prefer to work with the conditional marginals $y_x^A$. %
This is simpler both from a conceptual as well as from a computational point of view. 
Moreover, our polyhedral approach is flexible enough to be applicable in this setting as well.
Indeed, using it allows us to give an easy proof that the procedure given in Algorithm~\ref{algo:BipMatchCRS} satisfies~\ref{item:marginalSupp}--\ref{item:marginalMonotone} and can therefore, via Proposition~\ref{prop:marginalCRschemes}, be efficiently transformed into a monotone $(b, \beta(b))$-balanced\footnote{%
We recall that $\beta(b) \coloneqq  \expval \big[ \frac{1}{1 +  \max\{ \PoisD[1]{b}, \PoisD[2]{b} \}  }  \big]$ for $b \in [0,1]$, where $\PoisD[1]{b}$ and $\PoisD[2]{b}$ are two independent Poisson random variables with parameter $b$. For $b=1$, it holds $\beta(1) \geq 0.4762$.%
}
contention resolution scheme for $\bipmatchP[G]$.

\begin{proof}[Proof of Theorem~\ref{thm:mCRs_existence_bipartite}]
Consider the procedure described in Algorithm~\ref{algo:BipMatchCRS}.
	Given $x \in \bipmatchP[G]$ and $A \subseteq \supp(x)$, the construction of $y_x^A$ clearly ensures that $\supp(y_x^A) \subseteq A$, meaning that~\ref{item:marginalSupp} is satisfied.
	Moreover, conditioning on the set $\overline{A} \subseteq A$ and the vector $q \in \mathbb{Z}^E_{\geq 0}$, it holds for any vertex $v \in V$ that
	\begin{equation*}
	y_x^A (\delta(v)) 
	= \sum_{e \in \delta(v)} (y_x^A)_e
	\leq \sum_{e \in \delta(v)} \frac{q_e}{ \sum_{g \in \delta(v)} q_g }
	\leq	1
	\enspace .\footnote{Notice that we interpret, as previously discussed, the fraction $\frac{q_e}{\sum_{g\in \delta(v)} q_g}$ as $0$ if both numerator and denominator are $0$. Due to this, the last inequality may be strict, and can therefore not be replaced by an equality.}
	\end{equation*}
	Since it also holds that $y_x^A \geq 0$, the above implies $y_x^A \in \bipmatchP[G]$. Thus,~\ref{item:marginalConvComb} is fulfilled as well.
	As noted before, the vectors $y_x^A \in \bipmatchP[G]$ that are constructed by our procedure are random, and therefore correspond to conditional marginals. For $x \in \bipmatchP[G]$ and $A \subseteq \supp(x)$, the true marginals %
	 are given by $\expval[ y_x^A ]$, where the expectation is over the random set $\overline{A} \subseteq A$ and the realizations $q_e$ of the (conditioned) Poisson random variables. Since the vectors $y_x^A$ satisfy~\ref{item:marginalSupp} and~\ref{item:marginalConvComb}, so do the true marginals $\expval[ y_x^A ]$.
	We now show monotonicity of our procedure by proving that the true marginals satisfy~\ref{item:marginalMonotone}. To this end, let $x \in \bipmatchP[G]$ and $e=\{u,v\} \in A \subseteq B \subseteq \supp(x)$.
	Moreover, for each edge $g \in E$, we let $Q_g$ be an independent $\PoisD{x_g}$-random variable conditioned on it being at least $1$. 
	Since it is possible to couple the random sets $\overline{A} \subseteq A$ and $\overline{B} \subseteq B$, which we get in the subsampling step of Algorithm~\ref{algo:BipMatchCRS}, such that $\overline{A} = \overline{B}\cap A \subseteq \overline{B}$, we conclude that
	\begin{align*}
	\expval[(y_x^{A})_e] 
	&= \Prb[ e \in \overline{A} ]  \cdot 
	\expval \bigg[
	\frac{Q_e}{ \max\{  \sum_{g \in \delta(u)\cap \overline{A}} Q_g , \sum_{g \in \delta(v)\cap \overline{A}} Q_g  \} }  \,\bigg\vert\, e \in \overline{A}	\bigg]   \\
	&\geq \Prb[ e \in \overline{B} ]  \cdot 
	\expval \bigg[
	\frac{Q_e}{ \max\{  \sum_{g \in \delta(u)\cap \overline{B}} Q_g , \sum_{g \in \delta(v)\cap \overline{B}} Q_g  \} }   \,\bigg\vert\, e \in \overline{B}	\bigg]   \\
	&= \expval[(y_x^{B})_e] 
	\enspace ,
	\end{align*}
	where the expectation is taken over the random sets $\overline{A} \subseteq A$ and $\overline{B} \subseteq B$ as well as the random variables $Q_g$, $g \in E$.
	Hence, our procedure also satisfies~\ref{item:marginalMonotone}.
	
	To prove the claimed balancedness, %
	 fix $b \in [0,1]$, $x \in b \bipmatchP[G]$, and $e = \{u,v\} \in  \supp (x)$. From the discussion preceding this proof, we get that	
	the distribution of the random vector $y_x^{R(x)}$ constructed by Algorithm~\ref{algo:BipMatchCRS} for the independently rounded set $R(x)$ is the same as the distribution of the random vector $y_x$ constructed by Algorithm~\ref{algo:BipMatchCRSshort}.
	Hence, if we can show that 
	\begin{equation}\label{eq:BipMatchCRSshortBalanc}
		\expval [  (y_x)_e ]
		\geq \beta(b) \cdot x_e  
		\enspace ,
	\end{equation}
	then~\ref{item:marginalBalance} is fulfilled as well.
	It therefore only remains to show~\eqref{eq:BipMatchCRSshortBalanc}. 
	For this purpose, let $Q_g$ for $g \in E$ be independent random variables, where for each edge $g \in E$,  $Q_g$ follows a $\PoisD{x_g}$-distribution.
	We then have %
	\begin{align*}
		\expval[ (y_x)_e  ]  
		=    \expval \bigg[  \frac{Q_e}{\max\{  \sum_{g \in \delta(u)} Q_g , \sum_{g \in \delta(v)} Q_g  \}}   \bigg]
		=    \expval \bigg[  \frac{Q_e}{ Q_e + \max\{  \sum_{g \in \delta(u) \setminus \{e\}} Q_g , \sum_{g \in \delta(v) \setminus \{e\} } Q_g  \}}   \bigg]
		\enspace .
	\end{align*}
	In order to analyze the above expression, we use the following simplified notation. We let $\BernD[i]{\xi}$ denote a Bernoulli random variable with parameter $\xi$ and $\PoisD[j]{\zeta}$ a Poisson random variable with parameter $\zeta$, where we assume that all random variables that appear are independent when they have different subscripts. 
	Note that $x \in b \bipmatchP[G]$ implies $x(\delta(w)) \leq b$ for any vertex $w \in V$.
	Moreover, recall that the sum of two independent Poisson random variables $\PoisD[1]{\xi}$ and $\PoisD[2]{\zeta}$ has a $\PoisD{\xi+\zeta}$-distribution and that	
	$\sum_{i=1}^k \BernD[i]{\frac{\xi}{k}}$ approaches a $\PoisD{\xi}$-distribution as $k \to \infty$. Using these properties as well as dominated convergence and linearity of expectation, 
	 we get that
	\begin{align*}
	\expval[ (y_x)_e  ]  
	&=    \expval \bigg[  \frac{ \PoisD[e]{x_e}  }{ \PoisD[e]{x_e} + \max\{  \sum_{g \in \delta(u) \setminus \{e\}} \PoisD[g]{x_g} , \sum_{g \in \delta(v) \setminus \{e\} } \PoisD[g]{x_g}  \}}   \bigg]   \\
	& \geq  \expval \bigg[  \frac{ \PoisD[e]{x_e}  }{ \PoisD[e]{x_e} + \max\{  \PoisD[u]{b - x_e} ,  \PoisD[v]{b-x_e}  \}}   \bigg]   \\	
	&= \expval \bigg[  \lim_{k \to \infty}  \frac{ \sum_{i=1}^k \BernD[i]{\frac{x_e}{k}}  }{ \sum_{i=1}^k \BernD[i]{\frac{x_e}{k}} +  \max\{  \PoisD[u]{b - x_e} ,  \PoisD[v]{b-x_e}  \}}   \bigg]  \\
	&=  \lim_{k \to \infty}  \sum_{i=1}^k \expval \bigg[  \frac{ \BernD[i]{\frac{x_e}{k}}  }{ \BernD[i]{\frac{x_e}{k}} +  \sum_{j \in [k] \setminus \{i\}} \BernD[j]{\frac{x_e}{k}} +  \max\{  \PoisD[u]{b - x_e} ,  \PoisD[v]{b-x_e}  \}}   \bigg]  \\
	&=  \lim_{k \to \infty}  \sum_{i=1}^k   \frac{x_e}{k} \cdot  \expval \bigg[  \frac{ 1  }{ 1+ \sum_{j \in [k] \setminus \{i\}} \BernD[j]{\frac{x_e}{k}} +  \max\{  \PoisD[u]{b - x_e} ,  \PoisD[v]{b-x_e}  \}}   \bigg]	\\
	& \geq \lim_{k \to \infty}  \sum_{i=1}^k   \frac{x_e}{k} \cdot  \expval \bigg[  \frac{ 1  }{ 1+ \sum_{j=1}^k \BernD[j]{\frac{x_e}{k}} +  \max\{  \PoisD[u]{b - x_e} ,  \PoisD[v]{b-x_e}  \}}   \bigg] \\
	& = x_e \cdot   \expval \bigg[ \lim_{k \to \infty}   \frac{ 1  }{ 1+ \sum_{j=1}^k \BernD[j]{\frac{x_e}{k}} +  \max\{  \PoisD[u]{b - x_e} ,  \PoisD[v]{b-x_e}  \}}   \bigg] \\
	& = x_e \cdot   \expval \bigg[   \frac{ 1  }{ 1+ \PoisD[e]{x_e} +  \max\{  \PoisD[u]{b - x_e} ,  \PoisD[v]{b-x_e}  \}}   \bigg] \\
	& \geq x_e \cdot   \expval \bigg[   \frac{ 1  }{ 1+   \max\{  \PoisD[1]{b} ,  \PoisD[2]{b}  \}}   \bigg] \\
	& = x_e \cdot \beta(b)
	\enspace ,
	\end{align*}
	proving~\eqref{eq:BipMatchCRSshortBalanc}.
The last inequality holds since the distribution of $\PoisD[a]{\xi} + \max \{ \PoisD[b]{\zeta}, \PoisD[c]{\zeta}  \} $ is stochastically dominated by the distribution of $ \max \{\PoisD[a_1]{\xi}+  \PoisD[b]{\zeta}, \PoisD[a_2]{\xi}+ \PoisD[c]{\zeta}  \}$. A formal statement of this claim is given in Lemma~\ref{lem:stochDom} in Appendix~\ref{app:stochDom}.
Moreover, note that if there was another edge $h \in E \setminus \{e\}$ with the same endpoints as $e = \{u,v\}$, not all random variables in the expression $\max\{  \sum_{g \in \delta(u) \setminus \{e\}} \PoisD[g]{x_g} , \sum_{g \in \delta(v) \setminus \{e\} } \PoisD[g]{x_g}  \}$ would be independent. However, Lemma~\ref{lem:stochDom} (i.e., the statement mentioned above) also shows that we can assume independence of all these random variables since this only makes it more difficult to achieve a high balancedness.%

Since the procedure given in Algorithm~\ref{algo:BipMatchCRS} satisfies~\ref{item:marginalSupp}--\ref{item:marginalMonotone}, Proposition~\ref{prop:marginalCRschemes} states that we can efficiently transform it into a monotone $(b, \beta(b))$-balanced CR scheme for $\bipmatchP[G]$, which proves Theorem~\ref{thm:mCRs_existence_bipartite}.
\end{proof}

As mentioned earlier, the procedure described in Algorithm~\ref{algo:BipMatchCRS} avoids problems with instances that have edges with large $x$-value by simulating the splitting of such edges.
For instances where all edges have small $x$-value, this is not necessary.
Indeed, using Algorithm~\ref{algo:BipMatchCRS} for such instances basically amounts to the same as applying the very simple procedure described in Example~\ref{ex:CRscheme_from_marginals}.

\subsection{Optimality}

We now use our polyhedral approach %
 to prove that there is no \emph{monotone} contention resolution scheme for the bipartite matching polytope that achieves a higher balancedness than the monotone scheme we get from the procedure described in Algorithm~\ref{algo:BipMatchCRS}.
This shows that our CR scheme for bipartite matchings is optimal when requiring monotonicity, which is crucial when using contention resolution through Theorem~\ref{thm:crSchemeGuarantee} in the context of~\ref{eq:CSFM}.

\begin{proof}[Proof of Theorem~\ref{thm:mCRs_optimality_bipartite}]
	Let $\pi$ be a monotone $(b,c)$-balanced CR scheme for the bipartite matching polytope.
Moreover, for $n \in \mathbb{Z}_{>0}$ large, let $K_{n,n}=(U \dcup V,E)$ be the complete bipartite graph on $2n$ vertices and define $x \in [0,1]^E$ by $x_e \coloneqq \frac{b}{n}$ for each edge $e \in E$.  Clearly, it holds that $x \in b\bipmatchP[K_{n,n}]$. 
In Figure~\ref{fig:CRSbipartiteBadPoisPois}, it is shown how this instance looks like from the perspective of a fixed edge $e = \{u,v\}$.
The reason why we consider such an instance is Lemma~\ref{lem:CRscheme_smallx_bipmatch} which implies that the hardest instances are those where all edges have small $x$-value.  
\begin{figure}[ht]
	\centering

\tikzstyle{vertex}=[circle, draw=black, thick, inner sep=1pt, minimum size=2mm, fill=black!10]
\tikzstyle{edgenode}=[rectangle, fill=white, fill opacity=0.8, text opacity=1, inner sep=1pt]
\tikzstyle{edge}=[thick, draw=black]
\tikzstyle{selectedEdge}=[ultra thick, draw=\colorofset]
\tikzstyle{discardedEdge}=[semithick, dashed, draw=black]
\tikzstyle{arc}=[thick, draw=black!60, -latex]
\tikzstyle{boxed}=[thick, draw=black!60, rounded corners]
\tikzstyle{boxedFat}=[thick, draw=black, rounded corners]
\tikzstyle{boxedLight}=[thick, draw=black, rounded corners, fill = black!15, minimum height=0.75*1cm]

\begin{tikzpicture}[scale=0.75]

\node[vertex, minimum size=10/0.75, font=\small] (u) at (-1.75,-1) {$u$};
\node[vertex, minimum size=10/0.75, font=\small] (v) at (1.75,1) {$v$};

\coordinate (u0) at ($(u) + (-3,1.0)$) {};
\coordinate (u1) at ($(u) + (-3,1.5)$) {};
\coordinate (u2) at ($(u) + (-3,2.5)$) {};
\coordinate (u4) at ($(u) + (-3,0.5)$) {};
\node at ($(u4)+(0.175,0.1)$) {$\vdots$};
\coordinate (u5) at ($(u) + (-3,-0.5)$) {};

\coordinate (v0) at ($(v) + (3,-1.0)$) {};
\coordinate (v1) at ($(v) + (3,-1.5)$) {};
\node at ($(v1)+(-0.175,0.225)$) {$\vdots$};
\coordinate (v2) at ($(v) + (3,-2.5)$) {};
\coordinate (v4) at ($(v) + (3,-0.5)$) {};
\coordinate (v5) at ($(v) + (3,0.5)$) {};

\draw[edge] (u) edge node[font=\small, color=\colorofx, edgenode, pos=0.5] {$\sfrac{b}{n}$} (v);

\draw[edge] (u) edge node[font=\small, color=\colorofx, edgenode, pos=0.7] {$\sfrac{b}{n}$} ($(u)!0.85!(u1)$);
\draw[edge, dotted] ($(u)!0.85!(u1)$) edge (u1);
\draw[edge] (u) edge node[font=\small, color=\colorofx, edgenode, pos=0.7] {$\sfrac{b}{n}$} ($(u)!0.85!(u2)$);
\draw[edge, dotted] ($(u)!0.85!(u2)$) edge (u2);
\draw[edge] (u) edge node[font=\small, color=\colorofx, edgenode, pos=0.7] {$\sfrac{b}{n}$} ($(u)!0.85!(u5)$);
\draw[edge, dotted] ($(u)!0.85!(u5)$) edge (u5);

\draw[edge] (v) edge node[font=\small, color=\colorofx, edgenode, pos=0.7] {$\sfrac{b}{n}$} ($(v)!0.85!(v2)$);
\draw[edge, dotted] ($(v)!0.85!(v2)$) edge (v2);
\draw[edge] (v) edge node[font=\small, color=\colorofx, edgenode, pos=0.7] {$\sfrac{b}{n}$} ($(v)!0.85!(v4)$);
\draw[edge, dotted] ($(v)!0.85!(v4)$) edge (v4);
\draw[edge] (v) edge node[font=\small, color=\colorofx, edgenode, pos=0.7] {$\sfrac{b}{n}$} ($(v)!0.85!(v5)$);
\draw[edge, dotted] ($(v)!0.85!(v5)$) edge (v5);

\draw [decorate,decoration={brace,amplitude=5}]
($(u5)+(-0.25,-0.1)$) -- node[font=\small,rotate=90, yshift=12.5] {$n-1$ edges} ($(u2)+(-0.25,0.1)$);
\draw [decorate,decoration={brace,amplitude=5, mirror}]
($(v2)+(0.25,-0.1)$) -- node[font=\small,rotate=270, yshift=12.5] {$n-1$ edges} ($(v5)+(0.25,0.1)$);

\end{tikzpicture} 	\caption{Hard instance for all monotone CR schemes for the bipartite matching polytope.}%
	\label{fig:CRSbipartiteBadPoisPois}
\end{figure}
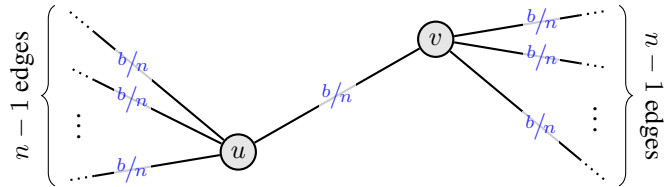

Observe that $K_{n,n}$ and $x \in b\bipmatchP[K_{n,n}]$ are completely symmetric with respect to permutations of $U$ and permutations of $V$. Thus, we can assume that the same is true for $\pi_{x}$.\footnote{Otherwise, we symmetrize $\pi_{x}$. This preserves monotonicity and does not decrease the balancedness of $\pi$.}
Moreover, since there is nothing to show if $b=0$, we only consider the case $b>0$. In particular, this means that $\supp(x) = E$.

We now study the marginals $y_{x}^A \coloneqq \expval [\chi^{ \pi_{x}(A) }  ]$ for $A \subseteq E$,
and claim that for every $e = \{u,v\} \in E$ and $A \subseteq E$ with $e \in A$, it holds that
\begin{equation}\label{eq:BipMatchCRSoptimal_main}
	(y_{x}^A)_e \leq \frac{1}{\max\{  |\delta(u) \cap A|, |\delta(v) \cap A|  \}}
	\enspace .
\end{equation}
By~\ref{item:marginalBalance}, this implies
\begin{equation}\label{eq:BipMatchCRSoptimal_aux}
\begin{aligned}
	c \leq
	\frac{\expval[   (y_{x}^{R(x)})_e   ] }{x_e} 
	& = \expval[   (y_{x}^{R(x)})_e \mid e \in R(x)  ] \\
	& \leq \expval\bigg[  \frac{1}{\max\{  |\delta(u) \cap R(x)|, |\delta(v) \cap R(x)| \}  }   \,\Big\vert\, e \in R(x)        \bigg]  \\
	& = \expval\bigg[   \frac{1}{ 1 + \max\{  |\delta(u) \cap (R(x)\setminus \{e\})|, |\delta(v) \cap (R(x)\setminus \{e\})|  \}} \bigg] \\
	& = \expval\bigg[  \frac{1}{1 +  \max\{  \BinD[1]{n-1}{\frac{b}{n}} , \BinD[2]{n-1}{\frac{b}{n}}    \}  }          \bigg]  
	\enspace ,
\end{aligned}
\end{equation}
where $\BinD[1]{n-1}{\frac{b}{n}}$ and $\BinD[2]{n-1}{\frac{b}{n}}$ are two independent $\BinD{n-1}{\frac{b}{n}}$-random variables.
Here, we used that each edge $g \in E$ is included in $R(x)$ with probability $x_g = \frac{b}{n}$, independently of all other edges.
Since~\eqref{eq:BipMatchCRSoptimal_aux} is true for any $n \in \mathbb{Z}_{>0}$, we can use that a $\BinD{n-1}{\frac{b}{n}}$-distribution converges to a $\PoisD{b}$-distribution as $n \to \infty$ to get that
\begin{align*}
c 
&\leq 
\lim_{n \to \infty}
 \expval\bigg[  \frac{1}{1 +  \max\{  \BinD[1]{n-1}{\frac{b}{n}} , \BinD[2]{n-1}{\frac{b}{n}}    \}  }          \bigg]  \\
 &=
 \expval\bigg[ \lim_{n \to \infty}  \frac{1}{1 +  \max\{  \BinD[1]{n-1}{\frac{b}{n}} , \BinD[2]{n-1}{\frac{b}{n}}    \}  }          \bigg]  \\
 &= \expval\bigg[  \frac{1}{1 +  \max\{  \PoisD[1]{b} , \PoisD[2]{b}    \}  }          \bigg]  \\
 &= \beta(b)
\enspace ,
\end{align*}
where $\PoisD[1]{b}$ and  $\PoisD[2]{b}$ are two independent $\PoisD{b}$-random variables. This is exactly what we want to show. Hence, in order to finish the proof of Theorem~\ref{thm:mCRs_optimality_bipartite}, it only remains to show~\eqref{eq:BipMatchCRSoptimal_main}.

To do so, let $e = \{u,v\} \in A \subseteq E$, where $u \in U$ and $v \in V$. %
Defining $B \coloneqq \delta(u) \cap A$, the fact that $y_{x}^B \in  \bipmatchP[K_{n,n}]$ implies
\begin{equation*}
\sum_{g \in B} (y_{x}^B)_g \leq \sum_{g \in \delta(u)} (y_{x}^B)_g = (y_{x}^B)(\delta(u))  \leq 1	   \enspace .\footnote{Even though this is not relevant for our derivations, we highlight that the inequality $\sum_{g\in B}(y_x^B)_g \leq \sum_{g\in \delta(u)}(y_x^B)_g$ actually even holds with equality because $\supp(y_x^B)\subseteq B$.}
\end{equation*}
We now use our assumption on the symmetry of $\pi_x$ to get that  $(y_{x}^B)_e =  (y_{x}^B)_g$ for each $g \in B$ (note that for any two edges $g_1, g_2 \in B \subseteq \delta(u)$, there is a permutation of $N_B(u) \coloneqq \{w \in V  \mid \{u,w\} \in B \} \subseteq V$ that leaves $B$ invariant and maps $g_1$ to $g_2$). Moreover, monotonicity of $\pi$ implies by~\ref{item:marginalMonotone} that the marginals satisfy $(y_{x}^B)_e \geq (y_{x}^A)_e$. Putting everything together, this means that
\begin{equation*}
1 \geq	
\sum_{g \in B} (y_{x}^B)_g
= |B| \cdot  (y_{x}^B)_e
\geq |B| \cdot  (y_{x}^A)_e
= |\delta(u) \cap A| \cdot (y_{x}^A)_e
\enspace .
\end{equation*}
Doing the same for $v$ instead of $u$, we also get
\begin{equation*}
1 
\geq	 |\delta(v) \cap A| \cdot (y_{x}^A)_e
\enspace ,
\end{equation*}
implying together with the previous inequality that
\begin{equation*}
1 
\geq	\max\{ |\delta(u) \cap A|,  |\delta(v) \cap A| \} \cdot  (y_{x}^A)_e
\enspace .
\end{equation*}
The above shows~\eqref{eq:BipMatchCRSoptimal_main}, and thereby concludes the proof of Theorem~\ref{thm:mCRs_optimality_bipartite}.
\end{proof}

\section{A monotone contention resolution scheme for general matchings}\label{sec:genMatch}
We now turn our attention to general matchings. Throughout this section, we consider a graph $G = (V,E)$ and denote the set of all matchings of $G$ by $\match[G]$.

\subsection{Developing the scheme}

In order to get a first monotone CR scheme for a matching constraint, we can reuse our scheme for bipartite matchings. %
Given the graph $G$ and a point $x \in \matchP[G]$, we first sample a bipartite subgraph $G' = (U \dcup (V \setminus U), E')$ by including each vertex $v \in V$ in $U$ independently with probability $\frac{1}{2}$ and only keeping edges with one endpoint in $U$ and one endpoint in $V \setminus U$. Then, we simply apply our monotone CR scheme for bipartite matchings to $x|_{E'} \in \degreeP[G']$. Since an edge $e \in E$ appears in $E'$ with probability $\frac{1}{2}$, the resulting monotone scheme for $\matchP[G]$ is $(b, \frac{1}{2} \cdot \beta(b))$-balanced.\footnote{For example, the scheme has a balancedness of $\frac{1}{2} \cdot \beta(1) \geq 0.2381$ for $b=1$. %
} 
This already beats the previously best monotone CR scheme for general matchings (see Table~\ref{tab:priorResults}).
However, using our polyhedral approach again, we can easily do better.
We recall that the matching polytope $\matchP[G]$ %
can be described as
\begin{equation*}
	\matchP[G] =  \degreeP[G] \cap \bigg\{ x \in \mathbb{R}^E_{\geq 0}  \,\bigg\vert\, 
			x(E[S]) \leq \frac{|S|-1}{2} \quad \forall S \subseteq V, |S| \textup{ odd}
	\bigg\}
	\enspace ,
\end{equation*}
where $E[S] \subseteq E$ denotes the set of all edges with both endpoints in $S \subseteq V$ and $\degreeP[G] = \big\{ x \in \mathbb{R}^E_{\geq 0}  \mid  x(\delta(v)) \leq 1 \quad \forall v \in V     \big\}$. %
In particular, it is well-known and easy to see that $\frac{2}{3} \degreeP[G] \subseteq \matchP[G]$. %
 Hence, given a point $x \in \matchP[G] \subseteq \degreeP[G] $, %
  we can apply the procedure described in Section~\ref{sec:bipMatch} %
 to get a point $y_x \in \degreeP[G]$ and scale it with $\frac{2}{3}$ to obtain a point in $\matchP[G]$. %
 This way, we only lose a factor of $\frac{1}{3}$ instead of $\frac{1}{2}$, resulting in a monotone $(b,\frac{2}{3} \cdot \beta(b))$-balanced CR scheme for $\matchP[G]$.%

In the following, we first present a very simple scheme whose balancedness is already higher than $\frac{2}{3}\cdot \beta(1) \approx 0.3174$. Understanding which cases cause this scheme to have problems then enables us to give an improved monotone CR scheme for $\matchP[G]$ that proves Theorem~\ref{thm:mCRs_existence_general}.

\begin{example}\label{ex:CRscheme_randomOrder}
	Given a graph $G=(V,E)$, a point $x \in \matchP[G]$, and a set $A \subseteq \supp(x)$, the following procedure yields a matching $M_x^A \subseteq A$ in $G$.
	\begin{enumerate}[label=\normalfont\arabic*., leftmargin=4em]
		\item Choose an order on the edges in $A$ uniformly at random.
		\item Include an edge $e = \{u,v\} \in A$ in $M_x^A$ if $e$ is the first one in the order among all edges in $(\delta(u) \cup \delta(v)) \cap A$, i.e., among all edges in $A$ adjacent to $e$.
	\end{enumerate}
	Since the above procedure makes sure that no two adjacent edges get selected, it clearly yields a matching contained in $A$. Also, the probability of an edge $e = \{u,v\} \in A$ being included in $M_x^A$ decreases when $(\delta(u) \cup \delta(v)) \cap A$ gets larger, i.e., when the number of edges in $A$ adjacent to $e$ increases.
	The above therefore describes a monotone CR scheme for $\matchP[G]$.
	
	In order to determine the balancedness of this scheme, we consider its marginals.\footnote{%
		Resorting to marginals does not significantly simplify the analysis here. Having the marginals already will save us some work later, though.%
}
	For an edge $e = \{u,v\} \in A$ to be included in $M_x^A$, it has to be the first one among all edges in $(\delta(u) \cup \delta(v)) \cap A$ with respect to the random order on $A$.
	Since this happens with probability exactly $\frac{1}{|(\delta(u) \cup \delta(v)) \cap A |}$, we obtain that for $x \in \matchP[G]$ and $A \subseteq \supp (x)$, the marginal vector $y_x^A$ of this CR scheme equals
		\begin{equation*}
			\big(y^A_x\big)_e \coloneqq
			\begin{cases}
				\frac{1}{|(\delta(u)\cup\delta(v))\cap A|} &\forall e = \{u,v\} \in A\enspace,\\
				0 &\forall e\in E\setminus A \enspace.
			\end{cases}
		\end{equation*}
	An illustration is provided in Figure~\ref{fig:generalMarginalCRS}.
	Note that since $y_x^A$ is the vector of marginals of a CR scheme for $\matchP[G]$, we automatically get that $y_x^A$ is contained in the matching polytope $\matchP[G]$ (and that $\supp (y_x^A) \subseteq A$, which is trivial to see anyway).
	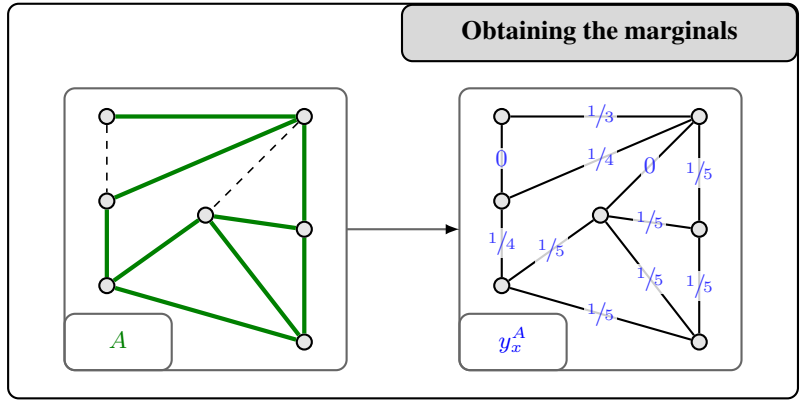
\begin{figure}[ht]
		\centering

\tikzstyle{vertex}=[circle, draw=black, thick, inner sep=1pt, minimum size=2mm, fill=black!10]
\tikzstyle{edgenode}=[rectangle, fill=white, fill opacity=0.8, text opacity=1, inner sep=1pt]
\tikzstyle{edge}=[thick, draw=black]
\tikzstyle{selectedEdge}=[ultra thick, draw=\colorofset]
\tikzstyle{discardedEdge}=[semithick, dashed, draw=black]
\tikzstyle{arc}=[thick, draw=black!60, -latex]
\tikzstyle{boxed}=[thick, draw=black!60, rounded corners]
\tikzstyle{boxedFat}=[thick, draw=black, rounded corners]
\tikzstyle{boxedLight}=[thick, draw=black, rounded corners, fill = black!15, minimum height=0.75*1cm]

\begin{tikzpicture}[scale=0.75]

\begin{scope}[shift={(3.5,0.5)}, local bounding box=marginal1]

\draw[boxed] (0,0) rectangle (5,5);
\draw[boxed] (0,0) rectangle (1.9,1);
\node[font=\footnotesize, color=\colorofset] at (0.95,0.5) {$\strut A$};

\node[vertex] (u1) at (0.75,4.5) {};
\node[vertex] (u2) at (0.75,3.0) {};
\node[vertex] (u3) at (0.75,1.5) {};
\node[vertex] (v1) at (4.25,4.5) {};
\node[vertex] (v2) at (4.25,2.5) {};
\node[vertex] (v3) at (4.25,0.5) {};
\node[vertex] (w1) at (2.50,2.75) {};

\draw[discardedEdge] (u1) -- (u2);
\draw[selectedEdge] (u1) -- (v1);
\draw[selectedEdge] (u2) -- (u3);
\draw[selectedEdge] (u2) -- (v1);
\draw[selectedEdge] (u3) -- (v3);
\draw[selectedEdge] (u3) -- (w1);
\draw[selectedEdge] (v1) -- (v2);
\draw[discardedEdge] (v1) -- (w1);
\draw[selectedEdge] (v2) -- (v3);
\draw[selectedEdge] (v2) -- (w1);
\draw[selectedEdge] (v3) -- (w1);

\end{scope}

\begin{scope}[shift={(10.5,0.5)}, local bounding box=marginal2]

\draw[boxed] (0,0) rectangle (5,5);
\draw[boxed] (0,0) rectangle (1.9,1);
\node[font=\footnotesize, color=\colorofx] at (0.95,0.5) {$\strut y_x^A$};

\node[vertex] (u1) at (0.75,4.5) {};
\node[vertex] (u2) at (0.75,3.0) {};
\node[vertex] (u3) at (0.75,1.5) {};
\node[vertex] (v1) at (4.25,4.5) {};
\node[vertex] (v2) at (4.25,2.5) {};
\node[vertex] (v3) at (4.25,0.5) {};
\node[vertex] (w1) at (2.50,2.75) {};

\draw[edge] (u1) -- node[font=\footnotesize, color=\colorofx, edgenode, pos=0.5] {$0$} (u2);
\draw[edge] (u1) -- node[font=\footnotesize, color=\colorofx, edgenode, pos=0.5] {$\sfrac{1}{3}$} (v1);
\draw[edge] (u2) -- node[font=\footnotesize, color=\colorofx, edgenode, pos=0.5] {$\sfrac{1}{4}$} (u3);
\draw[edge] (u2) -- node[font=\footnotesize, color=\colorofx, edgenode, pos=0.5] {$\sfrac{1}{4}$} (v1);
\draw[edge] (u3) -- node[font=\footnotesize, color=\colorofx, edgenode, pos=0.5] {$\sfrac{1}{5}$} (v3);
\draw[edge] (u3) -- node[font=\footnotesize, color=\colorofx, edgenode, pos=0.5] {$\sfrac{1}{5}$} (w1);
\draw[edge] (v1) -- node[font=\footnotesize, color=\colorofx, edgenode, pos=0.5] {$\sfrac{1}{5}$} (v2);
\draw[edge] (v1) -- node[font=\footnotesize, color=\colorofx, edgenode, pos=0.5] {$0$} (w1);
\draw[edge] (v2) -- node[font=\footnotesize, color=\colorofx, edgenode, pos=0.5] {$\sfrac{1}{5}$} (v3);
\draw[edge] (v2) -- node[font=\footnotesize, color=\colorofx, edgenode, pos=0.5] {$\sfrac{1}{5}$} (w1);
\draw[edge] (v3) -- node[font=\footnotesize, color=\colorofx, edgenode, pos=0.5] {$\sfrac{1}{5}$} (w1);

\end{scope}

\draw[boxedFat] ($(marginal1.north west) + (-1,1.5)$) rectangle ($(marginal2.south east) + (1,-0.5)$);

\node[boxedLight, font=\small, minimum width=0.75*7cm, anchor=north east, inner sep = 0pt] at ($(marginal2.north east)+(1,1.5)$) {\textbf{Obtaining the marginals\strut}};

\draw[arc] (marginal1.east) -- (marginal2.west);

\end{tikzpicture}  		\caption{The marginal vector corresponding to the described monotone contention resolution scheme for the matching polytope $\matchP[G]$.}\label{fig:generalMarginalCRS}
	\end{figure}

	To determine the balancedness of our scheme, we fix $x \in \matchP[G]$ and $e = \{u,v\} \in \supp (x)$. %
	We then obtain
	\begin{equation}\label{eq:gen_easyCR_balancedness}
	\begin{aligned}
	\Pr \Big[ e \in M_x^{R(x)} \;\Big\vert\; e \in R(x) \Big] 
	& =	\expval\Big[\big(y_x^{R(x)}\big)_e \;\Big\vert\; e \in R(x) \Big] \\
	& = \expval\bigg[   \frac{1}{  | (\delta(u) \cup \delta(v) ) \cap R(x)|  \}} \;\bigg\vert\; e \in R(x) \bigg] \\
	& = \expval\bigg[   \frac{1}{ 1 +  | (\delta(u) \cup \delta(v) ) \cap (R(x)\setminus \{e\})|  \}} \bigg] \\
	& \geq \expval\bigg[    \frac{1}{1 + 
		|\delta(u)\cap (R(x)\setminus \{e\})| +
		|\delta(v)\cap (R(x)\setminus \{e\})|}  \bigg]
	\enspace. 
	\end{aligned}
	\end{equation}
	Repeating the calculations in~\eqref{eq:bip_easyCR_balancedness_lb}, we see that the latter expression is at least $\frac{1}{3}$, meaning that the described scheme is $\frac{1}{3}$-balanced.
	In contrast to~\eqref{eq:bip_easyCR_balancedness}, however, the analysis is tight here.
	This is shown by the example given in Figure~\ref{fig:CRSgeneralBadBernBern}.
	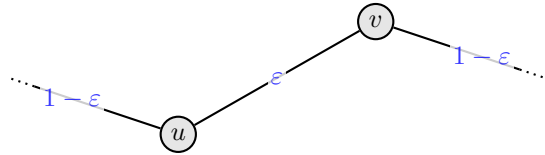
\begin{figure}[ht]
		\centering

\tikzstyle{vertex}=[circle, draw=black, thick, inner sep=1pt, minimum size=2mm, fill=black!10]
\tikzstyle{edgenode}=[rectangle, fill=white, fill opacity=0.8, text opacity=1, inner sep=1pt]
\tikzstyle{edge}=[thick, draw=black]
\tikzstyle{selectedEdge}=[ultra thick, draw=\colorofset]
\tikzstyle{discardedEdge}=[semithick, dashed, draw=black]
\tikzstyle{arc}=[thick, draw=black!60, -latex]
\tikzstyle{boxed}=[thick, draw=black!60, rounded corners]
\tikzstyle{boxedFat}=[thick, draw=black, rounded corners]
\tikzstyle{boxedLight}=[thick, draw=black, rounded corners, fill = black!15, minimum height=0.75*1cm]

\begin{tikzpicture}[scale=0.75]

\node[vertex, minimum size=10/0.75, font=\small] (u) at (-1.75,-1) {$u$};
\node[vertex, minimum size=10/0.75, font=\small] (v) at (1.75,1) {$v$};

\coordinate (u0) at ($(u) + (-3,1.0)$) {};
\coordinate (u1) at ($(u) + (-3,1.5)$) {};
\coordinate (u2) at ($(u) + (-3,2.5)$) {};
\coordinate (u4) at ($(u) + (-3,0.5)$) {};
\coordinate (u5) at ($(u) + (-3,-0.5)$) {};

\coordinate (v0) at ($(v) + (3,-1.0)$) {};
\coordinate (v1) at ($(v) + (3,-1.5)$) {};
\coordinate (v2) at ($(v) + (3,-2.5)$) {};
\coordinate (v4) at ($(v) + (3,-0.5)$) {};
\coordinate (v5) at ($(v) + (3,0.5)$) {};

\draw[edge] (u) edge node[font=\small, color=\colorofx, edgenode, pos=0.5] {$\epsilon$} (v);

\draw[edge] (u) edge node[font=\small, color=\colorofx, edgenode, pos=0.7] {$1-\epsilon$} ($(u)!0.85!(u0)$);
\draw[edge, dotted] ($(u)!0.85!(u0)$) edge (u0);

 \draw[edge] (v) edge node[font=\small, color=\colorofx, edgenode, pos=0.7] {$1-\epsilon$} ($(v)!0.85!(v0)$);
 \draw[edge, dotted] ($(v)!0.85!(v0)$) edge (v0);

\end{tikzpicture} 		\caption{The blue numbers indicate the $x$-value of each edge. For $\epsilon$ going to $0$, the balancedness of the middle edge $e=\{u,v\}$ converges to $\frac{1}{3}$. This holds since both edges adjacent to $e$ appear in $R(x)$ with probability $(1-\epsilon)^2$, %
		meaning that when $e = \{u,v\}$ appeared in $R(x)$, its two adjacent edges will most likely have appeared in $R(x)$ as well. In that case, $e = \{u,v\}$ is only selected if it is the first among these three edges in the random order of all edges that appeared in $R(x)$. 
		}\label{fig:CRSgeneralBadBernBern}
	\end{figure}
	
\end{example}

While the above already beats the previously best monotone CR scheme for matchings (see Table~\ref{tab:priorResults}), we can improve it further. %
Observe that, as in the bipartite setting, the presence of edges having large $x$-value are crucial to build a worst case for the above analysis.
However, as we already discussed in the bipartite case, one can reduce to instances with only small $x$-values through edge splitting. (This is formally stated by Lemma~\ref{lem:CRscheme_smallx_general} in Appendix~\ref{app:CRscheme_smallx}.)
Algorithm~\ref{algo:GenMatchCRS} below describes the scheme that results from applying the splitting technique to the CR scheme described in Example~\ref{ex:CRscheme_randomOrder}.
More precisely, analogous to the bipartite case, Algorithm~\ref{algo:GenMatchCRS} corresponds to splitting each edge into many parallel copies with small equal $x$-value, then applying the scheme from Example~\ref{ex:CRscheme_randomOrder}, and finally interpreting the result in the original graph. As previously, the Poisson random variables indicate how many of the parallel copies appeared. 
As we prove in the following, the CR scheme described by Algorithm~\ref{algo:GenMatchCRS} implies Theorem~\ref{thm:mCRs_existence_general}.

\begin{algorithm2e}[ht]
	\caption{Contention resolution for $\matchP[G]$}
	\label{algo:GenMatchCRS}
	\KwIn{Point $x \in \matchP[G]$, set $A \subseteq \supp(x)$.}
	\KwOut{Point $y_x^A \in \matchP$ with $\supp(y_x^A) \subseteq A$.} %
	Sample $\overline{A} \subseteq A$ by including edge $e \in A$ independently with probability  \smash{$ (1- \myexp^{-x_e})/ x_e$}\;
	For each edge $e \in \overline{A}$, let $q_e \in \mathbb{Z}_{\geq 1}$ be an independent realization of a $\PoisD{x_e}$-random variable conditioned on it being at least $1$\;
	For each edge $e \in E \setminus \overline{A}$, let $q_e \coloneqq 0$\;
	For each edge $e = \{u,v\} \in E$, let 
	\smash{$(y_x^A)_e \coloneqq \frac{ q_e }{  \sum_{g \in \delta(u) \cup \delta(v)} q_g  }  $}\label{algline:genMatchRatio}\;
	Return $y_x^A$\;
\end{algorithm2e}
For the ratio $\frac{q_e}{\sum_{g\in \delta(u)\cup \delta(v)} q_g}$ in line~\ref{algline:genMatchRatio} of the algorithm, we use the same convention as in the bipartite case, i.e., we consider this ratio to be $0$ if both numerator and denominator are $0$.
To show that Algorithm~\ref{algo:GenMatchCRS} indeed implies Theorem~\ref{thm:mCRs_existence_general}, we need to analyze the returned marginal vector $y_x^A$ for $x \in \matchP[G]$ and $A \subseteq \supp (x)$, and show that it fulfills~\ref{item:marginalSupp}--\ref{item:marginalMonotone} with the balancedness claimed in Theorem~\ref{thm:mCRs_existence_general}.

Be definition of the scheme, the returned random marginals $y_x^A$ clearly fulfill~\ref{item:marginalSupp}, i.e., $\supp(y_x^A) \subseteq A$. Moreover, the following lemma shows that~\ref{item:marginalConvComb} holds, too, i.e., $y_x^A$ is in the matching polytope $\matchP[G]$.
\begin{lemma}\label{lem:marginals_genMatchPolytope}
	Let $G=(V,E)$ be a graph and let $w \in \mathbb{R}^E_{\geq 0}$.
	Then, the vector $z \in \mathbb{R}^E_{\geq 0}$ defined by
	\begin{equation*}
		z_e \coloneqq
		\begin{cases}
			\frac{ w_e }{  \sum_{g \in \delta(u) \cup \delta(v)} w_g  }
						 &\forall e = \{u,v\} \in \supp(w) \enspace,\\
			0 &\forall e\in E\setminus \supp(w) \enspace,
		\end{cases}
	\end{equation*}
	is contained in the matching polytope $\matchP[G]$ of $G$.
\end{lemma}
\begin{proof}
Given $w \in \mathbb{R}_{\geq 0}^E$, note that we can without loss of generality assume that $\supp(w) = E$.
We then define a random matching $M \subseteq E$ as follows.
Letting $d_e$ for each edge $e \in E$ be a realization of an independent $\ExpD{w_e}$-random variable $D_e$, we select an edge $e = \{u,v\} \in E$ only if its realization $d_e$ is strictly smaller than all realizations $d_g$ for $g \in (\delta(u) \cup \delta(v)) \setminus \{e\}$, i.e., we consider all edges sharing (at least) one vertex with $e$. Clearly, the resulting set $M$ is a matching in $G$.
Moreover, we claim that $\expval[\chi^M] = z$, which proves $z \in \matchP[G]$.

For a fixed edge $e = \{u,v\} \in E$, we have
 	\begin{equation*}
 		\big(\expval[\chi^M] \big)_e =\Pr[e \in M] 
 		= \Pr [D_e < \min \{  D_g \mid g \in (\delta(u) \cup \delta(v)  ) \setminus \{e\}  \}    ]
 		\enspace .
 	\end{equation*}
To analyze the above expression, we use the following two properties of exponential random variables. 
For any two independent exponential random variables $\ExpD[1]{\xi}$ and $\ExpD[2]{\zeta}$, we have that:
\begin{enumerate}
\item $\min\{ \ExpD[1]{\xi}, \ExpD[2]{\zeta} \}$ has the same distribution as $\ExpD{\xi+\zeta}$, and
\item $\Pr[\ExpD[1]{\xi} < \ExpD[2]{\zeta}] = \frac{\xi}{\xi + \zeta}$.
\end{enumerate}
With this, and letting $\overline{w}_e \coloneqq \sum_{g \in (\delta(u)\cup\delta(v))\setminus\{e\}} w_g $, we get that
\begin{align*}
		\big(\expval[\chi^M] \big)_e 
		=	\Pr [D_e < \min \{  D_g \mid g \in (\delta(u) \cup \delta(v)  ) \setminus \{e\}  \}    ]
	&= \Pr [\ExpD[1]{w_e} < \ExpD[2]{ \overline{w}_e } ] \\
	&= \frac{w_e}{ w_e + \overline{w}_e  }
	= \frac{w_e}{ \sum_{g \in \delta(u)\cup\delta(v)} w_g }
	= z_e
	\enspace ,
\end{align*}
where $\ExpD[1]{w_e}$ and $ \ExpD[2]{ \overline{w}_e } $ are two independent exponential random variables with parameters $w_e$ and $\overline{w}_e$, respectively.
This proves $\expval[\chi^M] = z$, and thus Lemma~\ref{lem:marginals_genMatchPolytope}.
\end{proof}

Note that, as in Section~\ref{sec:bipMatch}, the vectors $y_x^A \in \matchP[G]$ produced by Algorithm~\ref{algo:GenMatchCRS} are random. The true marginals of the scheme are given by $\expval [ y_x^A]$, where the expectation is over the random set $\overline{A} \subseteq A$ and the realizations $q_e$ of the (conditioned) Poisson random variables. Since the vectors $y_x^A$ satisfy~\ref{item:marginalSupp} and~\ref{item:marginalConvComb}, so do the true marginals $\expval[ y_x^A ]$.

We now show that the CR scheme described by Algorithm~\ref{algo:GenMatchCRS} is monotone, i.e., it fulfills~\ref{item:marginalMonotone}, which means $\expval[(y_x^{A})_e]  \geq \expval[(y_x^{B})_e] $ for any $x\in \matchP[G]$ and $e\in A \subseteq B \subseteq \supp(x)$.

\begin{lemma}
The CR scheme defined by Algorithm~\ref{algo:GenMatchCRS} is monotone, i.e., it fulfills~\ref{item:marginalMonotone}.
\end{lemma}
\begin{proof}
Let $x \in \matchP[G]$ and $e=\{u,v\} \in A \subseteq B \subseteq \supp(x)$.
Moreover, for each edge $g \in E$, let $Q_g$ be an independent $\PoisD{x_g}$-random variable conditioned on it being at least $1$. 
Since it is possible to couple the random sets $\overline{A} \subseteq A$ and $\overline{B} \subseteq B$, which we get in the subsampling step of Algorithm~\ref{algo:GenMatchCRS}, such that $\overline{A} = \overline{B}\cap A  \subseteq \overline{B}$, we conclude that
\begin{align*}
\expval[(y_x^{A})_e] 
&= \Prb[ e \in \overline{A} ]  \cdot 
\expval \bigg[
\frac{Q_e}{  \sum_{g \in (\delta(u)\cup \delta(v)) \cap \overline{A}} Q_g }  \,\bigg\vert\, e \in \overline{A}	\bigg]   \\
&\geq \Prb[ e \in \overline{B} ]  \cdot 
\expval \bigg[
\frac{Q_e}{ \sum_{g \in (\delta(u)\cup \delta(v)) \cap \overline{B}} Q_g }   \,\bigg\vert\, e \in \overline{B}	\bigg]   \\
&= \expval[(y_x^{B})_e] 
\enspace ,
\end{align*}
where the expectation is taken over the random sets $\overline{A} \subseteq A$ and $\overline{B} \subseteq B$ as well as the random variables $Q_g$, $g \in E$.
This proves that~\ref{item:marginalMonotone} is satisfied.
\end{proof}

All that remains to be shown to prove Theorem~\ref{thm:mCRs_existence_general} is that the CR scheme defined by Algorithm~\ref{algo:GenMatchCRS} has a good balancedness, where the scheme has a balancedness of $c>0$ if the marginal vector $y_x^A$, for $A=R(x)$, fulfills~\ref{item:marginalBalance}.
To this end, we first present a simpler procedure with the same output distribution as $y_x^{R(x)}$. This procedure, shown below as Algorithm~\ref{algo:GenMatchCRSshort}, corresponds to merging the step of randomly rounding $x$ to obtain $R(x)$ with Algorithm~\ref{algo:GenMatchCRS} into a single procedure. Because it merges these two steps, it is not a CR scheme itself, as it does not take $R(x)$ as input to produce a feasible subset. This difference is crucial when we want to combine CR schemes for different constraints. We recall that in this case each CR scheme takes the same random set $R(x)$ as input, and one returns the intersection of the outputs of the different schemes. This is one main reason why Algorithm~\ref{algo:GenMatchCRS} cannot simply be replaced by the simpler Algorithm~\ref{algo:GenMatchCRSshort}.
Nevertheless, because the output distribution of Algorithm~\ref{algo:GenMatchCRS} applied to $R(x)$ is the same as the output of Algorithm~\ref{algo:GenMatchCRSshort}, as we show in Lemma~\ref{lem:genMatchSameDist} below, we can use the simpler Algorithm~\ref{algo:GenMatchCRSshort} to analyze the balancedness of Algorithm~\ref{algo:GenMatchCRS}.
\begin{algorithm2e}[ht]	
	\caption{Combining independent rounding with Algorithm~\ref{algo:GenMatchCRS}}
	\label{algo:GenMatchCRSshort}
	\KwIn{Point $x \in \matchP[G]$.}
	\KwOut{Point $y_x \in \matchP$ with $\supp(y_x) \subseteq \supp(x)$.}
	For each edge $e \in E$, let $q_e \in \mathbb{Z}_{\geq 0}$ be an independent realization of a $\PoisD{x_e}$-random variable\;
	For each edge $e = \{u,v\} \in E$, let 
	\smash{$(y_x)_e \coloneqq \frac{ q_e } {\sum_{g \in \delta(u) \cup \delta(v)} q_g   } $}\;
	Return $y_x$\;
\end{algorithm2e}

\begin{lemma}\label{lem:genMatchSameDist}
Let $x\in \matchP[G]$. Then the random marginal vector obtained by Algorithm~\ref{algo:GenMatchCRS} when applied to $R(x)$ has the same distribution as the random marginal vector obtained by Algorithm~\ref{algo:GenMatchCRSshort}.
\end{lemma}
\begin{proof}
When using Algorithm~\ref{algo:GenMatchCRS} to round a given point $x \in \matchP[G]$, we start by independently rounding $x$ to get a set $R(x) \subseteq \supp (x)$ and then apply Algorithm~\ref{algo:GenMatchCRS} to $x$ and $A=R(x)$.
In a first step, the latter amounts to sampling a subset $\overline{A} \subseteq A$ by including each edge $e \in A=R(x)$ independently with probability $(1 - \myexp^{-x_e})/x_e$.
A given edge $e \in E$ therefore appears in the random set $\overline{A}$ independently with probability
\begin{equation}\label{eq:GenMatchCRS_survive}
\begin{aligned}
\Pr [ e \in \overline{A}] 
&= \Pr[ e \in \overline{A} \mid e \in A] \cdot \Pr[e \in A] \\
&= \Pr[ e \in \overline{A} \mid e \in A] \cdot \Pr[e \in R(x)]
= \frac{1-\myexp^{-x_e}}{x_e} \cdot x_e = 1 - \myexp^{-x_e}
\enspace ,
\end{aligned}
\end{equation}
which equals the probability that a $\PoisD{x_e}$-random variables is at least $1$ (or, equivalently, does not equal $0$).
Given $\overline{A} \subseteq A = R(x)$, Algorithm~\ref{algo:GenMatchCRS} then sets $q_e \in \mathbb{Z}_{\geq 0}$ to be an independent realization of a $\PoisD{x_e}$-random variable conditioned on it being at least $1$ if $e \in \overline{A}$, and $0$ otherwise.
Using~\eqref{eq:GenMatchCRS_survive}, we see that the resulting distribution is indeed the same as the one of a $\PoisD{x_e}$-random variable.
Since all edges are treated independently, we conclude that the random vector $y_x^{A}$ obtained by Algorithm~\ref{algo:GenMatchCRS} when applied to $x$ and $A = R(x)$ has the same distribution as the random vector $y_x$ obtained by Algorithm~\ref{algo:GenMatchCRSshort} when applied to $x$.
\end{proof}

Finally, we now prove that Algorithm~\ref{algo:GenMatchCRS} has the balancedness claimed by Theorem~\ref{thm:mCRs_existence_general}. Due to Lemma~\ref{lem:genMatchSameDist}, this boils down to showing that for $x\in b\matchP[G]$, the random marginal vector $y_x$ returned by Algorithm~\ref{algo:GenMatchCRSshort} fulfills~\ref{item:marginalBalance} for $c=\gamma(b)$.\footnote{%
	We recall that $\gamma(b) \coloneqq  \expval \big[ \frac{1}{1 +  \PoisD{2b}  }  \big] = \frac{1-\myexp^{-2b}}{2b}$ for $b \in [0,1]$, where $\PoisD{2b}$ is a Poisson random variable with parameter $2b$. For $b=1$, it holds $\gamma(1)  \geq 0.4323$.%
}
This implies by Proposition~\ref{prop:marginalCRschemes} that Algorithm~\ref{algo:GenMatchCRS} can be efficiently transformed into a monotone $(b,\gamma(b))$-balanced CR scheme, showing Theorem~\ref{thm:mCRs_existence_general}.
The line of argument is very similar to the one of the proof of Theorem~\ref{thm:mCRs_existence_bipartite} in Section~\ref{sec:bipMatch}.

\begin{proof}[Proof of Theorem~\ref{thm:mCRs_existence_general}]
	Let $b \in [0,1]$, $x \in b \matchP[G]$, and $e = \{u,v\} \in  \supp (x)$. 
	As already mentioned above, Theorem~\ref{thm:mCRs_existence_general} follows if we can show that 
	\begin{equation}\label{eq:GenMatchCRSshortBalanc}
		\expval [  (y_x)_e ]
		\geq \gamma(b) \cdot x_e  
		\enspace ,
	\end{equation}
	where $y_x$ is the random vector returned by Algorithm~\ref{algo:GenMatchCRSshort}.

	To prove~\eqref{eq:GenMatchCRSshortBalanc}, we let $Q_g$ for $g \in E$ be independent random variables, where for each edge $g \in E$,  $Q_g$ follows a $\PoisD{x_g}$-distribution.
	Moreover, we define the set $E_{uv} \subseteq E$ of all edges in $E$ that go from $u$ to $v$ (including the edge $e$ itself), and set $x_{uv} \coloneqq x (E_{uv})$, $x_u \coloneqq x (\delta (u))$, and $x_v \coloneqq x(\delta (v))$. With this, we have that
	\begin{align*}
		\expval[ (y_x)_e  ]  
		=    \expval \bigg[  \frac{Q_e}{  \sum_{g \in \delta(u) \cup \delta(v)} Q_g }   \bigg]
		=    \expval \bigg[  \frac{Q_e}{ Q_e +   \sum_{g \in \delta(u) \setminus E_{uv} }  Q_g  + \sum_{g \in \delta(v) \setminus E_{uv} } Q_g   +  \sum_{g \in E_{uv} \setminus \{e\}} Q_g }   \bigg]
		\enspace .
	\end{align*}
	In order to analyze the above expression, we use the same simplified notation that was already helpful in the proof of Theorem~\ref{thm:mCRs_existence_bipartite}. We let $\BernD[i]{\xi}$ denote a Bernoulli random variable with parameter $\xi$ and $\PoisD[j]{\zeta}$ a Poisson random variable with parameter $\zeta$, where we assume that all random variables that appear are independent when they have different subscripts. 
Recall again that the sum of two independent Poisson random variables $\PoisD[1]{\xi}$ and $\PoisD[2]{\zeta}$ has a $\PoisD{\xi+\zeta}$-distribution and that	
	$\sum_{i=1}^k \BernD[i]{\frac{\xi}{k}}$ approaches a $\PoisD{\xi}$-distribution as $k \to \infty$. Using these properties as well as dominated convergence and linearity of expectation, we obtain
	\begin{align*}
		\expval[ (y_x)_e  ]  
		&=    \expval \bigg[  \frac{ \PoisD[e]{x_e}  }{ \PoisD[e]{x_e} +  \sum_{g \in \delta(u) \setminus E_{uv} } \PoisD[g]{x_g}   + \sum_{g \in \delta(v)  \setminus E_{uv} } \PoisD[g]{x_g}    +   \sum_{g \in E_{uv} \setminus \{e\}} \PoisD[g]{x_g} }   \bigg]   \\
		& =  \expval \bigg[  \frac{ \PoisD[e]{x_e}  }{ \PoisD[e]{x_e} +  \PoisD[u]{x_u - x_{uv}} +  \PoisD[v]{x_v-x_{uv}}  + \PoisD[uv]{x_{uv} - x_e} }   \bigg]   \\	
		&= \expval \bigg[  \lim_{k \to \infty}  \frac{ \sum_{i=1}^k \BernD[i]{\frac{x_e}{k}}  }{ \sum_{i=1}^k \BernD[i]{\frac{x_e}{k}} +   \PoisD[1]{x_u + x_v -x_{uv} - x_e} }   \bigg]  \\
		&=  \lim_{k \to \infty}  \sum_{i=1}^k \expval \bigg[  \frac{ \BernD[i]{\frac{x_e}{k}}  }{ \BernD[i]{\frac{x_e}{k}} +  \sum_{j \in [k] \setminus \{i\}} \BernD[j]{\frac{x_e}{k}} +  \PoisD[1]{x_u + x_v -x_{uv} - x_e} }   \bigg]  \\
		&=  \lim_{k \to \infty}  \sum_{i=1}^k   \frac{x_e}{k} \cdot  \expval \bigg[  \frac{ 1  }{ 1+ \sum_{j \in [k] \setminus \{i\}} \BernD[j]{\frac{x_e}{k}} +  \PoisD[1]{x_u + x_v-x_{uv} - x_e} }   \bigg]	\\
		& \geq \lim_{k \to \infty}  \sum_{i=1}^k   \frac{x_e}{k} \cdot  \expval \bigg[  \frac{ 1  }{ 1+ \sum_{j=1}^k \BernD[j]{\frac{x_e}{k}} +  \PoisD[1]{x_u + x_v-x_{uv} - x_e} }   \bigg] \\
		& = x_e \cdot   \expval \bigg[ \lim_{k \to \infty}   \frac{ 1  }{ 1+ \sum_{j=1}^k \BernD[j]{\frac{x_e}{k}} +  \PoisD[1]{x_u + x_v-x_{uv} - x_e} }   \bigg] \\
		& = x_e \cdot   \expval \bigg[   \frac{ 1  }{ 1+ \PoisD[e]{x_e} +  \PoisD[1]{x_u + x_v-x_{uv} - x_e} }   \bigg] %
		\enspace ,
	\end{align*}
	which shows 
	\begin{equation}\label{eq:GenMatchCRSBalancSharp}
	\expval[ (y_x)_e  ]  \geq  x_e \cdot   \expval \bigg[   \frac{ 1  }{ 1+    \PoisD{x_u + x_v-x_{uv}}  } \bigg]
	\enspace .
	\end{equation}
	In particular, it holds that	
	\begin{equation}\label{eq:GenMatchCRSBalancDull}
	\expval[ (y_x)_e  ]  \geq  x_e \cdot   \expval \bigg[   \frac{ 1  }{ 1+    \PoisD{2b}  } \bigg]
	 = x_e \cdot \gamma(b)
	\enspace ,
	\end{equation}
	where we used that $x \in b \matchP[G]$ implies $x(\delta(w)) \leq b$ for any vertex $w \in V$.
	The above proves~\eqref{eq:GenMatchCRSshortBalanc}, and therefore shows that Algorithm~\ref{algo:GenMatchCRS} satisfies~\ref{item:marginalBalance} with $c = \gamma(b)$.
\end{proof}
We remark that even though~\eqref{eq:GenMatchCRSBalancDull} was sufficient for our purpose above,~\eqref{eq:GenMatchCRSBalancSharp} shows that in some cases, we actually have some slack. This is something we exploit in the next subsection to show that even stronger CR schemes exist for $\matchP[G]$.
However, we note that our analysis of the balancedness of our scheme is tight. This is true since we can have an edge $e = \{u,v\} \in E$ such that $x_{uv} = x (E_{uv})$ is very small while $x_u = x(\delta(u))$ and $x_v = x(\delta(v))$ are almost equal to $b$ (as depicted, for example, in Figure~\ref{fig:CRSgeneralBadBernBern} for $b=1$).

The scheme given in Algorithm~\ref{algo:GenMatchCRS} is inspired by the very simple scheme for $\matchP[G]$ that is described in~\cite{BuchbinderFeldman2018} (it was first reported in~\cite{FeldmanNaorSchwartz2011b}).
Given a point $x \in \matchP[G]$ and a set $A \subseteq \supp (x)$, this scheme first constructs a random set $\overline{A} \subseteq A$ by including each edge $e \in A$ independently with probability $(1 - \myexp^{-x_e})/x_e$. Then, it only keeps edges in $\overline{A}$ that share no endpoint with other edges in $\overline{A}$ (clearly, the resulting set is a matching).
While our scheme starts with the same subsampling, we fix the second step which is quite wasteful.
Instead of discarding \emph{all} edges in $\overline{A}$ that share an endpoint with another edge in $\overline{A}$, we only discard an edge $e \in \overline{A}$ if it shares an endpoint with another edge in $\overline{A}$ \emph{that appears before $e$ in some random ordering of $\overline{A}$}. This idea is exactly what is reflected by our marginals.

Note that in order to turn Algorithm~\ref{algo:GenMatchCRS} into a monotone $(b, \gamma(b))$-balanced CR scheme in the sense of Definition~\ref{def:CRScheme}, we have to sample a matching according to the returned marginal vector $y_x^A \in \matchP[G]$. 
One way to do this is explicitly decomposing $y_x^A$ into a convex combination of matchings in $G$, which, as mentioned in Section~\ref{sec:ourTechniques}, can be done efficiently through standard techniques.
For the marginals $y_x^A$ produced by Algorithm~\ref{algo:GenMatchCRS}, however, there is a more direct way.
Having constructed the vector $q \in \mathbb{Z}_{\geq 0}^E$, we sample for each edge $e \in E$ an independent realization $d_e$ of an $\ExpD{q_e}$-random variable. Then, we only keep an edge $e=\{u,v\} \in E$ %
 if $d_e < \min\{d_g \mid g \in (\delta(u) \cup \delta(v)) \setminus \{e\}  \}$. 
Clearly, this results in a (random) matching $M$ in $G$. Moreover, the proof of Lemma~\ref{lem:marginals_genMatchPolytope} shows that 
\begin{equation*}
\big(\expval[\chi^M] \big)_e 
= \frac{q_e}{ \sum_{g \in \delta(u)\cup\delta(v)} q_g }
\qquad  \forall e=\{u,v\} \in E
\enspace ,
\end{equation*}
implying that $\expval[\chi^M] = y_x^A$.
By replacing line~\ref{algline:genMatchRatio} of Algorithm~\ref{algo:GenMatchCRS} with the above, we could therefore turn Algorithm~\ref{algo:GenMatchCRS} into a CR scheme according to Definition~\ref{def:CRScheme}.
Nevertheless, the marginal formulation of Algorithm~\ref{algo:GenMatchCRS} helped to analyze its balancedness and will also be useful for showing non-optimality of this scheme in the next subsection.

\subsection{Non-optimality}

Since the monotone CR scheme we developed for general matchings is very simple and attains a seemingly ``natural'' balancedness, the question arises whether it is optimal. The goal of this section is to show that it is not.
Before presenting a formal proof, we provide some intuition why it can be improved.

Since our scheme for bipartite matchings achieves a higher balancedness than our scheme for general matchings, one natural way to try to strengthen our CR scheme for general matchings is to combine it with the one for bipartite matchings.
More precisely, given a graph $G=(V,E)$, a point $x \in \matchP[G]$, and a subset $A \subseteq \supp (x)$, we would first sample a set $\overline{A} \subseteq A$ again. Then, we could identify the connected components of $G'=(V, \overline{A})$ and apply to each component that is bipartite our scheme for bipartite matchings (without the subsampling), and only to the other ones our scheme for general matchings (also without the subsampling).
As we show below, this is easily seen to yield a monotone CR scheme for the general matching polytope with no lower balancedness than that of Algorithm~\ref{algo:GenMatchCRS} (see Theorem~\ref{thm:mCRs_existence_general}).
Our goal, however, is to prove that it attains a strictly higher balancedness.
For completeness, a formal description of this combined scheme is given in 
Algorithm~\ref{algo:GenMatchCRSmixed}.
\begin{algorithm2e}[ht]
	\caption{Contention resolution for $\matchP[G]$}
	\label{algo:GenMatchCRSmixed}
	\KwIn{Point $x \in \matchP[G]$, set $A \subseteq \supp(x)$.}
	\KwOut{Point $y_x^A \in \matchP$ with $\supp(y_x^A) \subseteq A$.} %
	Sample $\overline{A} \subseteq A$ by including edge $e \in A$ independently with probability  \smash{$ (1- \myexp^{-x_e})/ x_e$}\;
	For each edge $e \in \overline{A}$, let $q_e \in \mathbb{Z}_{\geq 1}$ be an independent realization of a $\PoisD{x_e}$-random variable conditioned on it being at least $1$\;
	For each edge $e \in E \setminus \overline{A}$, let $q_e \coloneqq 0$\;
	Let $\overline{A}_b$ be the set of edges in bipartite components of $G'=(V, \overline{A})$\;
	For each edge $e = \{u,v\} \in E\setminus \overline{A}_b$, let 
	\smash{$(y_x^A)_e \coloneqq \frac{ q_e }{  \sum_{g \in \delta(u) \cup \delta(v)} q_g  }  $}\;\vspace{0.15cm}
	For each edge $e = \{u,v\} \in \overline{A}_b$, let 
	\smash{$(y_x^A)_e \coloneqq \frac{ q_e } { \max\{  \sum_{g \in \delta(u)} q_g , \sum_{g \in \delta(v)} q_g  \}  } $}\;
	Return $y_x^A$\;
\end{algorithm2e}

\begin{lemma}\label{lem:mCRs_general_combined_noworse}
	Algorithm~\ref{algo:GenMatchCRSmixed} defines a monotone $(b, \gamma(b))$-balanced contention resolution scheme for the general matching polytope. %
\end{lemma}

\begin{proof}
	Given a graph $G=(V,E)$, we consider the procedure described in Algorithm~\ref{algo:GenMatchCRSmixed} and show that it satisfies~\ref{item:marginalSupp}--\ref{item:marginalMonotone} of Proposition~\ref{prop:marginalCRschemes}.
	This then implies Lemma~\ref{lem:mCRs_general_combined_noworse}.
	
	Since we can treat different connected components separately and since the procedures in Algorithm~\ref{algo:BipMatchCRS} and Algorithm~\ref{algo:GenMatchCRS} are valid CR schemes for the bipartite and the general matching polytope, respectively, we immediately get that~\ref{item:marginalSupp} and~\ref{item:marginalConvComb} are fulfilled.
	
	Moreover, %
	we note that it holds for any $q \in \mathbb{Z}^E_{\geq 0}$ and any $e = \{u,v\} \in E$ that
	\begin{equation}\label{eq:mixedCRSdomination}
	\frac{ q_e } { \max\{  \sum_{g \in \delta(u)} q_g , \sum_{g \in \delta(v)} q_g  \}  }
	\geq
	\frac{ q_e } {\sum_{g \in \delta(u) \cup \delta(v)} q_g   }
	\enspace .
	\end{equation}
	This implies that the marginals constructed by Algorithm~\ref{algo:BipMatchCRS} always dominate the marginals constructed by Algorithm~\ref{algo:GenMatchCRS}.
	In particular, the balancedness of this new CR scheme will be no worse than the balancedness of our previous scheme for general matchings (Algorithm~\ref{algo:GenMatchCRS}), which was shown to be $(b, \gamma(b))$-balanced.
	
	In order to get monotonicity of our scheme, it only remains to observe that whenever an edge belongs to a bipartite connected component of some graph, then removing some other edges from the graph will not change this (the connected component might get smaller, though).
	Together with~\eqref{eq:mixedCRSdomination} and the fact that Algorithm~\ref{algo:BipMatchCRS} and Algorithm~\ref{algo:GenMatchCRS} are monotone schemes, this implies that Algorithm~\ref{algo:GenMatchCRSmixed} is monotone as well.

	Since the above shows that Algorithm~\ref{algo:GenMatchCRSmixed} also satisfies~\ref{item:marginalBalance} and~\ref{item:marginalMonotone}, Lemma~\ref{lem:mCRs_general_combined_noworse} follows.
\end{proof}

In the following, we provide a formal proof that Algorithm~\ref{algo:GenMatchCRSmixed} actually achieves a strictly higher balancedness than Algorithm~\ref{algo:GenMatchCRS}.
First, however, we try to build up intuition why this should indeed be true.
Given a graph $G=(V,E)$ and a point $x \in \matchP[G]$, consider an edge $e = \{u,v\} \in \supp (x)$.
As an example, assume that after the independent rounding and the subsampling step, $e$ ended up in a connected component consisting of $e=\{u,v\}$, one other edge $g$ incident to $u$ (but not $v$), and one other edge $h$ incident to $v$ (but not $u$). Moreover, assume that the respective (conditioned) Poisson random variables all realized to $1$, i.e., we have $q_e = q_g = q_h = 1$.
As depicted in Figure~\ref{fig:genCRSloss}, this can happen in two different ways.
For the edge $e$, the scheme from Algorithm~\ref{algo:GenMatchCRS} does not distinguish between these two cases and assigns the same value of $\frac{1}{3}$ to $e$.
On the other hand, in the case where $e$, $g$, and $h$ do not form a triangle, our new scheme from Algorithm~\ref{algo:GenMatchCRSmixed} applies the scheme for bipartite matchings to this component. This results in $e$ being assigned a value of $\frac{1}{2}$ instead of $\frac{1}{3}$.
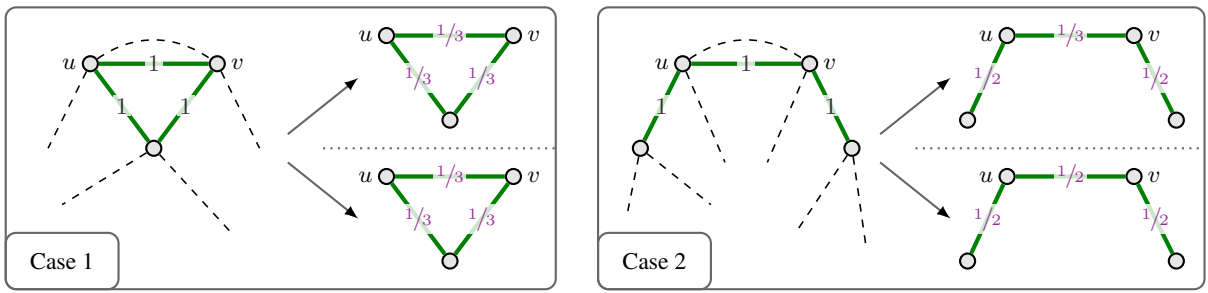
\begin{figure}[ht]
	\centering

\tikzstyle{vertex}=[circle, draw=black, thick, inner sep=1pt, minimum size=2mm, fill=black!10]
\tikzstyle{edgenode}=[rectangle, fill=white, fill opacity=0.8, text opacity=1, inner sep=1pt]
\tikzstyle{edge}=[thick, draw=black]
\tikzstyle{selectedEdge}=[ultra thick, draw=\colorofset]
\tikzstyle{discardedEdge}=[semithick, dashed, draw=black]
\tikzstyle{arc}=[thick, draw=black!60, -latex]
\tikzstyle{separator}=[thick, draw=black!60, dotted]
\tikzstyle{boxed}=[thick, draw=black!60, rounded corners]
\tikzstyle{boxedFat}=[thick, draw=black, rounded corners]
\tikzstyle{boxedLight}=[thick, draw=black, rounded corners, fill = black!15, minimum height=0.75*1cm]

\begin{tikzpicture}[scale=0.75]

\begin{scope}[]

\draw[boxed] (0,0) rectangle (9.75,5);
\draw[boxed] (0,0) rectangle (2,1);
\node[font=\footnotesize] at (1,0.5) {Case 1};

\node[vertex] (u) at (1.5,4) {};
\node[left=0.05, font=\footnotesize] at (u) {$u$};
\node[vertex] (v) at (3.75,4) {};
\node[right=0.05, font=\footnotesize] at (v) {$v$};
\node[vertex] (w0) at (2.625,2.5) {};
\coordinate (w1) at (0.75,2.5) {};
\coordinate (w2) at (4.5,2.5) {};
\coordinate (w3) at (1,1.5) {};
\coordinate (w4) at (4,1) {};

\draw[selectedEdge] (u) -- node[font=\footnotesize, color=black, edgenode, pos=0.5] {$1$} (v);
\draw[selectedEdge] (u) -- node[font=\footnotesize, color=black, edgenode, pos=0.5] {$1$} (w0);
\draw[selectedEdge] (v) -- node[font=\footnotesize, color=black, edgenode, pos=0.5] {$1$} (w0);
\draw[discardedEdge, bend angle=35, bend left] (u) to (v);
\draw[discardedEdge] (u) -- (w1);
\draw[discardedEdge] (v) -- (w2);
\draw[discardedEdge] (w0) -- (w3);
\draw[discardedEdge] (w0) -- (w4);

\node[vertex] (1u) at (6.75,4.5) {};
\node[left=0.05, font=\footnotesize] at (1u) {$u$};
\node[vertex] (1v) at (9,4.5) {};
\node[right=0.05, font=\footnotesize] at (1v) {$v$};
\node[vertex] (1w0) at (7.875,3) {};
\draw[selectedEdge] (1u) -- node[font=\footnotesize, color=\colorofy, edgenode, pos=0.5] {$\sfrac{1}{3}$} (1v);
\draw[selectedEdge] (1u) -- node[font=\footnotesize, color=\colorofy, edgenode, pos=0.5] {$\sfrac{1}{3}$} (1w0);
\draw[selectedEdge] (1v) -- node[font=\footnotesize, color=\colorofy, edgenode, pos=0.5] {$\sfrac{1}{3}$} (1w0);

\node[vertex] (2u) at (6.75,2) {};
\node[left=0.05, font=\footnotesize] at (2u) {$u$};
\node[vertex] (2v) at (9,2) {};
\node[right=0.05, font=\footnotesize] at (2v) {$v$};
\node[vertex] (2w0) at (7.875,0.5) {};
\draw[selectedEdge] (2u) -- node[font=\footnotesize, color=\colorofy, edgenode, pos=0.5] {$\sfrac{1}{3}$} (2v);
\draw[selectedEdge] (2u) -- node[font=\footnotesize, color=\colorofy, edgenode, pos=0.5] {$\sfrac{1}{3}$} (2w0);
\draw[selectedEdge] (2v) -- node[font=\footnotesize, color=\colorofy, edgenode, pos=0.5] {$\sfrac{1}{3}$} (2w0);

\draw[arc] (5,2.75) -- (6.25,3.75);
\draw[arc] (5,2.25) -- (6.25,1.25);

\draw[separator] (5.625,2.5) -- (9.75,2.5);

\end{scope}

\begin{scope}[shift={(10.5,0)}]

\draw[boxed] (0,0) rectangle (10.75,5);
\draw[boxed] (0,0) rectangle (2,1);
\node[font=\footnotesize] at (1,0.5) {Case 2};

\node[vertex] (u) at (1.5,4) {};
\node[left=0.05, font=\footnotesize] at (u) {$u$};
\node[vertex] (v) at (3.75,4) {};
\node[right=0.05, font=\footnotesize] at (v) {$v$};
\node[vertex] (w0) at (0.75,2.5) {};
\node[vertex] (w1) at (4.5,2.5) {};
\coordinate (w2) at (2.25,2.25) {};
\coordinate (w3) at (3,2.25) {};
\coordinate (w4) at (0.5,1.25) {};
\coordinate (w5) at (2,1.5) {};
\coordinate (w6) at (3.5,1) {};
\coordinate (w7) at (4.75,0.75) {};

\draw[selectedEdge] (u) -- node[font=\footnotesize, color=black,, edgenode, pos=0.5] {$1$} (v);
\draw[selectedEdge] (u) -- node[font=\footnotesize, color=black, edgenode, pos=0.5] {$1$} (w0);
\draw[selectedEdge] (v) -- node[font=\footnotesize, color=black, edgenode, pos=0.5] {$1$} (w1);
\draw[discardedEdge, bend angle=35, bend left] (u) to (v);
\draw[discardedEdge] (u) -- (w2);
\draw[discardedEdge] (v) -- (w3);
\draw[discardedEdge] (w0) -- (w4);
\draw[discardedEdge] (w0) -- (w5);
\draw[discardedEdge] (w1) -- (w6);
\draw[discardedEdge] (w1) -- (w7);

\node[vertex] (1u) at (7.25,4.5) {};
\node[left=0.05, font=\footnotesize] at (1u) {$u$};
\node[vertex] (1v) at (9.5,4.5) {};
\node[right=0.05, font=\footnotesize] at (1v) {$v$};
\node[vertex] (1w0) at (6.55,3) {};
\node[vertex] (1w1) at (10.25,3) {};
\draw[selectedEdge] (1u) -- node[font=\footnotesize, color=\colorofy, edgenode, pos=0.5] {$\sfrac{1}{3}$} (1v);
\draw[selectedEdge] (1u) -- node[font=\footnotesize, color=\colorofy, edgenode, pos=0.5] {$\sfrac{1}{2}$} (1w0);
\draw[selectedEdge] (1v) -- node[font=\footnotesize, color=\colorofy, edgenode, pos=0.5] {$\sfrac{1}{2}$} (1w1);

\node[vertex] (2u) at (7.25,2) {};
\node[left=0.05, font=\footnotesize] at (2u) {$u$};
\node[vertex] (2v) at (9.5,2) {};
\node[right=0.05, font=\footnotesize] at (2v) {$v$};
\node[vertex] (2w0) at (6.55,0.5) {};
\node[vertex] (2w1) at (10.25,0.5) {};
\draw[selectedEdge] (2u) -- node[font=\footnotesize, color=\colorofy, edgenode, pos=0.5] {$\sfrac{1}{2}$} (2v);
\draw[selectedEdge] (2u) -- node[font=\footnotesize, color=\colorofy, edgenode, pos=0.5] {$\sfrac{1}{2}$} (2w0);
\draw[selectedEdge] (2v) -- node[font=\footnotesize, color=\colorofy, edgenode, pos=0.5] {$\sfrac{1}{2}$} (2w1);

\draw[arc] (5,2.75) -- (6.25,3.75);
\draw[arc] (5,2.25) -- (6.25,1.25);

\draw[separator] (5.625,2.5) -- (10.75,2.5);

\end{scope}

\end{tikzpicture}  	\caption{It can happen in two different ways that the edge $e = \{u,v\}$ ends up in a component consisting of $e$, one other edge incident to $u$ (but not $v$), and one other edge incident to $v$ (but not $u$), where the respective (conditioned) Poisson random variables all realized to $1$.
		For both cases, the realizations of the (conditioned) Poisson variables are shown on the left, while the values set by the schemes from Algorithm~\ref{algo:GenMatchCRS} and Algorithm~\ref{algo:GenMatchCRSmixed} are given on the top and bottom right, respectively. 
		Since the component containing $e$ is not bipartite in Case~1, there is no difference between the two schemes.
		In Case~2, however, Algorithm~\ref{algo:GenMatchCRSmixed} assigns a value of $\frac{1}{2}$ instead of $\frac{1}{3}$ to $e$.
	}\label{fig:genCRSloss}
\end{figure}

Besides the example described above, there are many more settings in which an edge ends up in a bipartite connected component and gains, to different extents, from applying the bipartite scheme to this component.
A sharp analysis of our new procedure would require to consider all these different cases, and therefore seems very difficult.
However, as mentioned in the introduction already, the main motivation for studying this scheme is showing that one can improve on the scheme from Algorithm~\ref{algo:GenMatchCRS}.
To do this, an exact analysis is not needed; in fact, it is actually enough to only account for the case described above %
 (Case~2 in Figure~\ref{fig:genCRSloss}). 
Moreover, we only consider the balancedness for $b=1$.
This enables us to show  that Algorithm~\ref{algo:GenMatchCRSmixed} yields a monotone $(1, \gamma(1) + 0.0003)$-balanced CR scheme for the matching polytope $\matchP[G]$, thereby proving non-optimality of the scheme given in Algorithm~\ref{algo:GenMatchCRS}.\footnote{%
	Since the analysis of the balancedness of the scheme from Algorithm~\ref{algo:GenMatchCRS} was tight, getting a monotone $(1, \gamma(1) + \epsilon)$-balanced CR scheme for some $\epsilon > 0$ really does constitute an improvement.%
}

\begin{proof}[Proof of Theorem~\ref{thm:mCRs_existence_general_improved}]
Given a graph $G=(V,E)$, we consider the procedure defined by Algorithm~\ref{algo:GenMatchCRSmixed}.
From Lemma~\ref{lem:mCRs_general_combined_noworse}, we already know that it is a monotone CR scheme for the general matching polytope.
Hence, it only remains to consider the balancedness. %

	To do this, we fix a point $x \in \matchP[G]$ and an edge $e = \{u,v\} \in \supp (x)$. 
	We would like to show that conditioned on $e$ appearing in the independently rounded set $R(x)$, the following happens with constant probability (see Case~2 in Figure~\ref{fig:genCRSloss}): 
	\begin{enumerate}[label=\normalfont($C_\arabic*$), leftmargin=4em]
		\item\label{item:connectCompSurvive} $e=\{u,v\}$ survives the subsampling step, %
		\item\label{item:connectCompPath} it ends up in a connected component that is a path of length~$3$ with $e$ as middle edge,
		\item\label{item:connectCompPois} the realizations of the Poisson random variables of the edges in this path all equal~$1$. 
	\end{enumerate}
	The last condition is important since the gain from using the bipartite scheme could become smaller when the Poisson random variables take on different values.
	
	For an arbitrary $x \in \matchP[G]$, however, the above does not happen with constant probability.
	The good news is that it can only fail for $x \in \matchP[G]$ for which the balancedness of the edge $e$ is higher than $\gamma(1)$ already, namely in the following cases.
	 If the $x$-load $x_u \coloneqq x(\delta(u))$ or $x_v \coloneqq x(\delta(v))$ on one of the endpoints of $e$ is bounded away from $1$, then~\eqref{eq:GenMatchCRSBalancSharp} shows that the balancedness of $e$ is strictly larger than $\gamma (1) = \expval \big[ \frac{1}{1+\PoisD{2}}  \big]$. The same is also true if the $x$-load $x_{uv} \coloneqq x(E_{uv})$ between the endpoints of $e$ is bounded away from $0$, where $E_{uv}$ again denotes the set of all edges in $E$ that go from $u$ to $v$ (including $e$).
	This happens, for example, if $x_e$ itself is bounded away from $0$.
		More precisely, if $x(\delta(u) \setminus E_{uv}) < 0.99$, $x(\delta(u) \setminus E_{uv}) < 0.99$, or $x(E_{uv}) > 0.01$, it follows from~\eqref{eq:mixedCRSdomination} and~\eqref{eq:GenMatchCRSBalancSharp} that
	\begin{equation*}
		\expval[ (y_x^{R(x)})_e  ]  \geq  x_e \cdot   \expval \bigg[   \frac{ 1  }{ 1+    \PoisD{1.99}  } \bigg]
		= x_e \cdot \frac{1 - \myexp^{-1.99}}{1.99}
		\geq x_e \cdot 0.4338
		\geq x_e \cdot (\gamma(1) + 0.0003)
		\enspace .
	\end{equation*}
	Hence, we can without loss of generality assume that $x(\delta(u) \setminus E_{uv}) \geq 0.99$, $x(\delta(u) \setminus E_{uv}) \geq 0.99$, and $x(E_{uv}) \leq 0.01$.

	For such an $x \in \matchP[G]$, we can show that conditioned on $e \in R(x)$,~\ref{item:connectCompSurvive}--\ref{item:connectCompPois} happens with probability at least $0.0018$. This is the statement of Lemma~\ref{lem:CRSmixed_condProb} in Appendix~\ref{app:CRSmixed_condProb}, where we also provide a formal proof.
	Thus, using the scheme for bipartite matchings in this case increases the balancedness of $e$ by at least
	\begin{equation}
		0.0018 \cdot \Big(\frac{1}{2} - \frac{1}{3} \Big) 
		= 0.0018 \cdot \frac{1}{6} 
		= 0.0003
		\enspace .
	\end{equation}
	This means that the overall balancedness of $e$ is at least $\gamma(1) + 0.0003$, as desired.
\end{proof}

As was done with Algorithm~\ref{algo:BipMatchCRS} and Algorithm~\ref{algo:GenMatchCRS}, we can also combine the independent rounding step with Algorithm~\ref{algo:GenMatchCRSmixed}.
The resulting procedure, which is given in Algorithm~\ref{algo:GenMatchCRSmixedshort}, only takes a point $x \in \matchP[G]$ and returns a random vector $y_x \in \matchP[G]$ that has the same distribution as the random vector $y_x^{R(x)}$ returned by Algorithm~\ref{algo:GenMatchCRS} when applied to $R(x)$.
One can also interpret Algorithm~\ref{algo:GenMatchCRSmixedshort} as the combination of Algorithm~\ref{algo:BipMatchCRSshort} and Algorithm~\ref{algo:GenMatchCRSshort}.
Besides being useful for rounding points $x \in \matchP[G]$, Algorithm~\ref{algo:GenMatchCRSmixedshort} also helps to simplify the proof of Lemma~\ref{lem:CRSmixed_condProb} in Appendix~\ref{app:CRSmixed_condProb}.

\begin{algorithm2e}[ht]	
	\caption{Combining independent rounding with Algorithm~\ref{algo:GenMatchCRSmixed}}
	\label{algo:GenMatchCRSmixedshort}
	\KwIn{Point $x \in \matchP[G]$.}
	\KwOut{Point $y_x \in \matchP$ with $\supp(y_x) \subseteq \supp(x)$.}
	For each edge $e \in E$, let $q_e \in \mathbb{Z}_{\geq 0}$ be an independent realization of a $\PoisD{x_e}$-random variable\;
	Let $E_b$ be the set of edges in bipartite components of $ G' =(V,\supp (q) )$\;
	For each edge $e = \{u,v\} \in E \setminus E_b$, let 
	\smash{$(y_x)_e \coloneqq \frac{ q_e } {\sum_{g \in \delta(u) \cup \delta(v)} q_g   } $}\;\vspace{0.15cm}
	For each edge $e = \{u,v\} \in E_b$, let 
	\smash{$(y_x)_e \coloneqq \frac{ q_e } { \max\{  \sum_{g \in \delta(u)} q_g , \sum_{g \in \delta(v)} q_g  \}  } $}\;
	Return $y_x$\;
\end{algorithm2e}

\section{Conclusion}

In this work, we introduced a novel, very general technique for the construction and analysis of contention resolution schemes. 
The main idea is to take a different viewpoint and consider the marginals rather than the actual sets returned by the scheme, which allows for using polyhedral techniques.
We demonstrated the usefulness of this technique by presenting improved monotone contention resolution schemes for the bipartite and the general matching polytope.
While our technique allowed for giving simple analyses of our schemes, it moreover enabled us to prove that our scheme for bipartite matchings is optimal among all monotone such schemes.
For general matchings, we first presented a new monotone contention resolution scheme %
whose balancedness significantly improves on prior work and matches the previously best known lower bound on the correlation gap of the matching polytope. While the existence of a (not necessarily monotone) contention resolution scheme attaining this balancedness already followed from the lower bound on the correlation gap, it was not known before whether the same can also be achieved by a monotone scheme.
Moreover, we could show that combining our schemes for bipartite and general matchings yields a monotone contention resolution scheme for the general matching polytope with a strictly higher balancedness than that of the scheme for general matchings alone. %
However, we do not know the exact balancedness that our combined scheme attains.
More generally, it is still open what balancedness can be achieved by a monotone contention resolution scheme for general matchings.
Both our results on monotone contention resolution schemes for bipartite and general matchings also improved the best known lower bounds on the correlation gap of the bipartite and the general matching polytope, respectively.

\appendix

\section{Appendix}

\subsection{Contention resolution for points with small components}\label{app:CRscheme_smallx}

In the following, we prove a generalized version of Lemma~\ref{lem:CRscheme_smallx_bipmatch} which shows that for many constraint types, the existence of a monotone $c$-balanced contention resolution scheme follows from the existence of such a scheme that is guaranteed to achieve this balancedness for points with small components only.
The idea is that whenever we have an element with a large $x$-value, we ``split'' this element into so-called ``siblings'' that have the same properties with respect to feasibility and distribute the $x$-value uniformly among them. We then apply the CR scheme for this new instance and map the set that we get back to a feasible set of the original instance. Since an element that got split can only be taken once, we have to ensure that the CR scheme selects at most one of the siblings at a time.
In the following, we formalize this idea.

\begin{definition}\label{def:siblings}
	Consider an independence system $(E, \mathcal{F})$.
	We call distinct elements $e_1, \ldots, e_k \in E$ \emph{siblings} (or \emph{siblings of each other}) if
	\begin{enumerate}[label=\normalfont(\roman*), leftmargin=4em]
		\item for any subset $S \subseteq E \setminus \{e_1, \ldots, e_k\}$ and $i,j \in [k]$, it holds that $S \cup \{e_i\} \in \mathcal{F}$ if and only if $S \cup \{e_j\}\in \mathcal{F}$, and
		\item every subset $S \in \mathcal{F}$ contains at most one of the elements $e_1, \ldots, e_k$. 	
	\end{enumerate}
\end{definition}

As an example, let $G=(V,E)$ be a graph and let $\match[G] \subseteq 2^E$ be the family of all matchings in $G$. If we have more than one edge between a given pair of vertices, then these edges are easily seen to be siblings according to the above definition.

\begin{definition}\label{def:splitting}
	Consider an independence system $(E, \mathcal{F})$, an element $e \in E$, and $k \in \mathbb{Z}_{> 0}$.
	We define the independence system $(E', \mathcal{F}')$ that is obtained by \emph{splitting $e$ into $k$ siblings} by
	\begin{equation*}
		E' \coloneqq E \dcup \{e_1, \ldots, e_{k-1}\}
		\quad \textup{and} \quad
		\mathcal{F}' \coloneqq  \mathcal{F}
							\cup  \{  (S\setminus\{e\}) \cup \{e_i\} \mid e \in S \in \mathcal{F}, i \in [k-1]  \}
				\enspace .
	\end{equation*}
	This means that $k-1$ new elements $e_1, \ldots, e_{k-1}$ are added to $E$ such that $e, e_1, \ldots, e_{k-1}$ are siblings.
\end{definition}

In the example mentioned above, where we have a graph $G=(V,E)$ and the family of all its matchings $\match[G]$, splitting an edge $e \in E$ into $k$ siblings is exactly the same as simply adding $k-1$ new edges that are parallel to $e$. Also, the new family %
 that we get from $\match[G]$ is simply the family of all matchings in the new graph. 
Hence, the considered constraint type did not change. 
This is important for our reduction to work. 

\begin{definition}\label{def:familiesClosedSplitting}
	Consider a family of independence systems $\big( (E_i, \mathcal{F}_i ) \big)_{i \in I}$.
	We say that such a family is \emph{closed under splitting into siblings} if for any independence system $(E, \mathcal{F})$ it contains, any element $e \in E$, and any $k \in \mathbb{Z}_{>0}$, 
	the independence system $(E',\mathcal{F}')$ that is obtained by splitting $e$ into $k$ siblings is also a member of $\big( (E_i, \mathcal{F}_i ) \big)_{i \in I}$.\footnote{%
		Formally, it is enough if the family contains another independence system $(E_j, \mathcal{F}_j)$ such that there is a bijection $\rho \colon E' \to E_j$ which preserves feasibility.
}
\end{definition}

Some examples of families of independence systems that are easily seen to be closed under splitting into siblings are:
\begin{itemize}
	\item The family of all graphic matroids.
	\item The family of all linear matroids.
	\item The family of all (general) matroids.
	\item The family of all bipartite graphs $G = (V, E)$, where the feasible subsets of $E$ are the matchings in $G$.
	\item The family of all (general) graphs $G = (V, E)$, where the feasible subsets of $E$ are the matchings in $G$.
\end{itemize}

For such families of independence systems, the following lemma shows that it is enough to have a monotone contention resolution scheme that is guaranteed to be $c$-balanced for vectors with small components only.
\begin{lemma}\label{lem:CRscheme_smallx_general}
	Consider a family of independence systems $\big( (E_i, \mathcal{F}_i ) \big)_{i \in I}$ that is closed under splitting into siblings.
	Moreover, let $b,c \in [0,1]$ and $\epsilon > 0$.
	If there is a monotone contention resolution scheme for all combinatorial polytopes $P_{\mathcal{F}_i}$, $i \in I$, that is $c$-balanced for all $x \in b P_{\mathcal{F}_i}$ with $x \leq \epsilon \cdot \chi^{E_i}$, then there exists a monotone $(b,c)$-balanced contention resolution scheme for all combinatorial polytopes $P_{\mathcal{F}_j}$, $j \in I$ (without any restriction).
\end{lemma}

As observed above, the family of all bipartite graphs $G = (V, E)$, where the feasible subsets of $E$ are the matchings in $G$, is closed under splitting into siblings. Therefore, Lemma~\ref{lem:CRscheme_smallx_bipmatch} is an immediate consequence of Lemma~\ref{lem:CRscheme_smallx_general}.

Instead of proving Lemma~\ref{lem:CRscheme_smallx_general} directly, we show the statement below. %
By iteratively applying it, Lemma~\ref{lem:CRscheme_smallx_general} follows.

\begin{lemma}\label{lem:CRscheme_smallx_oneElement}
	Let $b,c \in [0,1]$ and consider an independence system $(E, \mathcal{F})$. For $e \in E$ and $k \in \mathbb{Z}_{> 0}$, let $(E', \mathcal{F}')$ be the independence system obtained by splitting $e$ into $k$ siblings $e, e_1, \ldots, e_{k-1}$ (see Definition~\ref{def:splitting}).
	Moreover, fix $x \in b P_{\mathcal{F}}$ and let $x' \in b P_{\mathcal{F}'}$ be the point we get from $x$ by uniformly distributing $x_e$ onto the elements $e, e_1,\ldots, e_{k-1}$ , i.e., 
		\begin{equation*}
		x'_g \coloneqq \begin{cases}
			x_g & \textup{if } g \notin \{e,e_1, \ldots, e_{k-1} \}  \enspace ,\\
			\frac{x_e}{k} & \textup{if } g \in \{e,e_1, \ldots, e_{k-1} \}  \enspace ,
		\end{cases} \qquad \forall g\in E'\enspace.
	\end{equation*}	
	If there is a monotone contention resolution scheme for $P_{\mathcal{F}'}$ %
	that is $c$-balanced for $x'$, then there exists a monotone contention resolution scheme for $P_{\mathcal{F}}$ %
	 that is $c$-balanced for $x$.
\end{lemma}

While the proof below is constructive, the construction is typically not efficient. This is because we may split an element into many smaller ones, potentially resulting in an exponential increase in the instance size. 
Therefore, we cannot directly use this construction to get efficient CR schemes, meaning that the statements we prove here are only existential. Nevertheless, as we have seen in the previous sections on (bipartite) matchings, Lemma~\ref{lem:CRscheme_smallx_general} tells us which instances are inherently hard and thereby helps to find good schemes. Moreover, the CR schemes we provided are heavily motivated by this non-efficient reduction and realize the splitting idea in an efficient way.

\begin{proof}[Proof of Lemma~\ref{lem:CRscheme_smallx_oneElement}]

Let $\pi'$ be a monotone CR scheme for $P_{\mathcal{F}'}$ %
 that is $c$-balanced for $x'$.
We now have to show that there exists a  monotone CR scheme $\pi$ for $P_{\mathcal{F}}$ %
 that is $c$-balanced for $x$.

For $x \in b P_{\mathcal{F}}$, we want to define $\pi_x$ via $\pi'_{x'}$ such that the $c$-balancedness of $\pi'_{x'}$ passes on to $\pi_x$.	
Note that to somehow use the $c$-balancedness property of $\pi'_{x'}$, we need to feed a set into $\pi'_{x'}$ that has the same distribution as $R(x') \subseteq E'$, i.e., the set we get by independently rounding $x' \in bP_{\mathcal{F}'}$.
However, because our goal is to define $\pi_x$ such that it is $c$-balanced, we start with the random set $R(x) \subseteq E$ we get by independently rounding $x \in b P_{\mathcal{F}}$.
In a first step, we transform the set $R(x)$ into a set $R'(x) \subseteq E'$ that has the same distribution as $R(x')$. We can then apply $\pi'_{x'}$ and, as we see later, map $\mathcal{F}'$-feasible subsets of $R'(x)$ back to $\mathcal{F}$-feasible subsets of $R(x)$ in a straightforward way.

We now formally describe how this transformation works for a general set. For $A \subseteq \supp (x) \subseteq E$, we define the random set $A' \subseteq \supp (x') \subseteq E'$ %
by
\begin{equation*}
		A' \coloneqq \begin{cases}
		A & \textup{if } e \notin A  \enspace ,\\
			(A\setminus\{e\}) \cup D & \textup{if } e \in A  \enspace ,
		\end{cases}
	\end{equation*}
where $D \subseteq \{e,e_1, \ldots, e_{k-1}\}$ is a random set such that $\Pr [ D = J  ] = \frac{1}{x_e} \cdot \big( \frac{x_e}{k} \big)^{j} \cdot \big( 1- \frac{x_e}{k}\big)^{k-j}$ for every non-empty subset $J \subseteq \{e, e_1, \ldots, e_{k-1} \} $ of cardinality $|J|=j$.
As one can easily check, such a probability distribution indeed exists and it holds for every non-empty subset $J \subseteq \{e,e_1, \ldots, e_{k-1} \}$ of cardinality $|J| = j$ that
$\Pr [ J \subseteq D ] = \frac{1}{x_e} \cdot \big( \frac{x_e}{k} \big)^{j} $.
As we will see below, applying this transformation to the random set $R(x) \subseteq \supp(x)$ yields a set $R'(x) \subseteq E'$ that has the same distribution as $R(x') \subseteq \supp (x')$, as desired.

For $A \subseteq \supp (x) \subseteq E$, we then simply let the new scheme return the feasible set 
$	\pi_x (A) \coloneqq \pi'_{x'} (A') $.
Observe that the set $\pi'_{x'} (A')$ might not be a subset of $E$. However, since it contains at most one of the elements $e,e_1, \ldots, e_{k-1}$, we can just replace that element with $e$ whenever this happens. This way, we can also interpret $\pi'_{x'} (A')$ as a subset of $E$.
Moreover, note that monotonicity of $\pi'_{x'}$ ensures that $\pi_x$ is monotone as well.

For the balancedness, we consider the random set $R'(x) \subseteq E'$ we get by independently rounding $x$ and then applying the above described transformation to the resulting set $R(x) \subseteq E$. 
For any element $g \in E'\setminus \{e,e_1, \ldots, e_{k-1}\}$, we have 
\begin{equation*}
\Pr [ g \in R'(x) ] = \Pr [ g \in R(x) ] = x_g = x'_g
\enspace .
\end{equation*}
Moreover, for $e' \in \{e, e_1, \ldots, e_{k-1}\}$, it holds that
\begin{equation*}
	\Pr [ e' \in R'(x) ] = \Pr [ e \in R(x) ] \cdot \Pr[ e' \in D ] = x_e \cdot \frac{1}{x_e} \cdot \frac{x_e}{k} = \frac{x_e}{k} =
	 x'_{e'}
	\enspace .
\end{equation*}
Since one can easily verify that the distribution of $D$ ensures mutual independence of the above events, the distribution of $R'(x) \subseteq E'$ is indeed exactly the same as the one of the set $R(x')$ which we get by independently rounding the point $x'$.
Consequently, the balancedness of all elements $g \in E \setminus \{e\}$ remains the same, i.e., we have $\Pr[g \in \pi_x (R(x))] \geq c \cdot x_g$.
	Moreover, for the element $e$, we obtain
	\begin{align*}
		\Pr [ e \in \pi_x (R(x))]
		&= \Pr [ e' \in \pi'_{x'} (R'(x)) \textup{ for any } e' \in \{e, e_1, \ldots, e_{k-1}\} ]  \\
		&= \Pr [ e' \in \pi'_{x'} (R(x')) \textup{ for any } e' \in \{e, e_1, \ldots, e_{k-1}\} ]
		\enspace .
	\end{align*}
	For ease of notation, we now denote the element $e \in E'$ by $e_0$. Using that $\pi'_{x'}$ is $c$-balanced, we then get for the above that
	\begin{align*}
		\Pr [ e \in \pi_x (R(x))] 
		&= \Pr [ e_i \in \pi'_{x'} (R(x')) \textup{ for any } i \in \{0,1,\ldots, k-1\} ] \\
		&=  \sum_{i=0}^{k-1} \Pr [ e_i \in \pi'_{x'} (R(x'))]  
		\geq \sum_{i=0}^{k-1} c \cdot x'_{e_i}    
		= \sum_{i=0}^{k-1} c \cdot \frac{x_e}{k}    
		= c \cdot x_e
		\enspace ,
	\end{align*}
	where we again used that no feasible set in $\mathcal{F}'$ contains more than one of the elements $e=e_0, e_1, \ldots, e_{k-1}$.
	This shows the desired balancedness of $\pi_x$, and thereby finishes the proof of Lemma~\ref{lem:CRscheme_smallx_oneElement}.	
\end{proof}

Note that the fractional points $x$ we have to round in the context of~\ref{eq:CSFM} are typically rational (they are the output of some preceding computation). If we restrict ourselves to rational points, then it is easy to see that Lemma~\ref{lem:CRscheme_smallx_oneElement} actually implies a stronger statement than Lemma~\ref{lem:CRscheme_smallx_general}. Namely, it shows that for a given family of independence systems that is closed under splitting into siblings, the existence of a monotone CR scheme that is guaranteed to be $c$-balanced for rational points actually follows from the existence of a monotone CR scheme that is guaranteed to be $c$-balanced only for rational points whose non-zero components are all equal to a \emph{single} small value.
\begin{corollary}\label{cor:CRscheme_smallEqualx_general}
	Consider a family of independence systems $\big( (E_i, \mathcal{F}_i ) \big)_{i \in I}$ that is closed under splitting into siblings.
	Moreover, let $b,c \in [0,1]$ and $\epsilon > 0$.
	If there is a monotone contention resolution scheme for all combinatorial polytopes $P_{\mathcal{F}_i}$, $i \in I$, that is $c$-balanced for all rational $x \in b P_{\mathcal{F}_i}$ such that there exists $N \in \mathbb{Z}_{> 0}$ with $ x_e = \frac{1}{N} \leq \epsilon$ for all $e \in \supp(x)$, then there exists a monotone contention resolution scheme for all combinatorial polytopes $P_{\mathcal{F}_j}$, $j \in I$, that is $c$-balanced for all rational $x \in b P_{\mathcal{F}_j}$ (without any further restriction).
\end{corollary}

\begin{proof}
	As Lemma~\ref{lem:CRscheme_smallx_general}, we also obtain Corollary~\ref{cor:CRscheme_smallEqualx_general} by iteratively applying Lemma~\ref{lem:CRscheme_smallx_oneElement} for carefully chosen numbers $k$.
	More precisely, for an independence system $(E, \mathcal{F})$ of the given family and a rational $x \in b P_{\mathcal{F}}$, there exists $N \in \mathbb{Z}_{> 0}$ such that for every $e \in E$, we have $x_e = \frac{k_e}{N}$ for some $k_e \in \mathbb{Z}_{\geq 0}$. Moreover, we can choose $N$ such that $\frac{1}{N} \leq \epsilon$.
	With the help of Lemma~\ref{lem:CRscheme_smallx_oneElement}, we then iteratively split each $e\in E$ into $k_e$ siblings.
\end{proof}

\subsection{Stochastic dominance}\label{app:stochDom}

At the end of the proof of Theorem~\ref{thm:mCRs_existence_bipartite}, a claim regarding the stochastic dominance of different distributions was made. Here, we formally state and prove it.
\begin{lemma}\label{lem:stochDom}
	Let $P,Q,X,Y,Z$ be mutually independent random variables such that
	\begin{itemize}
		\item $P,Q,X,Y,Z$ only take values in $\mathbb{Z}_{\geq 0}$, and
		\item $X,Y,Z$ all have the same distribution.
	\end{itemize} 
	It then holds that the distribution of $X + \max\{P,Q\}$ is stochastically dominated by the distribution of $\max\{ Y+P, Z+Q \}$.
\end{lemma}

\begin{proof}
	In order to prove Lemma~\ref{lem:stochDom}, we have to show that for any $k \in \mathbb{Z}_{\geq 0}$, it holds that
	\begin{equation}\label{eq:stochasticDominance}
	\Pr[ \max\{Y + P, Z + Q\} \geq k] \geq \Pr[ X + \max\{ P,  Q\} \geq k]
	\enspace .
	\end{equation}
	We do this by proving that for all $p,q \in \mathbb{Z}_{\geq 0}$, we have
	\begin{equation}\label{eq:stochasticDominanceConditioned}
	\Pr[ \max\{Y + P, Z + Q\} \geq k \mid P=p, Q=q] 
	\geq \Pr[ X + \max\{ P,  Q\} \geq k \mid P=p, Q=q ]  
	\enspace ,
	\end{equation}
	which is even stronger than~\eqref{eq:stochasticDominance}.
	
	To show~\eqref{eq:stochasticDominanceConditioned}, fix $p,q \in \mathbb{Z}_{\geq 0}$. By symmetry, we can without loss of generality assume that $p \geq q$.
	Moreover, using that the random variables $P,Q,X,Y,Z$ are mutually independent and that $X,Y,Z$ all have the same distribution, proving~\eqref{eq:stochasticDominanceConditioned} becomes trivial by observing that
	\begin{align*}
	\Pr[ X + \max\{ P,  Q\} \geq k \mid P=p, Q=q ] 
	&=\Pr[ X + \max\{ p, q\} \geq k ] \\
	&= \Pr[ X +  p \geq k ] \\
	& =  \Pr[ Y +  p \geq k ] \\
	& \leq \Pr[ \max\{Y + p, Z + q\} \geq k] \\
	& = \Pr[ \max\{Y + P, Z + Q\} \geq k \mid P=p, Q=q]
	\enspace .
	\end{align*}
	This shows~\eqref{eq:stochasticDominanceConditioned}, and therefore proves Lemma~\ref{lem:stochDom}.
\end{proof}

\subsection{Probability of ending up in a good connected component}\label{app:CRSmixed_condProb}

Let $G=(V,E)$ be a graph, and fix a point $x \in \matchP[G]$ and an edge $e = \{u,v\} \in \supp (x)$.
In the proof of Theorem~\ref{thm:mCRs_existence_general_improved}, it was claimed that under certain conditions, with constant probability the edge $e$ ends up in a specific situation that enables the scheme from Algorithm~\ref{algo:GenMatchCRSmixed} to assign a larger value to $e$ than the scheme from Algorithm~\ref{algo:GenMatchCRS} does.
The formal statement is given in the following lemma.
\begin{lemma}\label{lem:CRSmixed_condProb}
	Let $x \in \matchP[G]$ and $e = \{u,v\} \in \supp (x)$ such that 
	$x(\delta(u) \setminus E_{uv}) \geq 0.99$, $x(\delta(u) \setminus E_{uv}) \geq 0.99$, and $x(E_{uv}) \leq 0.01$, where $E_{uv}$ denotes the set of all edges in $E$ that go from $u$ to $v$ (including $e$ itself).
	Let $C$ be the event that the following happens when independently rounding $x$ and applying Algorithm~\ref{algo:GenMatchCRSmixed} to the resulting set $R(x)$:
	\begin{enumerate}[label=\normalfont($C_\arabic*$), leftmargin=4em]
		\item $e=\{u,v\}$ survives the subsampling step, %
		\item\label{item:connectCompP3} it ends up in a connected component that is a path of length~$3$ with $e$ as middle edge,
		\item the realizations of the Poisson random variables of the edges in this path all equal~$1$. 
	\end{enumerate}
	It then holds that
	\begin{equation*}
	\Pr [ C | e \in R(x) ] \geq 0.0018
	\enspace .
	\end{equation*}
\end{lemma}

\begin{proof}
	Note that we cannot replace~\ref{item:connectCompP3} by the condition that $e$ ends up in a connected component consisting of $e$, one other edge $g$ incident to $u$ (but not $v$), and one other edge $h$ incident to $v$ (but not $u$).
	The reason why this does not work is that we could end up with a connected component that is a triangle (see Figure~\ref{fig:genCRSloss}).
	Since the possibility of getting triangles complicates the analysis of $\Pr [C | e \in R(x)]$, we try to avoid triangles by proving something stronger.
	
	To do this, we first observe that we can partition the vertices in $V \setminus \{u,v\}$ into two sets $V_u$ and $V_v$ such that $x(E_{u,V_u}) \geq 0.33$ and $x(E_{v, V_u}) \geq 0.33$,
	where for two disjoint sets $S_1, S_2 \subseteq V$, we denote the set of all edges in $E$ having one endpoint in $S_1$ and the other in $S_2$ by $E_{S_1, S_2}$ (set braces are left out for singletons).
	In fact, one can obtain such a partition by starting with $V_u = V_v = \emptyset$ and greedily assigning vertices $w \in V \setminus \{u,v\}$ to $V_u$ or $V_v$, depending on whether the total $x$-load between $w$ and $u$ or $w$ and $v$ is higher.
	As soon as we have assigned enough vertices for either $x(E_{u,V_u}) \geq 0.33$ or $x(E_{v,V_v}) \geq 0.33$ to hold, we assign all remaining vertices to the other set.
	When considering the vertices $w \in V \setminus \{u,v\}$ in decreasing order with respect to $x(E_{w,\{u,v\}})$, then this is easily seen to yield a partition with the desired properties.

	\begin{figure}[ht]
		\centering

\tikzstyle{vertex}=[circle, draw=black, thick, inner sep=1pt, minimum size=2mm, fill=black!10]
\tikzstyle{edgenode}=[rectangle, fill=white, fill opacity=0.8, text opacity=1, inner sep=1pt]
\tikzstyle{edge}=[thick, draw=black]
\tikzstyle{selectedEdge}=[ultra thick, draw=\colorofset]
\tikzstyle{discardedEdge}=[semithick, dashed, draw=black]
\tikzstyle{arc}=[thick, draw=black!60, -latex]
\tikzstyle{boxed}=[thick, draw=black!60, rounded corners]
\tikzstyle{boxedFat}=[thick, draw=black, rounded corners]
\tikzstyle{boxedLight}=[thick, draw=black, rounded corners, fill = black!15, minimum height=0.75*1cm]

\begin{tikzpicture}[scale=0.75]

\begin{scope}[]

\draw [black!80, thick, fill=black!20] plot [smooth cycle] coordinates {(1.125,3) (0.25,2.5) (0.25,0.5) (2.25,0.75)  (2.5,3)};
\draw [black!80, thick, fill=black!20] plot [smooth cycle] coordinates {(4.125,3) (3.25,2.5) (3.25,0.5) (5.25,0.375) (5.5,1.5) (5.125,2.75)};
\node[black!80, font=\footnotesize] at (1,0.75) {$V_u$};
\node[black!80, font=\footnotesize] at (4,0.75) {$V_v$};

\node[vertex] (u) at (1.5,4) {};
\node[left=0.05, font=\footnotesize] at (u) {$u$};
\node[vertex] (v) at (3.75,4) {};
\node[right=0.05, font=\footnotesize] at (v) {$v$};
\node[vertex] (w0) at (0.75,2.5) {};
\node[vertex] (w1) at (4.5,2.5) {};
\coordinate (w2) at (1.5,1.5) {};
\coordinate (w3) at (0.5,1.25) {};
\coordinate (w4) at (4,1.25) {};
\coordinate (w5) at (4.75,0.75) {};
\coordinate (w6) at (1.75,2.5) {};
\coordinate (w7) at (3.75,2) {};

\draw[selectedEdge] (u) -- node[font=\footnotesize, color=black,, edgenode, pos=0.5] {$1$} (v);
\draw[selectedEdge] (u) -- node[font=\footnotesize, color=black, edgenode, pos=0.4] {$1$} (w0);
\draw[selectedEdge] (v) -- node[font=\footnotesize, color=black, edgenode, pos=0.4] {$1$} (w1);
\draw[discardedEdge, bend angle=35, bend left] (u) to (v);
\draw[discardedEdge] (u) to (w1);
\draw[discardedEdge] (v) -- (w2);
\draw[discardedEdge] (w0) -- (w3);
\draw[discardedEdge] (w0) -- (w4);
\draw[discardedEdge] (w1) -- (w5);
\draw[discardedEdge] (u) -- (w6);
\draw[discardedEdge] (v) -- (w7);

\end{scope}

\end{tikzpicture}  		\caption{Above, a partition $V \setminus \{u,v\} = V_u \dcup V_v$ and a possible outcome of $\overline{R}(x)$, i.e., the set of edges we get after the independent rounding and the subsampling step, is shown.
		Only the connected component of $e =\{u,v\}$ in $G'=(V, \overline{R}(x))$ is considered, where dashed edges indicate edges in $E$ which are incident to this component but did not appear in  $\overline{R}(x)$.
		For the edges in $\overline{R}(x)$, the realizations of the (conditioned) Poisson random variables are given as well.
		In the case depicted above, the event~$D$ holds.
		}\label{fig:genCRSpartition}
	\end{figure}
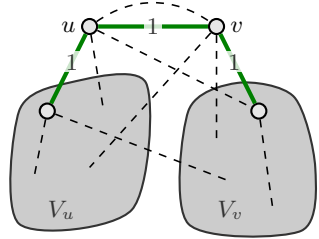

	Using such a partition $V\setminus \{u,v\} = V_u \dcup V_v$,~\ref{item:connectCompP3} is implied by $e$ ending up in a connected component consisting of $e$, one other edge $g$ going from $u$ to $V_u$, and one other edge $h$ going from $v$ to $V_v$. Using this stronger condition, we do not have to worry about triangles anymore.
	Instead of $C$, we therefore consider the event $D$ that the following happens when independently rounding $x$ and applying Algorithm~\ref{algo:GenMatchCRSmixed} to the resulting set $R(x)$ (see Figure~\ref{fig:genCRSpartition}):
	\begin{enumerate}[label=\normalfont($D_\arabic*$), leftmargin=4em]
		\item\label{item:connectComp1} $q_e = 1$, meaning that $e=\{u,v\}$ appears in $\overline{R}(x)$ and its Poisson random variable realizes to $1$,
		\item\label{item:connectComp2} $\sum_{g \in E_{u,v} \setminus \{e\} } q_g = 0$, meaning that apart from $e$, no other edge in $E_{uv}$ appears in $\overline{R}(x)$,
		\item\label{item:connectComp3} $\sum_{g \in E_{u,V_u} } q_g = 1$, meaning that exactly one edge $e_u = \{u,u'\}$ in $E_{u,V_u}$ appears in $\overline{R}(x)$ and its Poisson random variable realizes to $1$,
		\item\label{item:connectComp4} $\sum_{g \in E_{v,V_v} } q_g = 1$, meaning that exactly one edge $e_v = \{v,v'\}$ in $E_{v,V_v}$ appears in $\overline{R}(x)$ and its Poisson random variable realizes to $1$,
		\item\label{item:connectComp5} $\sum_{g \in E_{u,V_v} } q_g + \sum_{g \in E_{v,V_u} } q_g = 0$, meaning that no edge in $E_{u,V_v}$ and no edge in $E_{v, V_u}$ appears in $\overline{R}(x)$,
		\item\label{item:connectComp6} $\sum_{g \in \delta(u') \setminus \{e_u\} } q_g + \sum_{g \in \delta(v') \setminus \{e_v\} } q_g = 0$, meaning that apart from $e_u = \{u,u'\}$, no other edge in $\delta(u')$ appears in $\overline{R}(x)$, and apart from $e_v =\{v, v'\}$, no other edge in $\delta(v')$ appears in $\overline{R}(x)$.
	\end{enumerate}
	
	Clearly, the event $D$ implies the event $C$.
	Our goal is therefore to get a lower bound on the probability of the event $D$ conditioned on $e$ appearing in the independently rounded set $R(x)$.
	In order to do this, we directly use Algorithm~\ref{algo:GenMatchCRSmixedshort}, which combines the independent rounding step with Algorithm~\ref{algo:GenMatchCRSmixed}.
	Recall that in this setting, we can identify $\overline{R}(x)$ with $\supp (q)$.
	Moreover, we make the following useful observations.
	\begin{itemize}
		\item The event~\ref{item:connectComp1} is the only event that is affected by conditioning on $e \in R(x)$.
		\item The events~\ref{item:connectComp1}--\ref{item:connectComp5} are mutually independent.
		\item The event~\ref{item:connectComp6} depends on the edges which appear in~\ref{item:connectComp3} and~\ref{item:connectComp4}. It is possible to use a bound for the probability of this event which is independent of these particular edges, though.
		\item The event~\ref{item:connectComp6} is independent of~\ref{item:connectComp1} and~\ref{item:connectComp2}, and it is positively correlated with~\ref{item:connectComp3},~\ref{item:connectComp4}, and~\ref{item:connectComp5}. 
		\item For the events~\ref{item:connectComp2}--\ref{item:connectComp6}, the quantities we are interested in follow a Poisson distribution.\footnote{%
			Here, we use again that the sum of two independent Poisson random variables $\PoisD[1]{\xi}$ and $\PoisD[2]{\zeta}$ has a $\PoisD{\xi+\zeta}$-distribution.
			}
	\end{itemize}

	Using Algorithm~\ref{algo:GenMatchCRSmixedshort} and the observations from above, we obtain
	\begin{align*}
		\Pr [ D | e \in R(x) ]
		 & = \Pr[ D_1, \ldots, D_6 | e \in R(x) ] \\
		 & = \Pr[D_1 | e \in R(x) ] \cdot \Pr[D_2] \cdot \Pr[D_3] \cdot \Pr[D_4] \cdot \Pr[D_5] \cdot \Pr[D_6 | D_3, D_4, D_5] \\
		 & \geq \Pr[D_1 | e \in R(x) ] \cdot \Pr[D_2] \cdot \Pr[D_3] \cdot \Pr[D_4] \cdot \Pr[D_5] \cdot \Pr[D_6 ] \\
		 & \geq \frac{\Pr[\PoisD{x_e} = 1]}{x_e} \cdot \Pr[\PoisD{x(E_{u,v}) - x_e} = 0] \cdot \Pr[\PoisD{x(E_{u,V_u})} = 1]  \\
		 & \phantom{{}\geq{} } \cdot  \Pr[\PoisD{x(E_{v,V_v})} = 1]  \cdot \Pr[\PoisD{x(E_{u,V_v}) + x(E_{v,V_u}) } = 0]  \\
		 & \phantom{{}\geq{} }  \cdot \Pr[\PoisD{x(\delta(u')) + x(\delta(v'))} = 0] \\
		 & = \myexp^{- x_e} \cdot \myexp^{-x(E_{u,v})+x_e} \cdot \big( x(E_{u,V_u}) \cdot \myexp^{-x(E_{u,V_u})}  \big) \cdot \big( x(E_{v,V_v}) \cdot \myexp^{-x(E_{v,V_v})}  \big) \\
		 &  \phantom{{}={} }  \cdot  \myexp^{-x(E_{u,V_v}) - x(E_{v,V_u})}  \cdot  \myexp^{-x(\delta(u')) - x(\delta(v'))}  \\
		 & = x(E_{u,V_u}) \cdot x(E_{v,V_v}) \cdot  \myexp^{-x(\delta(u) \cup \delta(v) ) + x_e} \cdot \myexp^{-x(\delta(u')) - x(\delta(v'))}
		\enspace ,
	\end{align*} 
	where the second inequality follows from $\Pr[ \PoisD{\xi} = 0 ] = \myexp^{-\xi}$ being decreasing in $\xi \geq 0$.
	Since $x(E_{u,V_u}) \geq 0.33$, $x(E_{v,V_v}) \geq 0.33$, and $x(\delta(w)) \leq 1$ for any $w \in V$, we get from the above that
	\begin{align*}
		\Pr [ D | e \in R(x) ]
		& \geq x(E_{u,V_u}) \cdot x(E_{v,V_v}) \cdot  \myexp^{-x(\delta(u) \cup \delta(v) ) + x_e} \cdot \myexp^{-x(\delta(u')) - x(\delta(v'))} \\
		& \geq 0.33 \cdot 0.33 \cdot \myexp^{-2} \cdot \myexp^{-2} \\
		& \geq 0.0018
		\enspace .
	\end{align*} 
	Recalling that the event $D$ implies the event $C$, we conclude that
	\begin{equation*}
		\Pr [ C | e \in R(x) ] \geq \Pr [ D | e \in R(x) ] \geq 0.0018
		\enspace ,
	\end{equation*}
	as desired.	
\end{proof}

\bibliographystyle{plain}

\end{document}